\documentclass[a4paper,UKenglish,cleveref, autoref, thm-restate]{oasics-v2021}

\usepackage[utf8]{inputenc}
\usepackage[T1]{fontenc}

\usepackage[inline]{enumitem}
\usepackage[all]{foreign}
\usepackage{amsmath,amssymb,amsthm}
\usepackage{thm-restate}
\usepackage{amsmath}
\usepackage{lineno,hyperref}
\usepackage[capitalise]{cleveref}
\usepackage{xspace}
\usepackage{thmtools}
\usepackage{mathtools}
\usepackage{autonum}
\usepackage[all]{foreign}
\usepackage{xparse}
\usepackage{xcolor}
\usepackage{forest}
\usepackage{makecell}
\usepackage{multicol}
\usepackage{rotating}
\usepackage{marginnote}
\usepackage{tikz}
\usepackage{algorithm}
\usepackage{algpseudocode}
\usepackage{todonotes}
\usepackage{booktabs}
\usepackage{lipsum}
\usepackage{packages/commands} 
\usepackage{packages/graphics} 
\usepackage{packages/logic}    
\usepackage{packages/symbols}  

\usepackage{tikz}
\usetikzlibrary{positioning}
\usetikzlibrary{backgrounds}
\usetikzlibrary{fit}
\usetikzlibrary{automata}
\usetikzlibrary{shapes.geometric}
\usetikzlibrary{decorations.pathreplacing,calligraphy}

\newcommand{\roberto}[1]{#1}

 \usepackage[table,xcdraw]{xcolor}
 \usepackage{adjustbox}
\usepackage{graphicx}
\usepackage{makecell}

\newif\ifextended
\extendedtrue

\usepackage{etoolbox}

\title{The $\Sigma$-Chain Product: A Succinct Model of Automata
(De)Composition \ifextended (Extended Version) \fi}

\author{Roberto Borelli}{University of Padua, Italy}{roberto.borelli@studenti.unipd.it}{https://orcid.org/0000-0003-2586-8183}{}

\author{Davide Bresolin}{University of Padua, Italy \and
  \url{https://www.math.unipd.it/~bresolin/}}{davide.bresolin@unipd.it}{https://orcid.org/0000-0003-2253-9878}{}

\author{Luca Geatti}{University of Udine, Italy \and
  \url{https://users.dimi.uniud.it/~luca.geatti/}}{luca.geatti@uniud.it}{https://orcid.org/0000-0002-7125-787X}{}

\author{Angelo Montanari}{University of Udine, Italy \and
  \url{https://users.dimi.uniud.it/~angelo.montanari/index.php}}{angelo.montanari@uniud.it}{https://orcid.org/0000-0002-4322-769X}{}

\author{Matteo Zavatteri}{Fondazione Bruno Kessler, Italy \and 
  \url{https://matteozavatteri.github.io/}}{zavatteri@fbk.eu}{https://orcid.org/0000-0001-6696-2972}{}

\authorrunning{R. Borelli, D. Bresolin, L. Geatti, A. Montanari, M. Zavatteri} 

\Copyright{Roberto Borelli, Davide Bresolin, Luca Geatti, Angelo Montanari and Matteo Zavatteri} 


\ifextended
\hideOASIcs
\nolinenumbers
\else 
    \relatedversiondetails[cite=extended]{Extended Version}{URL to related version}
\fi

\acknowledgements{%
  Luca Geatti acknowledges financial support from the \emph{Programma
  Operativo Complementare} of the \emph{Regione Autonoma Friuli Venezia
  Giulia}, under CUP G23C25000880002, Operation Code 2025/13846, Project
  Code 2025/13846/16. Davide Bresolin and Luca Geatti acknowledge financial support from the INdAM – GNCS Project
  ``Neuro-symbolic semantic contract for rigorous reachability and
  certificates in CPS'' (CUP E53C25002010001).
}

\EventNoEds{2}
\EventLongTitle{33rd International Symposium on Temporal Representation and Reasoning}
\EventShortTitle{TIME 2026}
\EventAcronym{TIME}
\EventYear{2026}
\EventDate{September 1--3, 2026}
\EventLocation{Cork, Ireland}
\EventLogo{}

\begin{document}

\maketitle

\begin{abstract}
  The cascade product is a fundamental construction in automata theory,
  enabling hierarchical composition of automata and playing a central role
  in decomposition results such as the Krohn–Rhodes theorem. However, its
  use is limited by the exponential size required to represent cascades,
  which stems from the fact that each component may depend on all preceding
  ones, leading to exponentially large alphabets.

  To address this issue, we introduce the $\Sigma$-chain product,
  a restricted variant in which each component depends only on the input
  alphabet and the component immediately preceding it. We show that $\Sigma$-chains
  achieve linear-size representations and can be exponentially more
  succinct than cascades.
  We prove that $\Sigma$-chains and cascades are expressively equivalent 
  even when restricting the components
  to specific classes of automata, such as permutation-reset automata. As
  a consequence, we derive that a language is regular if and only if it is
  recognized by a $\Sigma$-chain of permutation-reset automata.
  Finally, we analyze structural properties of $\Sigma$-chains of reset
  automata, including a relation with well-known subclasses of star-free
  languages.
\end{abstract}

\newpage

\section{Introduction}
\label{sec:intro}


The \emph{cascade product}
of \roberto{deterministic} automata generalises the direct product by arranging the component
automata in a linear, hierarchical structure~\cite{krohn1965algebraic,DBLP:journals/corr/abs-2010-16235,DBLP:conf/birthday/Maler10}.  While in the direct
product every component reads symbols from a fixed alphabet $\Sigma$,
in the cascade product the component at level $i$ reads symbols from
$\Sigma$ and, additionally, observe the state of every component at
level $j<i$. Formally, the automaton at level $i$ reads symbols from the
Cartesian product $\Sigma \times Q_1 \times \dots \times Q_{i-1}$, where
$Q_j$ denotes the state set of the automaton at level $j$; the first
component, as in the direct product, reads symbols from $\Sigma$ alone.
A transition labelled $(a, q_1, \dots, q_{i-1})$ is taken if and only if
the current input symbol is $a \in \Sigma$ and the automaton at level $j$
is currently in state $q_j$, for every $1 \le j < i$. Crucially, the
automaton obtained by applying the cascade product is defined over the
original alphabet $\Sigma$ alone.

The cascade product was introduced in the context of semigroup and
semiautomaton decomposition theory, and lies at the heart of the
Krohn–Rhodes Decomposition Theorem~\cite{krohn1965algebraic}. This theorem
states that every \roberto{deterministic} semiautomaton (\ie an automaton without initial and
accepting states) can be decomposed into a sequence of permutation-reset
semiautomata whose cascade product \emph{covers} the original semiautomaton
-- that is, the original is a homomorphic image of a sub-semiautomaton of
the product.\roberto{\footnote{\roberto{A stronger version of the theorem exists, where the decomposition consists of permutation automata and two-state reset automata. We use the weaker formulation to keep all components of the same type.}}} In other words, the cascade product can simulate all the
behaviours of the original semiautomaton using only very simple components,
namely permutation-reset semiautomata. A semiautomaton over a state set $Q$
is called \emph{permutation-reset} if every symbol in its alphabet induces
a transformation $Q \to Q$ that is either a permutation (\ie a bijection)
or a reset (a map whose image has cardinality one). A semiautomaton is
called \emph{reset} if every symbol induces either a reset function or the
identity (a special case of a permutation). A central insight of the
Krohn–Rhodes theorem is that every counter-free semiautomaton --
corresponding to star-free regular expressions -- can be decomposed into
a cascade product of two-state reset semiautomata.

Beyond its theoretical significance in automata theory and semigroup
algebra, the Krohn–Rhodes Decomposition Theorem has found applications in
the translation of automata to Linear Temporal Logic
formulas~\cite{DBLP:conf/focs/MalerP90,DBLP:conf/birthday/Maler10} and in
the transformation of finite-word automata into formulas of pure-past
Linear Temporal Logic -- formulas whose temporal operators refer
exclusively to the past -- in what is known as the \emph{pastification}
problem~\cite{DBLP:conf/stacs/BorelliGMM25,DBLP:conf/kr/ArtaleGGMM23}. An
application of the cascade product in formal verification appears
in~\cite{DBLP:conf/overlay/Geatti25}, where the approach is to decompose
the specification to be model-checked against a system into a sequence of
components arranged in a cascade structure, and to execute the
model-checking algorithm one component at a time, starting from the first.
The rationale is that component-by-component verification is generally more
efficient than the monolithic verification of the full system. More
recently, the cascade product has been extended from semiautomata to full
automata in~\cite{DBLP:conf/stacs/BorelliGMM25}, where the
relationship between cascades of a given number of reset automata and the
corresponding classes of regular expressions has been studied. It has been
shown that the cascade product of permutation-reset automata (respectively,
of reset automata) captures exactly the class of regular (respectively,
star-free) languages.

A significant drawback that limits the applicability of the cascade product
is the size of its representation. On one hand, the automaton obtained by computing the
cascade product of $n$ (semi)automata $\autom_1,\dots,\autom_n$ has
a number of states exponential in $n$, a phenomenon shared with the direct
product. On the other hand, and more problematically, even representing the
\emph{sequence of components} $\seq{\autom_1, \dots, \autom_n}$ -- which we
call a \emph{cascade} of (semi)automata -- without computing the actual product requires exponential space.
This is a consequence of the component alphabets: in particular, the
alphabet of the last component $\autom_n$ has cardinality $|\Sigma| \times
|Q_1| \times \dots \times |Q_{n-1}|$, which is already exponential in $n$.
\roberto{Since we consider deterministic automata, the number of transitions of $\autom_n$ is also exponential in $n$. While this represents a significant drawback from the perspective of space complexity, it is also a key source of expressive power, allowing $\autom_n$ to distinguish every possible configuration of the preceding components together with the current input symbol, and to select a different successor state for each such configuration.}
By contrast, representing the sequence of components of a direct product
requires only space linear in $n$ and in the size of the largest automaton.

Motivated by the desire to overcome this representational inefficiency
while surpassing the expressive power of the direct product (which only captures closure under intersection), we
introduce in this paper the notion of \emph{$\Sigma$-chains} and the
corresponding \emph{$\Sigma$-chain product}.
A $\Sigma$-chain is a sequence of automata in which the first component
reads symbols from $\Sigma$, while each subsequent component depends only on $\Sigma$ and on the immediately preceding component, that
is, where the component at level $i$ reads symbols from the Cartesian product $\Sigma \times Q_{i-1}$,
where $Q_{i-1}$ is the state set of the automaton at level $i-1$. The
$\Sigma$-chain product is the resulting product automaton, defined
analogously to the cascade product: the $i$-th automaton takes a transition
labeled $(a, q_{i-1})$ if and only if the current input symbol is $a$ and
the preceding component is in state $q_{i-1}$. The $\Sigma$-chain product
is thus a restriction of the cascade product in which each component's
dependency is confined to $\Sigma$ and the component immediately preceding it, and
a generalization of the direct product in that it admits a hierarchical
dependency structure among components.

We begin by analyzing the representational complexity of $\Sigma$-chains. We show
that, while the $\Sigma$-chain product of $n$ automata
is an automaton with a number of states
exponential in $n$ (as with both the cascade and the direct product), the
$\Sigma$-chain itself -- \ie the sequence of components -- has size linear
in $n$ (as with the direct product, but unlike the cascade). Furthermore,
we demonstrate that $\Sigma$-chains can be exponentially more succinct than
cascades: there exists a family of languages recognisable by
$\Sigma$-chains of $n$ reset automata (thus representable in space
polynomial in $n$) such that any cascade of reset automata for those
languages requires at least $n$ components, and hence at least exponential
space.

We then turn to the expressive power of $\Sigma$-chains.  We show
that $\Sigma$-chains and cascades of automata are expressively equivalent.
We first observe that every $\Sigma$-chain is, by construction, a special
case of a cascade: the dependency restriction imposed on each component is
simply a syntactic constraint, and lifting it yields a valid cascade with
identical components. The converse direction is less immediate. We show
that every cascade can be transformed into an equivalent $\Sigma$-chain via
a \emph{chainification} technique: intuitively, the dependence of
a component -- say, the third -- on the states of a non-adjacent
predecessor -- say, the first -- can be eliminated by encoding the relevant
state information into the state space of the intervening component.
We prove, moreover, that for the classes of permutation-reset, pure-reset (all
symbols induce a reset), and permutation automata (all symbols induce
a permutation), the resulting $\Sigma$-chain has the same
number of components as the original cascade, and each component belongs to
the same class. \roberto{This comes at the cost that, at level $i$, the number of states becomes exponential in $i$}. A corollary of this transformation is
that $\Sigma$-chains of permutation-reset automata recognise exactly the
class of regular languages. \roberto{This means that, for permutation-reset automata, maximal expressive power can be achieved without requiring the $i$-th level to depend on all previous levels. Rather, it only needs to depend on the immediately preceding level and the input alphabet.} 
For reset automata, the proposed transformation does not preserve the component type: a reset automaton may become a permutation-reset automaton. While the number of components is preserved, the resulting $\Sigma$-chain may no longer consist entirely of reset automata. Therefore, we can only establish that languages recognized by $\Sigma$-chains of reset automata are included in the star-free languages, leaving the converse open.
\cref{tab:expressiveness} summarizes and compares the expressiveness results for $\Sigma$-chains against known results on cascades and direct products.

\begin{table}[]
    \centering
\begin{adjustbox}{max width=1\textwidth}
\centering
\begin{tabular}{|l|>{\centering\arraybackslash}m{3cm}c>{\centering\arraybackslash}m{3cm}c>{\centering\arraybackslash}m{3cm}|}
\hline
 &   {\textbf{Direct Product}} &  & {\textbf{$\Sigma$-chains}} & & {\textbf{Cascades}}\\ \hline
$\PureResets$ & \makecell[c]{$\bigcap_i (\Sigma^*R_i \cup J_i)$\\ $J_i
  = \epsilon$ or $J_i = \emptyset$ \\ $R_i, I_i \subseteq{\Sigma} $, $R_i
  \cap I_i = \emptyset$ \\ \cite{DBLP:conf/stacs/BorelliGMM25}} & \makecell[c]{$\subsetneq$\\
  \textbf{\cref{lemma:directprodct}}} & \emph{open} & \makecell[c]{$=$ \\
  \textbf{\cref{coro:cascades:chains}}} & \emph{open} \\ \hline
$\Resets$ & \makecell[c]{$\bigcap_i (\Sigma^*R_i I_i^*\cup J_i)$ \\ $J_i
  = I_i^*$ or $J_i = \emptyset$ \\ $R_i, I_i \subseteq{\Sigma} $, $R_i \cap
  I_i = \emptyset$ \\ \cite{DBLP:conf/stacs/BorelliGMM25} }& \makecell[c]{$\subsetneq$ \\
  \textbf{\cref{lemma:directprodct}}} & \makecell{\emph{open$^\dagger$}} & \makecell[c]{$\subseteq$ \\ 
  \textbf{\cref{coro:cascades:chains}}} & \makecell{$\SF$ \\ \cite{journalversionofstacs}}  \\ \hline
  $\PermutationResets$  &\emph{open} & \emph{open}& \makecell[c]{$\REG$ \\ 
  \textbf{\cref{coro:cascades:chains}}} & \makecell[c]{$=$ \\ 
  \textbf{\cref{coro:cascades:chains}}} & \makecell{$\REG$ \\ \cite{journalversionofstacs}} \\ \hline
\end{tabular}
\end{adjustbox}
  \caption{Summary of expressiveness results on direct products,
  $\Sigma$-chains and cascades of automata,
  for the classes of pure-reset ($\PureResets$), reset ($\Resets$) and
  permutation-reset ($\PermutationResets$) automata.  $\SF$ denotes the
  class of star-free languages and $\REG$ denotes the class of all regular
  languages.
  $^\dagger$ Some preliminary results  on the expressiveness of $\Sigma$-chains of reset automata are discussed in \cref{sec:chains:reset}.
  }
  \label{tab:expressiveness}
\end{table}

We next investigate $\Sigma$-chains of reset automata, showing that they can recognize several natural and well-studied classes of star-free languages, including Piecewise Testable languages~\cite{simon2005piecewise} and R-trivial languages~\cite{brzozowski1980languages}. We also consider two
hierarchies of languages obtained by restricting $\Sigma$-chains of reset
automata by height (number of components) and by width (number of states
per level).  We prove that the height hierarchy is strict -- for every $h$,
there exist languages recognisable by a $\Sigma$-chain of reset automata of
height $h$ but not of height $h-1$ -- whereas the width hierarchy
collapses: every $\Sigma$-chain of reset automata is equivalent to one in
which the number of states of each component depends only on the size of
the alphabet $\Sigma$.  
Finally, we address the pastification problem for $\Sigma$-chains of reset automata, namely the transformation of such structures into equivalent formulas in pure past Linear Temporal Logic.
A straightforward approach would consist in first translating a $\Sigma$-chain into a cascade and then applying the procedure of~\cite{journalversionofstacs} to obtain the corresponding formula.
Instead, we adapt the construction of~\cite{journalversionofstacs} so as to directly translate a $\Sigma$-chain of reset automata into a formula of quadratic size, thereby avoiding the exponential blow-up caused by the intermediate transformation from $\Sigma$-chains to cascades.

The paper is structured as follows.  Related work is discussed in
\cref{sec:related}. In \cref{sec:background}, we provide some background
knowledge. In \cref{sec:chains}, we introduce the notion of
$\Sigma$-chains and $\Sigma$-chain product of automata, we state some basic
results and we discuss succinctness and expressiveness through the use of
the chainification (the transformation of cascades into $\Sigma$-chains)
technique.  In \cref{sec:chains:reset} we focus on $\Sigma$-chains of reset
automata,  in particular we study height and width hierarchies, we compare
$\Sigma$-chains of reset automata to piecewise-testable and \rtrivial
languages and we discuss an efficient pastification procedure for
$\Sigma$-chains.  Finally, in \cref{sec:conclusions} we assess the work
done, and we outline some directions for future works.
\ifextended
    Proofs of all results are provided in the appendix.
\else
    Proofs of all results are provided in the Extended Version \cite{extende}.
\fi

\section{Related Work}
\label{sec:related}

The composition of (semi)automata has been extensively studied in automata theory and semigroup theory. In particular, the cascade product plays a central role in decomposition results such as the Krohn–Rhodes Theorem~\cite{krohn1965algebraic}, which shows that every finite semiautomaton can be decomposed into a cascade of permutation-reset semiautomata.
Since then, several works have investigated structural, algebraic, and logical properties of cascades of automata and semiautomata, while others have also considered restrictions of the cascade product. 

In \cite{gelderie2011classifying}, Gelderie studies structural properties of cascade products and introduces the notion of scope: a cascade has scope $k$ if each transition of each component depends on the states of at most the previous $k$ components. Cascades of scope $1$ have a strictly linear dependency structure and can be readily transformed into $\Sigma$-chains. Gelderie further shows that piecewise-testable languages are recognized by cascades of reset semiautomata of scope $1$, and \rtrivial languages by cascades of scope $2$.

Our work differs in two ways. First, we consider full automata rather than semiautomata. Second, instead of studying cascades of bounded-scope, we introduce $\Sigma$-chains, where the $i$-th component depends only on the input letter and the state of component $i-1$. We exhibit $\Sigma$-chains of reset automata for generalizations of both piecewise-testable and \rtrivial languages. As a byproduct, since \rtrivial languages are recognized by finite unions of $\Sigma$-chains of reset automata and cascades are closed under union, \rtrivial languages are also recognized by cascades of reset semiautomata of scope $1$, improving Gelderie's bound.

In \cite{gecseg2012products}, Gecseg surveys a wide range of results on homomorphically and isomorphically complete classes of semiautomata with respect to $\alpha_j$-products \cite{imreh1977alphai}. Informally, the $\alpha_0$-product corresponds to a cascade product in which the dependency structure among components is acyclic, i.e., no loops are allowed between the dependencies of the automata.
In \cite{domosi1983vi}, Domosi introduces the notion of $v_i$-product, where each semiautomaton is allowed to depend on at most $i$ other semiautomata. The $\alpha_0$-$v_1$ product \cite{gecseg1991} combines these two restrictions: the dependency structure is acyclic and each component depends on at most one other component.

This construction is closely related to the $\Sigma$-chain product introduced in this paper. Indeed, both frameworks restrict the dependency structure among components to be essentially linear. In particular, every $\Sigma$-chain product defines an $\alpha_0$-$v_1$ product whose dependency graph is connected. Conversely, an $\alpha_0$-$v_1$ product may consist of several independent connected components; once an initial state and a set of final states are fixed, these components can be sequentially concatenated (in any order) to obtain a corresponding $\Sigma$-chain.
In \cite{gecseg1991}, it is shown that the $\alpha_0$-$v_1$ product is equivalent, in terms of homomorphic representations, to the general product, i.e., a product in which each component may depend on all the others. At first glance, this may seem to contradict the results of the present paper, since general products of reset automata can recognize languages that are not star-free, suggesting that $\Sigma$-chain of reset automata products might also exceed the expressive power of the star-free class.
However, this apparent discrepancy is resolved by observing that the equivalence result of \cite{gecseg1991} relies on infinite products. Furthermore, the same work shows that, when restricting to finite products of finite automata, both the $\alpha_0$-product and the $v_1$-product are strictly more expressive than the $\alpha_0$-$v_1$ product with respect to homomorphic representations.
The expressiveness results for $\Sigma$-chain products we establish in this paper can thus be seen as an extension of the finite-product results of~\cite{gecseg1991}, providing a more refined understanding of the expressive power of finite $\alpha_0$-$v_1$ products.

\section{Background}
\label{sec:background}

In this section, we introduce the necessary background on automata and
cascades of automata.

\subsection{Automata}

A \roberto{deterministic} \emph{semiautomaton} $\autom$ is a tuple $\tuple{\Sigma,Q,\delta}$ such
that:
\begin{enumerate*}[label=(\roman*)]
  \item $\Sigma$ is the alphabet, \ie a finite set of symbols;
  \item $Q$ is the set of states;
  \item $\delta : Q \times \Sigma \to Q$ is the transition function.
\end{enumerate*}
An \roberto{deterministic} \emph{automaton} $\autom = \tuple{\Sigma,Q,\delta,s_0,F}$ is
a semiautomaton augmented with an initial state $s_0$ and a set $F$ of
final states. \roberto{From now on, the terms automaton and semiautomaton refer exclusively to deterministic ones.} With $\delta^*$ we denote the Kleene closure of $\delta$.
The \emph{size} of $\autom$, denoted with $|\autom|$, is defined as
$|\delta|$.

Given an automaton $\autom = \tuple{\Sigma,Q,\delta,q_0,F}$ and a (finite)
word $\sigma \coloneqq \seq{\sigma_0,\dots,\sigma_n} \in \Sigma^*$,
\emph{the run $\tau \in Q^+$ induced by $\sigma$} is a sequence
$\seq{q_0,q_1,\dots,q_{n+1}}$ starting from the initial state $q_0$ and 
such that $\delta(q_i,\sigma_i) = q_{i+1}$,
for all $0\le i\le n$. We say that $\tau$ is \emph{accepting} iff $q_{n+1}
\in F$. A word $\sigma \in \Sigma^*$ is \emph{accepted} by $\autom$ iff the
run induced by $\sigma$ is accepting. We define the \emph{language of
$\autom$} (or \emph{the language recognized by $\autom$}), denoted by
$\lang(\autom)$, as the set of accepted words.  Given a state $q \in Q$,
let $\lang_q(\autom)$ be the set of words inducing a run $\tau \coloneqq
\seq{q_0,\dots,q_{m}}$ with $q_{m} = q$.

Each symbol of the alphabet of an automaton $\autom$ (with set of states
$Q$) induces a transformation $Q \to Q$, defined as follows.

\begin{definition}[Function induced by a symbol]
  Let $\autom = \tuple{\Sigma,Q,\delta,q_0,F}$ be an automaton.
  For each symbol $a \in \Sigma$, we define the \emph{function induced by
  $a$ in $\autom$}, denoted with $\tau_{a}^{\autom}$ or simply $\tau_{a}$
  (when $\autom$ is clear from the context), as the function $\tau_{a}
  : Q \to Q$ such that, for all $s \in Q$, we have that $\tau_a(s) = s'$ iff
  $\delta(s,a) = s'$.
\end{definition}

Let $\autom = \tuple{\Sigma,Q,\delta,q_0,F}$ be an automaton and let $a \in
\Sigma$ be a symbol in the alphabet.  We say that $\tau_a$ is a \emph{reset
function} (or simply a \emph{reset}) iff there exists a $q' \in Q$ such
that, for all $q \in Q$, we have that $\tau_a(q) = q'$; in this case, we say
that $\tau_a$ is a \emph{reset on $q'$}.  If instead $\tau_a$ is bijective,
we say that $\tau_a$ is a \emph{permutation function} (or simply
a \emph{permutation}). 

Reset automata and permutation automata are defined as follows.

\begin{definition}[The classes of Pure-reset, Reset, Permutation, and
  Permutation-Reset automata]
\label{def:classes}
  Let $\autom = \tuple{\Sigma,Q,\delta,q_0,F}$ be an automaton.  $\autom$
  is called: 
  \begin{itemize}
    \item a \emph{pure-reset automaton} iff, for each $a \in \Sigma$,
      $\tau_a$ is a reset function;
    \item a \emph{reset automaton} iff, for each symbol $a \in \Sigma$, $\tau_{a}$ is either the \emph{identity} function or a \emph{reset} function;
    \item a \emph{permutation automaton} iff, for each symbol $a \in
      \Sigma$, $\tau_{a}$ is a permutation function.
    \item a \emph{permutation-reset automaton} iff, for each symbol $a \in
      \Sigma$, $\tau_a$ is either a reset or a permutation function.
  \end{itemize}
  We denote with $\PureResets$, $\Resets$, $\Permutations$ and
  $\PermutationResets$ the classes of pure-reset, reset, permutation, and
  permutation-reset automata, respectively.
\end{definition}

Let $A \in \set{\PureResets, \Resets, \Permutations, \PermutationResets}$
be a class of automata. We denote with $\DP[A]$ the class of languages
definable by a direct product of automata in the class $A$.

\begin{proposition}[Size of direct products]
\label{prop:size:direct}
  Let $D$ be a set of automata $D = \set{\autom_1,\dots,\autom_n}$ such
  that $\autom_i$ has alphabet $\Sigma$ and set of states $Q_i$ with $|Q_i|
  \geq 2$. Let $m = \max\set{|Q_i| : 1\le i\le m}$. It holds that:
  \begin{itemize}
    \item $|\autom_1 \times \dots \times \autom_n| \ge |\Sigma| \cdot 2^{n}$;
    \item $|D| = \sum_i |\autom_i| \leq |\Sigma| \cdot n \cdot m$.
  \end{itemize}
\end{proposition}

\subsection{Cascades and the Cascade Product}
The cascade
product~\cite{DBLP:conf/focs/MalerP90,DBLP:conf/birthday/Maler10} between
two (semi)automata $\autom$ and $\automb$, is a generalization of the
direct product in which $\automb$ reads the Cartesian product of the
input alphabet with the set of states of $\autom$, and where the
(semi)automaton $\automb$ transitions from state $s$ to state $s'$ reading
the symbol $(a,q)$ if the current input symbol is $a$ and the
(semi)automaton $\autom$ is in state $q$. The (semi)automaton obtained by
applying the cascade product is defined over the Cartesian product of the set of states of $\autom$ and
$\automb$ and the alphabet $\Sigma$, and is such that transition function is applied
component-wise.  Additionally, for full automata~\cite{DBLP:conf/stacs/BorelliGMM25}, the initial state of the
product automaton is the state comprising the initial states of $\autom$
and $\automb$ and the set of final states is the Cartesian
product of the final states of $\autom$ and $\automb$.

\begin{definition}[Cascade Product of automata]
\label{def:cascade:dfa}
  Let $\Sigma$ be a finite alphabet and let $\autom
  = \tuple{\Sigma,Q,\delta,q_0,F}$ and $\automb = \tuple{\Sigma\times Q,
  Q', \delta',q_0',F'}$ be two automata over the alphabets $\Sigma$ and
  $\Sigma \times Q$, respectively.
  The \emph{cascade product of $\autom$ and $\automb$}, denoted by $\autom
  \circ \automb$, is the automaton $\tuple{\Sigma,Q\times
  Q',\delta'',(q_0,q_0'),F \times F'}$ where $\delta''$ is defined as:
  \begin{center}
    $\delta''((q,q'),a) = (\delta(q,a), \delta'(q',(a,q)))$
  \end{center}
  for all $(q,q') \in
  Q \times Q'$ and for all $a \in \Sigma$.
\end{definition}

In the following, we will use the term \emph{cascade} to denote a sequence
of automata $\cascade = \seq{\autom_1,\autom_2,\dots,\autom_n}$ to which
the cascade product is applied, and we define the \emph{size} of
$\cascade$, denoted with $|\cascade|$, as $\sum_{i=1}^{n} |\autom_i|$.
Instead, with the term \emph{cascade product}, we will denote the automaton
$C = \autom_1 \circ \dots \circ \autom_n$ obtained by applying the cascade
product to the given sequence of automata.  We define the \emph{height} of
the cascade $\cascade = \seq{\autom_1,\dots,\autom_n}$ as $n$. Given
a cascade $\cascade = \seq{\autom_1,\autom_2,\dots,\autom_n}$, we
define the \emph{language of $\cascade$}, denoted with $\lang(\cascade)$,
as the language $\lang(\autom_1 \circ \dots \circ \autom_n)$.

The proposition below shows that, contrary to the case of the direct
product (\cf \Cref{prop:size:direct}), representing a cascade (\ie the
\emph{sequence} of components to which the cascade product is applied)
requires no less space than representing the product automaton (\ie the
automaton obtained after applying the cascade product), which is
exponential with respect to the number of components. This is due to the
fact that, for cascades of automata, the alphabet of each component is the
Cartesian product between $\Sigma$ and the set of states of all previous
components. In particular, in a cascade $\seq{\autom_1,\dots,\autom_n}$,
the alphabet of the last component $\autom_n$ has cardinality $|\Sigma|
\times |Q_1| \times \dots \times |Q_{n-1}|$, which is exponential in $n$.

\begin{proposition}[Size of cascades~\cite{journalversionofstacs}]
\label{prop:size:cascades}
  For each cascade $\cascade = \seq{\autom_1,\dots,\autom_n}$, where $Q_i$
  is the set of states of $\autom_i$, for each $1\le i\le n$, it holds
  that:
  \begin{itemize}
    \item $|\autom_1 \circ \dots \circ \autom_n| \ge |\Sigma| \cdot 2^{n}$;
    \item $|\cascade| \ge 2 \cdot |\Sigma| \cdot (2^n - 1)$.
  \end{itemize}
\end{proposition}

Let $A \in \set{\PureResets, \Resets, \Permutations, \PermutationResets}$
be a class of automata. We denote with $\C[A]$ the class of languages
recognized by cascades of automata of class $A$. 
A language $\lang \subseteq \Sigma^*$ is regular iff it can be defined
starting from finite subsets of $\Sigma^*$ by the applications of union,
negation, concatenation, and closure under Kleene star. $\lang$ is
\emph{star-free} iff it can be generated starting from finite subset of
$\Sigma^*$ only by the application of union, negation, and concatenation.
The following results characterizes the expressiveness of $\C[\Resets]$ and
$\C[\PermutationResets]$.

\begin{theorem}[\cite{journalversionofstacs}]
\label{lemma:star:free}
  The following characterizations hold:
  \begin{itemize}
    \item A language $\lang$ is regular if and only if $\lang \in
      \C[\PermutationResets]$;
    \item A language $\lang$ is star-free if and only if $\lang \in
      \C[\Resets]$.
  \end{itemize}
\end{theorem}

\section{$\Sigma$-Chains of automata}
\label{sec:chains}

In this section we define \emph{$\Sigma$-chains of automata} and the
associated \emph{$\Sigma$-chain product}. We then study their main
properties in terms of succinctness and expressive power.

\subsection{Definitions, Basic Properties and Succinctness}

We define \emph{$\Sigma$-chains of automata} as a restricted variant of
cascades of automata, in which the alphabet of the first component is
$\Sigma$ (as in the case of cascades), while the alphabet of each successive component
$\autom_i$ is the Cartesian product between $\Sigma$ and the set of states
of the previous component $\autom_{i-1}$.

\begin{definition}[$\Sigma$-chain Product of automata]
\label{def:chain:dfa}
  Let $\Sigma$ be an alphabet and let $\autom_1,\dots,\autom_n$ be $n$
  automata such that:
  \begin{itemize}
      \item $\autom_1 = \seq{\Sigma, Q^1, \delta^1, q_0^1, F^1}$;
      \item $\autom_i = \seq{\Sigma\times Q^{i-1}, Q^i, \delta^i, q_0^i,
        F^i}$, for each $i>1$.
  \end{itemize}
  We define the \emph{$\Sigma$-chain product between
  $\autom_1,\dots,\autom_n$}, denoted with $\autom_1 \odot \dots \odot
  \autom_n$, as the automaton $\tuple{\Sigma,Q^1\times \dots\times
  Q^n,\delta,(q_0^1,\dots,q_0^n),F^1 \times \dots \times F^n}$ such that:
  \begin{align}
    \delta((q^1,q^2,\dots,q^n),a) = (\delta^1(q^1,a), \delta^2(q^2,(a,q^1)),
    \dots, \delta^n(q^n,(a,q^{n-1})))
  \end{align}
  for all $(q^1,q^2,\dots,q^n) \in Q^1 \times Q^2 \times \dots \times  Q^n$
  and for all $a \in \Sigma$.
\end{definition}

In the following, we will use the term \emph{$\Sigma$-chain} to denote
a sequence of automata to which the $\Sigma$-chain product is applied. To
distinguish them from cascades, we will denote this sequence as $\cascade
= \chain{\autom_1,\autom_2,\dots,\autom_n}$.
As with cascades, we define the \emph{size} of the $\Sigma$-chain $\cascade
= \chain{\autom_1,\autom_2,\dots,\autom_n}$, denoted with $|\cascade|$, as
$\sum_{i=1}^{n} |\autom_i|$, and its \emph{height} as $n$.

Given a $\Sigma$-chain $\cascade
= \chain{\autom_1,\autom_2,\dots,\autom_n}$, we define the \emph{language
of $\cascade$}, denoted with $\lang(\cascade)$, as the language
$\lang(\autom_1 \odot \dots \odot \autom_n)$. From now on, given a cascade
(or a $\Sigma$-chain) $\cascade$ and a $\Sigma$-chain (or a cascade)
$\cascade'$, we say that $\cascade$ is \emph{equivalent} to $\cascade'$ iff
$\lang(\cascade) = \lang(\cascade')$.


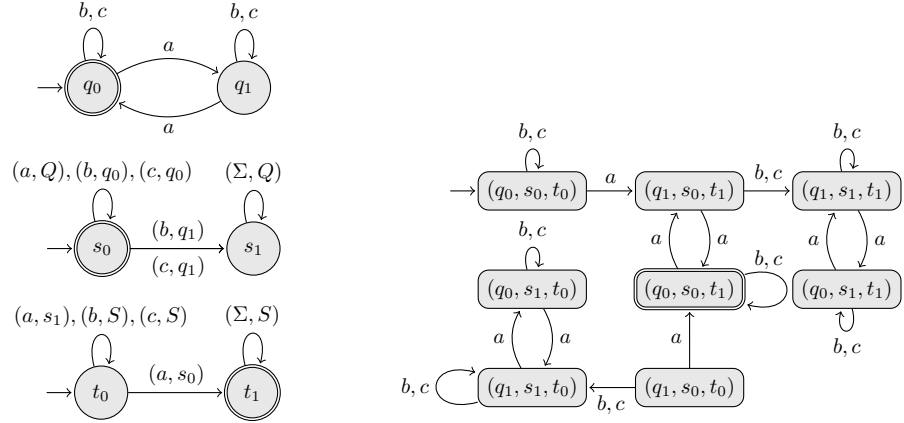
\begin{figure}[t!]
\centering
\begin{subfigure}{.45\textwidth}
  \centering
  \begin{subfigure}{\textwidth}
    \centering
    \begin{tikzpicture}[shorten >=1pt,node distance=2.5cm,on grid,auto,
      scale=0.8, every node/.style={scale=0.8}]
      \tikzstyle{every state}=[fill={rgb:black,1;white,10}]
        \node[state,initial,initial text=,accepting] (q_0) {$q_0$};
        \node[state] (q_1) [right of=q_0] {$q_1$};

        \path[->]
        (q_0) edge [bend left] node {$a$} (q_1)
        (q_1) edge [bend left] node {$a$} (q_0)
        (q_0) edge [loop above] node {$b,c$} ()
        (q_1) edge [loop above] node {$b,c$} ();
    \end{tikzpicture}
  \end{subfigure}\vspace{4pt}
  \begin{subfigure}{\textwidth}
    \centering
    \begin{tikzpicture}[shorten >=1pt,node distance=2.5cm,on grid,auto,
      scale=0.8, every node/.style={scale=0.8}]
      \tikzstyle{every state}=[fill={rgb:black,1;white,10}]
        \node[state,initial,initial text=,accepting] (s_0) {$s_0$};
        \node[state] (s_1) [right of=s_0] {$s_1$};

        \path[->]
        (s_0) edge [] node {$(b,q_1)$} (s_1)
        (s_0) edge [] node[below] {$(c,q_1)$} (s_1)
        (s_0) edge [loop above] node {$(a,Q),(b,q_0),(c,q_0)$} ()
        (s_1) edge [loop above] node {$(\Sigma,Q)$} ();
    \end{tikzpicture}
  \end{subfigure}\vspace{4pt}
  \begin{subfigure}{\textwidth}
    \centering
    \begin{tikzpicture}[shorten >=1pt,node distance=2.5cm,on grid,auto,
      scale=0.8, every node/.style={scale=0.8}]
      \tikzstyle{every state}=[fill={rgb:black,1;white,10}]
        \node[state,initial,initial text=] (t_0) {$t_0$};
        \node[state,accepting] (t_1) [right of=t_0] {$t_1$};

        \path[->]
        (t_0) edge [] node {$(a,s_0)$} (t_1)
        (t_0) edge [loop above] node {$(a,s_1),(b,S),(c,S)$} ()
        (t_1) edge [loop above] node {$(\Sigma,S)$} ();
    \end{tikzpicture}
  \end{subfigure}
\end{subfigure}
\begin{subfigure}{.45\textwidth}
  \centering
  \begin{tikzpicture}[shorten >=1pt,node distance=2.6cm,on grid,auto,
    scale=0.8, every node/.style={scale=0.8},
    every loop/.style={min distance=5mm,looseness=5}]
    \tikzstyle{every state}=[fill={rgb:black,1;white,10},
      shape=rectangle, rounded corners=4pt,
      inner sep=4pt, minimum size=0pt]
      \node[state,initial,initial text=] (x_000)
        {$(q_0,s_0,t_0)$};
      \node[state] (x_101) [right of=x_000] {$(q_1,s_0,t_1)$};
      \node[state] (x_111) [right of=x_101] {$(q_1,s_1,t_1)$};
      \node[state] (x_010) [below=1.3cm of x_000] {$(q_0,s_1,t_0)$};
      \node[state,accepting] (x_001) [below=1.3cm of x_101] {$(q_0,s_0,t_1)$};
      \node[state] (x_011) [below=1.3cm of x_111] {$(q_0,s_1,t_1)$};
      \node[state] (x_110) [below=1.3cm of x_010] {$(q_1,s_1,t_0)$};
      \node[state] (x_100) [below=1.3cm of x_001] {$(q_1,s_0,t_0)$};

      \path[->]
      (x_000) edge [] node {$a$} (x_101)
      (x_101) edge [bend left] node {$a$} (x_001)
      (x_101) edge [] node {$b,c$} (x_111)
      (x_111) edge [bend left] node {$a$} (x_011)
      (x_010) edge [bend left] node {$a$} (x_110)
      (x_001) edge [bend left] node {$a$} (x_101)
      (x_011) edge [bend left] node {$a$} (x_111)
      (x_110) edge [bend left] node {$a$} (x_010)
      (x_100) edge [] node {$b,c$} (x_110)
      (x_100) edge [] node {$a$} (x_001)
      (x_000) edge [loop above] node {$b,c$} ()
      (x_111) edge [loop above] node {$b,c$} ()
      (x_010) edge [loop above] node {$b,c$} ()
      (x_001) edge [loop right] node[above,yshift=2mm,xshift=-3mm] {$b,c$} ()
      (x_011) edge [loop below] node {$b,c$} ()
      (x_110) edge [loop left] node {$b,c$} ();
  \end{tikzpicture}
  \end{subfigure}
  \caption{%
    On the left, starting from the top, the three components $\autom_1$,
    $\autom_2$, and $\autom_3$ of a $\Sigma$-chain $\cascade$, with
    $\Sigma=\set{a,b,c}$. On the right, the associated $\Sigma$-product. We
    denoted final states with double circles and with $(a,Q)$ we denote the
    set of symbols $\set{(a,q) \suchthat q \in Q}$.
  }
  \label{fig:sigma:chain:example}
\end{figure}

An example of a $\Sigma$-chain and of the corresponding $\Sigma$-chain
product is shown in~\cref{fig:sigma:chain:example}. The three automata on
the left constitute the components of the $\Sigma$-chain $\cascade
= \chain{\autom_1,\autom_2,\autom_3}$.  Automaton $\autom_1$ is
a permutation-reset automaton, whereas $\autom_2$ and $\autom_3$ are reset
automata.  The language recognized by $\cascade$ consists of all words over
the alphabet $\Sigma = \set{a, b, c}$ satisfying the following conditions: 
\begin{enumerate*}[label=(\roman*)]
  \item the number of occurrences of the symbol \emph{``a''} is even; 
  \item for every $i^{th}$ occurrence of \emph{``a''} with odd index $i$,
    the immediately following symbol must also be \emph{``a''}; and
  \item there is at least one occurrence of \emph{``a''}. 
\end{enumerate*}
Each component of the $\Sigma$-chain enforces one of these conditions. In
particular, since $\tau_a$ is a permutation in $\autom_1$ and $q_0$ is its
only accepting state, $\autom_1$ accepts precisely those inputs containing
an even number of occurrences of \emph{``a''}. The automaton $\autom_2$
transitions to the sink state $s_1$ upon reading either $(b, q_1)$ or $(c,
q_1)$. This implies that if the most recently read \emph{``a''} corresponds
to an odd occurrence (\ie $\autom_1$ is in state $q_1$) and the current
symbol is $b$ or $c$, then the input word--and all its continuations--must
be rejected. The automaton obtained by applying the $\Sigma$-chain product
to $\cascade$ is depicted on the right-hand side
of~\cref{fig:sigma:chain:example}.

Let $A \in \set{\PureResets, \Resets, \Permutations, \PermutationResets}$
be a class of automata. We denote by $\SCH[A]$ the class of languages
recognized by $\Sigma$-chains whose components are automata in the class
$A$.  
The following proposition shows that $\Sigma$-chains are a particular case
of cascades. In particular, each $\Sigma$-chain of automata in a class $A
\in \set{\PureResets, \Resets, \Permutations, \PermutationResets}$ can be
converted to an equivalent cascade of automata of the same class $A$ with
the same height and with the same number of states in each component.

\begin{restatable}{proposition}{chaintocascades}
\label{prop:chain:to:cascades}
Let $\cascade = \chain{\autom_1,\dots,\autom_n}$ be a $\Sigma$-chain where,
for every $i=1,\dots,n$, $\autom_i$ is of class $A_i \in \set{\PureResets,
\Resets, \Permutations, \PermutationResets}$ and has 
$m_i$ states. Then, there exists a cascade $\cascade' = \seq{\autom'_1,\dots,\autom'_n}$ such
that $\lang(\cascade) = \lang(\cascade')$, where each $\autom'_i$ is of
class $A_i$ and has $m_i$ states (for every $i=1,\dots,n$).
\end{restatable}

From~\cref{prop:chain:to:cascades,lemma:star:free}, it directly follows
that the language of any $\Sigma$-chain of reset automata is a star-free
language and that the language of any $\Sigma$-chain of permutation-reset
automata is a regular language.

We analyze the size of $\Sigma$-chains with respect to the number of its
components and we compare it to the size of its associated $\Sigma$-chain
product.
The next proposition shows that, as for direct products (\cf
\Cref{prop:size:direct}), the size of a $\Sigma$-chain is always linear in
the number of its components, while the size of its associated
$\Sigma$-chain product is exponential (in the number of components).  This
is a key difference with respect to the case of cascades, in which
\emph{both} the size of a cascade and the size of its cascade product are
exponential in the number of the components (\cref{prop:size:cascades}).

\begin{restatable}[Size of $\Sigma$-chains and $\Sigma$-chain products]{proposition}{chainsize}
\label{prop:chain:size}
  Let $\cascade = \chain{\autom_1,\dots,\autom_n}$ be a $\Sigma$-chain where,
  for every $i=1,\dots,n$, $Q_i$ is the set of states of $\autom_i$, and let $m = \max\set{|Q_i| : 1\le
  i\le n}$. The following statements hold:
 \begin{itemize}
    \item $|\autom_1 \odot \dots \odot \autom_n| 
      \ge |\Sigma| \cdot 2^{n}$;
    \item $|\cascade| 
      \le |\Sigma|\cdot n\cdot m^2$.
  \end{itemize}
\end{restatable}

The preceding proposition raises the following question: \emph{do there
exist languages that admit a compact representation via $\Sigma$-chains but
not via cascades?} The next theorem answers this positively for the class of
reset automata, by exhibiting a family of languages
$\set{\lang_n}_{n=1}^{\infty}$ witnessing an exponential succinctness gap.
Each $\lang_n$ can be recognized by a $\Sigma$-chain of $n$ reset automata
of polynomial size (\cf~\cref{prop:chain:size}), while any cascade of reset
automata recognizing $\lang_n$ must have at least $n$ components, and
therefore exponential size (\cf~\cref{prop:size:cascades}).

\begin{restatable}[Succinctness of $\Sigma$-chains]{theorem}{chainsuccintness}
\label{teo:chain:succinctness}
    Let $\Sigma$ be a finite alphabet.  For every $h>0$, let $\lang_h
    = \Sigma^{h} \cdot \Sigma^*$. The following holds:
    \begin{itemize}
        \item the smallest cascade of reset automata $\cascade_h$
          recognizing $\lang_h$ has size $|\cascade_h| \in
          \Omega(2^h\cdot|\Sigma|)$;
        \item there exists a $\Sigma$-chain of two-state reset automata
          $\cascade'_h$ recognizing $\lang_h$ and of size $|\cascade'_h|
          \in O(h \cdot |\Sigma|)$.
    \end{itemize}
\end{restatable}
\begin{proof}[Proof (sketch)]
  It is possible to construct a $\Sigma$-chain $\cascade'_h
  = \chain{\autom_1,\dots,\autom_h}$ recognizing $\lang_h$ as follows (for
  every $1\le i\le h$, $\autom_i$ has only two states, the initial state
  $q_0^i$ and the only final state $q_1^i$):
  \begin{enumerate*}[label=(\roman*)]
    \item[(1)] in $\autom_1$, each symbol in $\Sigma$ induces a reset on
      $q_1^i$;
    \item[(2)] for every $1\le i\le h$, $\autom_i$ transitions from its initial
      state to its final state only when the preceding component is in its
      final state (regardless of the symbol read from $\Sigma$), \ie each
      symbol $(a,q_1^{i-1})$ induces a reset on $q_1^i$ and all the others
      induce resets on $q_0^i$.
  \end{enumerate*}
  This ensures that the automaton $\autom_1 \odot \dots \odot \autom_h$
  reaches its final state $(q_1^1,\dots,q_1^h)$ only after consuming at
  least $h$ symbols from $\Sigma$.
  The lower bound on the size of the smallest cascade for $\lang_h$ is
  proved by Theorem 5.8 of~\cite{journalversionofstacs}.
\end{proof}

\cref{teo:chain:succinctness} shows that, for certain languages, the use of
$\Sigma$-chains can yield representations that are exponentially more
succinct than those based on cascades. This motivates the study of the
expressive power of $\Sigma$-chains, that is, the characterization of the
languages recognizable via $\Sigma$-chains over a given class $A$ of automata.

\subsection{Expressiveness}

We now investigate the expressive power of $\Sigma$-chains, focusing in
particular on the case where the components are automata drawn from a fixed
class $A \in \set{\PureResets, \Resets,\Permutations, \PermutationResets}$ (\ie
pure-reset, reset, permutation, and permutation-reset automata).

First, we compare the $\Sigma$-chain product with the direct product,
showing that the former is strictly more expressive than the latter.

Then, we prove that $\Sigma$-chains of permutation-reset automata capture
exactly the class of regular languages. The result is established via
a \emph{chainification method}, namely, a construction that transforms any
cascade of permutation-reset automata into an equivalent $\Sigma$-chain
of permutation-reset automata. The expressive equivalence with regular
languages then follows from the known equivalence between regular languages
and cascades of permutation-reset automata (\cf~\cref{lemma:star:free}).

We further show that the chainification method extends to the class
$\PureResets$, yielding that $\Sigma$-chains of pure-reset automata are
expressively equivalent to cascades of pure-reset automata.

For the class $\Resets$ of reset automata, we show that $\Sigma$-chains of
reset automata recognize only star-free languages. Whether the converse
holds remains open. The natural proof strategy, namely applying chainification
to cascades of reset automata, which are known to recognize all star-free
languages (\cf~\cref{lemma:star:free}), fails: chainification
produces $\Sigma$-chains of permutation-reset automata rather than reset
automata.

\paragraph*{Comparison with the direct product}
As already discussed in the introduction, the $\Sigma$-chain product can be
viewed both as a restricted variant of the cascade product and as
a generalization of the direct product. In fact, $\Sigma$-chains retain
a key feature of direct products, namely their linear size with respect to
the number of components (\cf~\cref{prop:chain:size}). This naturally
raises the question of their relative expressive power. We show that any
direct product can be simulated by a $\Sigma$-chain product ($\DP[A]
\subseteq \SCH[A]$, for each $A \in \set{\PureResets, \Resets,
\Permutations, \PermutationResets}$) and that the opposite is \emph{not}
possible: we exhibit two simple languages that cannot be recognized by any
direct product of reset (resp.  pure-reset) automata, while they are
definable by a $\Sigma$-chain of reset (resp. pure-reset) automata.

\begin{restatable}{lemma}{lemmadirectprodct}
\label{lemma:directprodct}
  The following inclusions hold:
  \begin{itemize}
    \item $\DP[A] \subseteq \SCH[A]$, for each $A \in \set{\PureResets,
      \Resets, \Permutations, \PermutationResets}$;
    \item $\DP[\Resets] \subsetneq \SCH[\Resets]$ and $\DP[\PureResets]
      \subsetneq \SCH[\PureResets]$.
  \end{itemize}
In particular, the strict inclusions are witnessed by the languages
$\lang_1 = \Sigma^*a\Sigma^*a\Sigma^*$ and $\lang_2 = \Sigma^*aa$, over
$\Sigma = \set{a,b}$, as we have that $\lang_1 \in \SCH[\Resets] \setminus
\DP[\Resets]$ and $\lang_2 \in \SCH[\PureResets] \setminus \DP[\PureResets]$.
\end{restatable}

\paragraph*{Transformation of cascades into $\Sigma$-chains (Chainification)}

We describe the transformation of cascades into equivalent $\Sigma$-chains
starting from the simplest case, that is, the case of cascades of three
automata, since every cascade of height two is trivially a $\Sigma$-chain. 

Let $\cascade=\seq{\autom_1, \autom_2, \autom_3}$ be a cascade of three
automata: we shall construct an equivalent $\Sigma$-chain
$\cascade'=\chain{\autom'_1, \autom'_2, \autom'_3}$ such that
$\lang(\cascade) = \lang(\cascade')$ and if $\autom_i$ is
a permutation-reset (resp., permutation, pure-reset) automaton, then also $\autom'_i$ is
a permutation-reset (resp., permutation, pure-reset) automaton, for each $1\le i\le 3$.

The main challenge in this translation lies in handling the dependence of
the third component $\autom_3$ on the first component $\autom_1$, which is not
permitted in a $\Sigma$-chain. To eliminate this dependence, we encode
the required information into the states of the second component $\autom'_2$.
The construction proceeds as follows:
\begin{enumerate*}[label=(\roman*)]
  \item the second component $\autom'_2$ is defined so as to combine the
    state information of $\autom_1$ and $\autom_2$;
  \item the third component $\autom'_3$ then reads the state of
    $\autom'_2$, from which it can recover the information about both
    $\autom_1$ and $\autom_2$.
\end{enumerate*}
As stated by the following Lemma, if the second component $\autom_2$ is
permutation-reset (resp. permutation, pure-reset), the new component $\autom'_2$ can be
defined to be permutation-reset (resp. permutation, pure-reset).

\begin{restatable}{lemma}{cascadethreetochainthree}
\label{lemma:cascade3:to:chain3}
  Let $\cascade = \seq{\autom_1, \autom_2, \autom_3}$ be a cascade of
  permutation-reset (resp. permutation, pure-reset) automata. There exists
  a $\Sigma$-chain $\cascade' = \chain{\autom'_1, \autom'_2, \autom'_3}$ of
  permutation-reset (resp. permutation, pure-reset) automata such that $\lang(\cascade)
  = \lang(\cascade')$.
\end{restatable}
\begin{proof}[Proof (sketch)]
  Define $\autom'_1$ as $\autom_1$. To ensure that $\autom'_2$ carries all the information needed by $\autom'_3$, which reads only symbols in $\Sigma$ and states of $\autom'_2$, we define the state set $S'$ of $\autom'_2$ as the Cartesian product of the state sets of $\autom_1$ and $\autom_2$.
Furthermore, $\autom'_2$ must satisfy the following condition: for every state $(q, s) \in S'$ and every transition labeled $(a, q)$, where $q$ matches the first component of the current state, the transition leads to a state $(q', s')$ such that $\delta_1(q,a) = q'$ and $\delta_2(s,(a,q))
  = s'$, \ie
  \begin{align}
    \forall (a,q) \in \Sigma \times Q \suchdot ( \delta'_2((q,s),(a,q))
    = (q',s') \ \Leftrightarrow \ \delta_1(q,a) = q' \land
    \delta_2(s,(a,q)) = s' )
  \end{align}
  Moreover, it is possible to extend the definition of $\delta'_2$ on all other input pairs 
  (\ie the pairs $(a, q'')$ with $q'' \neq q$) in order to guarantee that $\tau^{\autom'_2}_{(a,q)}$
  is a permutation (when $\tau^{\autom_2}_{(a,q)}$ is a permutation) or
  a reset (when $\tau^{\autom_2}_{(a,q)}$ is a reset).
  The automaton $\autom'_3$ is defined as $\autom_3$ by transforming each
  transition $\tau_{(a,q,s)}$ into $\tau_{(a,(q,s))}$.
  Crucially, all states of $\autom'_1 \odot \autom'_2 \odot \autom'_3$ of
  type $(q,(q',s),t)$ with $q\neq q'$ are not reachable from the initial
  state, having that the language of $\autom'_1 \odot \autom'_2 \odot
  \autom'_3$ is equal to the language of $\autom_1 \circ \autom_2 \circ
  \autom_3$.
\end{proof}


\begin{figure}[t!]
\centering
\begin{subfigure}{.45\textwidth}
  \centering
  \begin{subfigure}{\textwidth}
    \centering
    \begin{tikzpicture}[shorten >=1pt,node distance=2.5cm,on grid,auto,
      scale=0.8, every node/.style={scale=0.8}]
      \tikzstyle{every state}=[fill={rgb:black,1;white,10}]
        \node[state,initial,initial text={$\autom_1$: }] (q_0) {$q_0$};
        \node[state,accepting] (q_1) [right of=q_0] {$q_1$};

        \path[->]
        (q_0) edge [bend left] node {$a$} (q_1)
        (q_1) edge [bend left] node {$b$} (q_0)
        (q_0) edge [loop above] node {$b$} ()
        (q_1) edge [loop above] node {$a$} ();
    \end{tikzpicture}
  \end{subfigure}\vspace{4pt}
  \begin{subfigure}{\textwidth}
    \centering
    \begin{tikzpicture}[shorten >=1pt,node distance=2.5cm,on grid,auto,
      scale=0.8, every node/.style={scale=0.8}]
      \tikzstyle{every state}=[fill={rgb:black,1;white,10}]
        \node[state,initial,initial text={$\autom_2$: }] (s_0) {$s_0$};
        \node[state,accepting] (s_1) [right of=s_0] {$s_1$};

        \path[->]
        (s_0) edge [bend left] node {$(a,q_1)$} (s_1)
        (s_1) edge [bend left] node {$(b,q_0)$} (s_0)
        (s_0) edge [loop above] node[left,align=center]{$(b,q_0)$\\$id_2$} ()
        (s_1) edge [loop above] node[right,align=center]{$(a,q_1)$\\$id_2$} ();
    \end{tikzpicture}
  \end{subfigure}\vspace{4pt}
  \begin{subfigure}{\textwidth}
    \centering
    \begin{tikzpicture}[shorten >=1pt,node distance=2.5cm,on grid,auto,
      scale=0.8, every node/.style={scale=0.8}]
      \tikzstyle{every state}=[fill={rgb:black,1;white,10}]
        \node[state,initial,initial text={$\autom_3$: }] (t_0) {$t_0$};
        \node[state,accepting] (t_1) [right of=t_0] {$t_1$};

        \path[->]
        (t_0) edge [] node {$(a,q_1,s_1)$} (t_1)
        (t_0) edge [loop above] node {$id_3$} ()
        (t_1) edge [loop above] node[right,align=center] {$(a,q_1,s_1)$\\$id_3$} ();
    \end{tikzpicture}
  \end{subfigure}
\end{subfigure}
\begin{subfigure}{.45\textwidth}
  \centering
  \begin{subfigure}{\textwidth}
    \centering
    \begin{tikzpicture}[shorten >=1pt,node distance=2.6cm,on grid,auto,
      scale=0.8, every node/.style={scale=0.8},
      every loop/.style={min distance=5mm,looseness=5}]
      \tikzstyle{every state}=[fill={rgb:black,1;white,10},
        shape=rectangle, rounded corners=4pt,
        inner sep=4pt, minimum size=0pt]
        \node[state,initial,initial text={$\autom'_2$: }] (x_00) {$(q_0,s_0)$};
        \node[state] (x_10) [right of=x_00] {$(q_1,s_0)$};
        \node[state] (x_01) [below of=x_00] {$(q_0,s_1)$};
        \node[state,accepting] (x_11) [right of=x_01] {$(q_1,s_1)$};

        \path[->]
        (x_00) edge[loop above] node[left] {$(b,q_0)$} ()
        (x_00) edge[bend left=45] node[yshift=-1mm] {$(a,q_0),(b,q_1)$} (x_10)
        (x_00) edge[bend left=20] node[sloped,yshift=-1mm] {$(a,q_1)$} (x_11)
        (x_10) edge[] node[above,yshift=-1mm] {$(b,q_1)$} 
                      node[below,yshift=+1mm] {$(\Sigma,q_0)$} (x_00)
        (x_10) edge[bend left] node[] {$(a,q_1)$} (x_11)
        (x_01) edge[bend left] node[] {$(b,q_0)$} (x_00)
        (x_01) edge[bend right=45] node[below] {$(\Sigma,q_1),(a,q_0)$} (x_11)
        (x_11) edge[loop right] node[below,yshift=-2mm] {$(a,q_1)$} ()
        (x_11) edge[bend left=20] node[sloped,below] {$(b,q_0)$} (x_00)
        (x_11) edge[] node[above,yshift=-1mm] {$(b,q_1)$} 
                      node[below,yshift=+1mm] {$(a,q_0)$} (x_01);
    \end{tikzpicture}
  \end{subfigure}\vspace{4pt}
  \begin{subfigure}{\textwidth}
    \centering
    \begin{tikzpicture}[shorten >=1pt,node distance=2.5cm,on grid,auto,
      scale=0.8, every node/.style={scale=0.8}]
      \tikzstyle{every state}=[fill={rgb:black,1;white,10}]
        \node[state,initial,initial text={$\autom'_3$: }] (t_0) {$t_0$};
        \node[state,accepting] (t_1) [right of=t_0] {$t_1$};

        \path[->]
        (t_0) edge [] node {$(a,(q_1,s_1))$} (t_1)
        (t_0) edge [loop above] node[] {$id'_3$} ()
        (t_1) edge [loop above] node[right,align=center] {$(a,(q_1,s_1))$\\$id'_3$} ();
    \end{tikzpicture}
  \end{subfigure}
\end{subfigure}
\caption{%
  On the left, the cascade $\seq{\autom_1,\autom_2,\autom_3}$ over
  $\Sigma=\set{a,b}$, where $Q$, $S$, $T$ are the sets of states of
  $\autom_1$, $\autom_2$, $\autom_3$, respectively, $id_2
  = \set{(a,q_0),(b,q_1)}$ and $id_3 = (\Sigma \times Q \times S) \setminus
  \set{(a,q_1,s_1)}$.
  On the right, the automata $\autom'_2$ (with set of states $S'$) and
  $\autom'_3$ of the $\Sigma$-chain obtained through the chainification
  technique of~\cref{lemma:cascade3:to:chain3}, where $id'_3 = (Q \times
  S') \setminus \set{(a,(q_1,s_1))}$.
  The cascade on the left is equivalent to the $\Sigma$-chain
  $\chain{\autom_1,\autom'_2,\autom'_3}$.
}
\label{fig:chainification}
\end{figure}
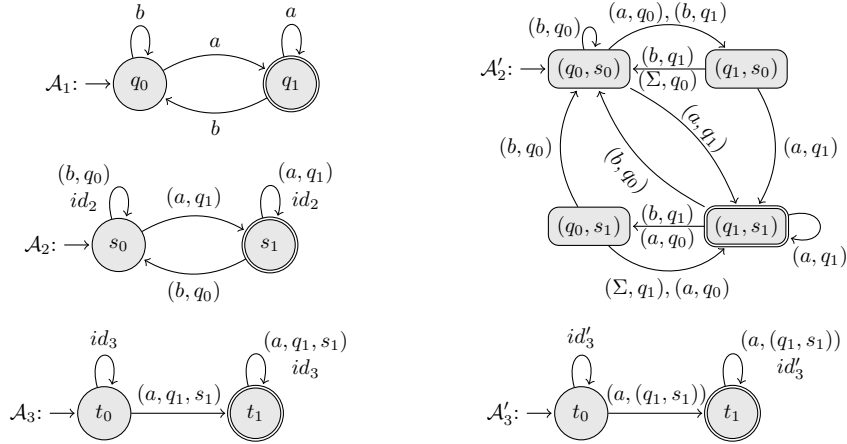

\Cref{fig:chainification} illustrates an example of chainification for
a cascade $\seq{\autom_1, \autom_2, \autom_3}$ over the alphabet $\Sigma
= \set{a, b}$ (left side of~\cref{fig:chainification}).
The components $\autom'_2$ and $\autom'_3$ of the equivalent $\Sigma$-chain
are shown on the right side of~\cref{fig:chainification}. It is important
to note that all transitions $(a,q)$ in $\autom'_2$ can be defined as
either permutations or reset functions. Furthermore, we observe that
symbols acting as identities in $\autom_2$ may become non-identity
permutations in $\autom'_2$; this is the case, for instance, for the symbol
$(a, q_0)$.  Indeed, while $\autom_2$ is a reset automaton, $\autom'_2$ is
a permutation-reset automaton.

\Cref{lemma:cascade3:to:chain3} can be generalized to deal with cascades of
unbounded height as stated by the following theorem.

\begin{restatable}[Chainification Theorem]{theorem}{chainification}
\label{teo:chainification}
  Let $\cascade = \seq{\autom_1, \dots, \autom_n}$ be a cascade of
  permutation-reset (resp. permutation, pure-reset) automata. There exists
  a $\Sigma$-chain $\cascade' = \chain{\autom'_1, \dots, \autom'_n}$ of
  permutation-reset (resp. permutation, pure-reset) automata such that $\lang(\cascade)
  = \lang(\cascade')$ \roberto{and $|\cascade'| \in \mathcal{O}(|\cascade|)$}.
\end{restatable}

\paragraph*{Implications on the expressive power of $\Sigma$-chains}

\Cref{prop:chain:to:cascades,teo:chainification} imply that $\Sigma$-chains
are expressively equivalent to cascades for the class $\PureResets$ of
pure-reset automata, for the class $\Permutations$ of permutation automata and for the class $\PermutationResets$ of
permutation-reset automata. Moreover, $\Sigma$-chains of permutation-reset
automata characterize the class of regular languages
(\cf~\cref{lemma:star:free}).

As for the class $\Resets$, we have that if a language $\lang$ is
recognized by a $\Sigma$-chain of reset automata then $\lang$ is also
recognized by a cascade of reset automata (\cref{prop:chain:to:cascades})
and $\lang$ is star-free (\cref{lemma:star:free}).

\begin{restatable}{corollary}{corocascadeschains}
\label{coro:cascades:chains}
It holds that:
\begin{itemize}
    \item $\SCH[\PermutationResets] = \C[\PermutationResets]$ and thus
      a language $\lang$ is regular if and only if $\lang \in
      \SCH[\PermutationResets]$;
    \item $\SCH[\PureResets] = \C[\PureResets]$;
    \item $\SCH[\Resets] \subseteq \C[\Resets]$ and thus if $\lang \in
      \SCH[\Resets]$ then $\lang$ is star-free;
    \item $\SCH[\Permutations] = \C[\Permutations]$.
\end{itemize}
\end{restatable}

%
%
%

\section{$\Sigma$-Chains of Reset Automata}
\label{sec:chains:reset}

In this section, we investigate $\Sigma$-chains of reset automata and the languages they recognize ($\SCH[\Resets]$). First, we
analyze how imposing structural constraints on $\Sigma$-chains --
specifically, bounding their \emph{height} (\ie the number of components)
and \emph{width} (\ie the number of states per component) -- affects the
class of recognizable languages.

Second, we compare $\Sigma$-chains of reset automata with 
Piecewise Testable~\cite{simon2005piecewise} and \rtrivial
languages~\cite{brzozowski1980languages}, showing that every language in
these classes can be expressed as a Boolean combination and as a finite
union of $\Sigma$-chains of reset automata, respectively. Moreover, we
establish linear upper bounds on the size of such $\Sigma$-chain
representations.

Finally, we address the \emph{pastification} problem for $\Sigma$-chains of reset automata, namely, the transformation of these structures into
equivalent formulas in pure past Linear Temporal Logic. We show a direct transformation yielding a formula of quadratic size when represented as a DAG (\cf dag-complexity \cite{journalversionofstacs}). 

\subsection{Height and Width Hierarchies}
\label{section:height}
First, we show that increasing the height of a $\Sigma$-chain of reset
automata (\ie the number of its components) yields an infinite hierarchy of
languages: for every level $h$, there exists a language at level $h$ that
is definable by a $\Sigma$-chain of $h$ reset automata but is not
recognizable by any $\Sigma$-chain of height $h-1$. This result is
analogous to what is known for cascades of reset
automata (\cf Height-Hierarchy Lemma of~\cite{journalversionofstacs}).

\begin{restatable}{proposition}{heighthierarchylemma}
\label{height:hierarchy:lemma}
  Let $\Sigma$ be a finite alphabet.
  For each $h > 0$, let $\lang_h = \Sigma^{h} \cdot \Sigma^*$. It holds that:
  \begin{itemize}
    \item $\lang_h$ is definable by a $\Sigma$-chain of (two-state) reset
      automata $\cascade_h$ of height $h$;
    \item $\lang_h$ cannot be defined by any $\Sigma$-chain of reset
      automata of height less than $h$.
  \end{itemize}
\end{restatable}

As shown by  \cref{height:hierarchy:lemma} and by the Height-Hierarchy Lemma of~\cite{journalversionofstacs}, both $\Sigma$-chains and cascades of reset automata admit a height hierarchy: for every $h \geq 2$ there exist languages definable by a product of height $h$ that cannot be defined by any product of smaller height.
In contrast, such a hierarchy collapses for direct products, as formalized by the following lemma.

\begin{restatable}{lemma}{lemmanoheight}
Let $\Sigma = \set{a,b}$. Every language definable by a direct product of reset automata over $\Sigma$ is also definable by a direct product of two reset automata over $\Sigma$.
\end{restatable}

We now investigate the effect of increasing the number of states in each
component of a $\Sigma$-chain of reset automata. The following proposition
shows that this does not yield an infinite hierarchy of languages. In
particular, every $\Sigma$-chain of reset automata is equivalent to
a $\Sigma$-chain of the same height in which the number of states in each
component depends only on the size of the alphabet $\Sigma$.

\begin{restatable}{proposition}{lemmaGenWidthCollpase}
\label{lemma:GenWidthCollpase}
  Let $\Sigma$ be a finite alphabet with $|\Sigma| = k \geq 2$ and let
  $\lang \subseteq \Sigma^*$ be a language accepted by a $\Sigma$-chain of
  reset automata $\cascade = \chain{\autom_1,\autom_2,\dots,\autom_h}$, for some $h>0$.
  %
  It holds that $\lang$ is accepted by some $\Sigma$-chain of reset
  automata $\cascade' = \chain{\autom'_1,\autom'_2,\dots,\autom'_h}$ where
  $\autom'_i$ has at most $f(i)=\frac{k^{i+1}-1}{k-1}$ states (for $i
  = 1,\dots,h$).
\end{restatable}

The $f(i) = \frac{k^{i+1}-1}{k-1}$ bound on the number of states in the
$i$-th component of a $\Sigma$-chain, established in the previous
proposition, is -- owing to the restriction that each component depends
only on its immediate predecessor -- significantly smaller than the
corresponding bound for cascades of reset automata (\cf Width-Collapse
Lemma in~\cite{journalversionofstacs}), namely $g^{i-1}(k+1)$ with $g(x)
= x^2 - x + 1$.

\subsection{Comparison with Piecewise Testable Languages}
Piecewise testable languages were introduced by
Simon~\cite{simon2005piecewise} as a natural class of formal languages
whose membership can be decided by inspecting only the scattered
subsequences of bounded length occurring in a given word, regardless of the
positions at which they appear.  Intuitively, two words are
indistinguishable with respect to a piecewise testable language whenever
they share the same set of scattered subsequences up to a fixed length.
This class has since attracted considerable attention in the theory of
formal languages, owing in part to Simon's elegant algebraic
characterization, which identifies piecewise testable languages with
precisely those regular languages whose syntactic monoid is
$\mathcal{J}$-trivial~\cite{simon2005piecewise}, and in part to their
natural logical characterization as exactly the languages definable by
Boolean combinations of existential first-order sentences over
words~\cite{thomas1982classifying}. 

\begin{definition}[Piecewise Testable Languages~\cite{simon2005piecewise,pin1986varieties}]
\label{def:piecewise}
  Let $\Sigma$ be a finite alphabet and let $\sigma = \seq{a_1,\dots,a_n}
  \in \Sigma^n$ be a word of length $n$, for some $n>0$. We define
  $\lang_\sigma \subseteq\Sigma^*$ as the language $\lang_\sigma = \Sigma^*
  a_1 \Sigma^* \dots \Sigma^* a_n \Sigma^*$.  A language $\lang$  is
  \emph{$k$-piecewise testable} if $\lang$  is a Boolean combination of
  languages $\lang_{\sigma_i}$ with $\sigma_i \in \Sigma^{\leq k}$.
  A language $\lang$ is \emph{piecewise testable} if it is $k$-piecewise
  testable for some $k>0$.
\end{definition}

To show that $k$-piecewise testable languages are recognizable by Boolean
combinations of (languages recognized by) $\Sigma$-chains of $k$ reset
automata, we first establish a generalization of the fact that, for every
$n > 0$ and every $\sigma \in \Sigma^+$, there exists a $\Sigma$-chain of
$n$ reset automata recognizing $\lang_\sigma$. In particular, the following
lemma shows that $\Sigma$-chains of reset automata can recognize languages
of the form $\Sigma^* \sigma_1 \Sigma^* \cdots \Sigma^* \sigma_n \Sigma^*$,
where each $\sigma_i$ is not restricted to a single symbol of the alphabet
$\Sigma$ (as in~\cref{def:piecewise}), but may be an arbitrary word. This
is achieved using a number of components equal to the sum of the lengths of
all the $\sigma_i$.

\begin{restatable}{lemma}{lemmagenpiecewiselanguages}
\label{lemma:gen:piecewise:languages}
  Let $\Sigma$ be a finite alphabet and let $\sigma_1,\dots,\sigma_m$ be
  $m$ words in $\Sigma^+$, for some $m>0$. Let
  $\lang_{\sigma_1,\dots,\sigma_m}$ be the language
  $\Sigma^* \sigma_1 \Sigma^* \dots
  \Sigma^* \sigma_m \Sigma^*$. There exists a $\Sigma$-chain $\cascade$ of
  (two-states) reset automata such that:
  \begin{itemize}
    \item $\lang(\cascade) = \lang_{\sigma_1,\dots,\sigma_m}$;
    \item the height of $\cascade$ is $|\sigma_1|+\dots +|\sigma_m|$.
  \end{itemize}
\end{restatable}
\begin{proof}[Proof (sketch)]
  We first construct, for each $1\le i\le m$, a $\Sigma$-chain for the
  language $\Sigma^* \cdot \sigma_i \cdot \Sigma^*$. Let $\sigma_i = a_1
  a_2 \dots a_n$. The $\Sigma$-chain consists of $n$ two-state reset
  automata: for each $1\le j\le n$, the component $\autom_j$ transitions
  from its initial to its noninitial state (by means of a reset function)
  if and only if the current symbol is $a_j$ and the previous component (if
  any) is in its noninitial state.  Moreover, the last component
  ($\autom_n$) must remain in its current state (by means of an identity
  function) for all combinations of symbols in $\Sigma\setminus\set{a_n}$
  and states of $\autom_{n-1}$, in order to ensure that the last $\Sigma^*$
  in $\Sigma^* \cdot \sigma_i \cdot \Sigma^*$. Finally, we set the final
  states of $\autom_j$ to be the singleton containing only its
  noninitial state if $j=n$, otherwise we define it as the set of all the
  two states of $\autom_j$.

  We can then compose the $\Sigma$-chains for each $\sigma_i$ obtained in
  this fashion to obtain a $\Sigma$-chain of $\sum_{i=1}^{m} |\sigma_i|$
  components recognizing $\Sigma^* \sigma_1 \Sigma^* \dots \Sigma^*
  \sigma_m \Sigma^*$.
\end{proof}

As a straightforward corollary of~\cref{lemma:gen:piecewise:languages},
applied to the case where $|\sigma_i| = 1$ for every $1 \le i \le n$, we
obtain the following result.

\begin{restatable}{corollary}{lemmakpiecewise}
\label{lemma:k:piecewise}
  Let $\Sigma$ be a finite alphabet and let $\lang\subseteq \Sigma^*$. If
  $\lang$ is a $k$-piecewise testable language, then $\lang$ is a Boolean
  combination of languages, each of which is definable by a $\Sigma$-chain
  of (two-states) reset automata of height at most $k$.
\end{restatable}

\subsection{Comparison with \rtrivial Languages}

\rtrivial languages are  those regular languages whose syntactic monoid is
$\mathcal{R}$-trivial and are dual to $\mathcal{L}$-trivial
languages~\cite{brzozowski1980languages}.  Beyond their algebraic
characterizations, \rtrivial languages are deeply connected to automata
theory. In particular, it holds that a language $\lang$ is \rtrivial if and
only if the minimum automaton accepting $\lang$ is partially
ordered~\cite{brzozowski1980languages}.  Surprisingly, there is also
a characterization in terms of the cascade product. A language $\lang$ is
\rtrivial if and only if the minimum automaton accepting $\lang$ is covered
by a cascade product of half-resets where an half-reset is a two-state
semiautomaton $\autom = \seq{\Sigma, \set{q_0, q_1}, \delta}$ such that
$\Sigma$ can be partitioned in two non-empty sets $B_0$ and $B_1$ such that
$\delta(q_0,a_0) = q_0$ for every $a_0 \in B_0$, $\delta(q_0,a_1) = q_1$
for every $a_1 \in B_1$ and $\delta(a,q_1) = q_1$ for every $a \in \Sigma$
~\cite{brzozowski1980languages}.

As with piecewise-testable languages, we start by defining \rtrivial languages in terms of
regular expressions.

\begin{definition}[\rtrivial Languages~\cite{brzozowski1980languages,pin1986varieties}]
\label{def:rtrivial}
  Let $\Sigma$ be a finite alphabet. A language $\lang \subseteq \Sigma^*$
  is \rtrivial if $\lang$ is a disjoint union of languages of the form
  $I_0^*a_1I_1^*a_2I_2^*\dots a_nI_n^*$ where:
  \begin{enumerate*}[label=(\roman*)]
    \item $n \geq 0$;
    \item $a_i \in \Sigma$, for each $1\le i\le n$;
    \item $I_i\subseteq \Sigma \setminus \set{a_{i+1}}$, for each $0\le i<
      n$; and
    \item $I_n \subseteq \Sigma$.
  \end{enumerate*}
\end{definition}

To show that \rtrivial languages are recognizable by finite unions of
(languages recognized by) $\Sigma$-chains of reset automata, we first
establish a more general result.  In particular, we show that
$\Sigma$-chains of reset automata can recognize languages of the form
$I_0^*B_1I_1^*B_2\dots B_nI_n^*$ where each $I_i$ and each $B_j$ is
a subset of $\Sigma$ without further restrictions.  Moreover, the obtained
$\Sigma$-chain is of height $n+1$. 

\begin{restatable}{lemma}{lemmagentriviallanguages}
\label{lemma:gen:trivial:languages}
  Let $\Sigma$ be a finite alphabet. For some $m>0$, let $I_0, \dots, I_m,
  B_1,\dots,B_m$ be subsets of $\Sigma$.  Let
  $\lang_{B_1,\dots,B_m}^{I_0,\dots,I_m} = I_0^* B_1 I_1^* \dots I_{m-1}^*
  B_m I_m^*$. There exists a $\Sigma$-chain of two-states reset automata
  $\cascade$ such that:
  \begin{itemize}
    \item $\lang(\cascade) = \lang_{B_1,\dots,B_m}^{I_0,\dots,I_m}$;
    \item $\cascade$ has height $m+1$.
  \end{itemize}
\end{restatable}
\begin{proof}[Proof (sketch)]
  The proof follows an approach similar to the proof
  of~\cref{lemma:gen:piecewise:languages}.
\end{proof}

Applying~\cref{lemma:gen:trivial:languages} to the case where $|B_i| = 1$
(for every $1 \le i \le n$) and $I_i\subseteq \Sigma \setminus
\set{a_{i+1}}$ (for $0 \leq i \leq  n-1$) leads to the following result.

\begin{restatable}{corollary}{lemmartrivial}
\label{lemma:r:trivial}
Let $\Sigma$ be a finite alphabet. Let $\lang\subseteq \Sigma^*$ be an \rtrivial language. $\lang$ is recognized by a finite union of $\Sigma$-chains of two-states reset automata.
\end{restatable}

\subsection{$\Sigma$-Chains of reset automata and \ppLTL}

In \cite{DBLP:conf/focs/MalerP90}, Maler proposed a procedure to compile cascades of reset semi-automata into equivalent formulae of \ppLTL, namely formulae of \LTL with only past temporal operators. We call such a procedure a \textit{pastification} procedure. In \cite{journalversionofstacs}, this pastification procedure was applied to cascades of reset automata, providing bounds on the complexity of the resulting formula $\pastify{(\cascade)}$ (where $\cascade$ is a cascade of reset automata) under two measures:
\begin{itemize}
\item The tree-complexity (denoted $\treesize{\pastify{(\cascade)}}$): the size of the formula $\pastify{(\cascade)}$;
\item The dag-complexity (denoted $\dagsize{\pastify{(\cascade)}}$): the number of syntactically distinct sub-formulas of $\pastify{(\cascade)}$. The dag-complexity measures the space needed to represent the formula using the sub-term sharing technique.
\end{itemize}
Theorem 8.4 of \cite{journalversionofstacs} provides bounds on both measures:
while the tree-complexity of the obtained \ppLTL formula may be exponential in the size of $\cascade$, the dag-complexity has only a linear bound.

Given a $\Sigma$-chain $\cascade$ of $n$ reset automata, \Cref{prop:chain:to:cascades} allows us to first transform it into an equivalent cascade $\cascade'$ of $n$ reset automata, whose size may however be exponential in $n$. Applying Theorem 8.4 of \cite{journalversionofstacs} then yields an equivalent \ppLTL formula, but whose dag-complexity is exponential in $n$.
A direct translation, however, achieves a significantly better bound. Indeed, a straightforward extension of the techniques of \cite{journalversionofstacs} yields a transformation from a $\Sigma$-chain of $n$ reset automata into an equivalent \ppLTL formula with only quadratic dag-complexity.

\begin{restatable}{theorem}{pastification}
\label{pastification}
  Let $\AP$ be a set of atomic propositions and let $\cascade
  = \chain{\autom_1, \dots, \autom_n}$ be a $\Sigma$-chain of reset automata over
  $\Sigma = 2^{\AP}$. There exists a \ppLTL formula $\pastify(\cascade)$ such that:
  \begin{itemize}
      \item $\lang(\pastify(\cascade)) = \lang(\cascade) \setminus
        \set{\varepsilon}$;
      \item $\dagsize{\pastify(\cascade)} \in \mathcal{O}(|\cascade|\cdot ( |\AP| + n)) \in  \mathcal{O}(|\cascade|\cdot( |\AP| + |\cascade|))$. 
      \end{itemize}
\end{restatable}

By \Cref{teo:chain:succinctness} and \Cref{pastification}, the language $\Sigma^{h}\Sigma^*$ admits a  pastification of dag-complexity $\mathcal{O}(h \cdot |\Sigma| \cdot (|\AP|+h))$ after a prior transformation into a $\Sigma$-chain of reset automata.
In contrast, translating the same language into a cascade of reset automata yields a $\ppLTL$ formula of dag-complexity $\mathcal{O}(2^h \cdot |\Sigma| \cdot |\AP|)$.
Therefore, $\Sigma$-chains can produce exponentially more succinct pastifications than cascades of reset automata.

\section{Conclusions and Future Work}
\label{sec:conclusions}

In this paper, we have introduced the notion of $\Sigma$-chain and the associated $\Sigma$-chain product operation, in which the component automata are combined in a feed-forward fashion, with each component depending solely on the initial alphabet $\Sigma$ and the states of the immediately preceding component.
We have investigated the expressive power of $\Sigma$-chains, establishing that every regular language can be decomposed into a $\Sigma$-chain of permutation-reset automata, and that every cascade of permutation-reset (pure-reset, and permutation) automata can be translated into an equivalent $\Sigma$-chain of the same class and height. While an analogous result for reset automata remains open, we have shown that (Boolean combinations of) $\Sigma$-chains of reset automata capture piecewise testable and \rtrivial languages. Finally, we addressed the pastification problem, providing a translation of $\Sigma$-chains of reset automata into \ppLTL formulas of quadratic dag-complexity.

A natural direction for future work is to determine whether every star-free language -- that is, a language recognizable by a cascade of reset automata -- is also recognizable by some $\Sigma$-chain of reset automata, which requires a complete characterization of the languages recognizable by $\Sigma$-chains.
A further open problem is whether, as in the case of cascades of reset automata, every
$\Sigma$-chain of reset automata can be translated into an equivalent one in which each component has only two states, possibly at the cost of increasing the height.

Several questions remain open regarding the pastification problem. It would be interesting to investigate whether the quadratic dag-complexity bound is tight, and whether the structural properties of $\Sigma$-chains can be exploited to obtain more efficient pastification procedures for \LTL formulas in general.

\roberto{Finally, the cascade product admits the wreath product as its algebraic counterpart. Identifying and studying the algebraic structure corresponding to the $\Sigma$-chain product is an interesting direction for future research.}

\newpage

\bibliography{biblio}

@book{gecseg2012products,
  title={{Products of automata}},
  author={G{\'e}cseg, Ferenc},
  volume={7},
  year={2012},
  publisher={Springer Science \& Business Media}
}

@article{imreh1977alphai,
  title={{On {$\alpha_i$}-products of automata}},
  author={Imreh, Bal{\'a}zs},
  journal={Acta Cybernetica},
  volume={3},
  number={4},
  pages={301--307},
  year={1977},
  publisher={University of Szeged}
}

@article{domosi1983vi,
  title={{On {$v_i$}-products of automata}},
  author={D{\"o}m{\"o}si, P{\'a}l and Imreh, Bal{\'a}zs},
  journal={Acta Cybernetica},
  volume={6},
  number={2},
  pages={149--162},
  year={1983},
  publisher={University of Szeged}
}

@article{gecseg1991,
title = {{On {$\alpha_0$-$v_1$}-products of automata}},
journal = {Theoretical Computer Science},
volume = {80},
number = {1},
pages = {35-51},
year = {1991},
issn = {0304-3975},
doi = {10.1016/0304-3975(91)90204-F},
author = {F. Gécseg and H. Jürgensen}
}

@inproceedings{gelderie2011classifying,
  title={{Classifying regular languages via cascade products of automata}},
  author={Gelderie, Marcus},
  booktitle={International Conference on Language and Automata Theory and Applications},
  pages={286--297},
  year={2011},
  organization={Springer}
}

@book{pin1986varieties,
  title={Varieties of formal languages},
  author={Pin, Jean Eric},
  year={1986},
  publisher={Plenum Publishing Co.}
}

@inproceedings{DBLP:conf/focs/MalerP90,
  author       = {Oded Maler and
                  Amir Pnueli},
  title        = {{Tight Bounds on the Complexity of Cascaded Decomposition of Automata}},
  booktitle    = {31st Annual Symposium on Foundations of Computer Science, St. Louis,
                  Missouri, USA, October 22-24, 1990, Volume {II}},
  pages        = {672--682},
  publisher    = {{IEEE} Computer Society},
  year         = {1990},
  doi          = {10.1109/FSCS.1990.89589},
  bibsource    = {dblp computer science bibliography, https://dblp.org}
}

@inproceedings{DBLP:conf/birthday/Maler10,
  author       = {Oded Maler},
  editor       = {Zohar Manna and
                  Doron A. Peled},
  title        = {{On the Krohn-Rhodes Cascaded Decomposition Theorem}},
  booktitle    = {Time for Verification, Essays in Memory of Amir Pnueli},
  series       = {Lecture Notes in Computer Science},
  volume       = {6200},
  pages        = {260--278},
  publisher    = {Springer},
  year         = {2010},
  doi          = {10.1007/978-3-642-13754-9_12},
  bibsource    = {dblp computer science bibliography, https://dblp.org}
}

@article{krohn1965algebraic,
  title={{Algebraic theory of machines. I. Prime decomposition theorem for finite semigroups and machines}},
  author={Krohn, Kenneth and Rhodes, John},
  journal={Transactions of the American Mathematical Society},
  volume={116},
  pages={450--464},
  year={1965}
}

@article{DBLP:journals/corr/abs-2010-16235,
  author       = {Karl{-}Heinz Zimmermann},
  title        = {{On Krohn-Rhodes theory for semiautomata}},
  journal      = {CoRR},
  volume       = {abs/2010.16235},
  year         = {2020},
  url          = {https://arxiv.org/abs/2010.16235},
  eprinttype    = {arXiv},
  eprint       = {2010.16235},
  bibsource    = {dblp computer science bibliography, https://dblp.org}
}

@inproceedings{DBLP:conf/stacs/BorelliGMM25,
  author       = {Roberto Borelli and
                  Luca Geatti and
                  Marco Montali and
                  Angelo Montanari},
  editor       = {Olaf Beyersdorff and
                  Michal Pilipczuk and
                  Elaine Pimentel and
                  Kim Thang Nguyen},
  title        = {{On Cascades of Reset Automata}},
  booktitle    = {42nd International Symposium on Theoretical Aspects of Computer Science,
                  {STACS} 2025, Jena, Germany, March 4-7, 2025},
  series       = {LIPIcs},
  pages        = {20:1--20:22},
  publisher    = {Schloss Dagstuhl - Leibniz-Zentrum f{\"{u}}r Informatik},
  year         = {2025},
  doi          = {10.4230/LIPICS.STACS.2025.20},
  bibsource    = {dblp computer science bibliography, https://dblp.org}
}

@article{journalversionofstacs,
  author  = {Roberto Borelli and Luca Geatti and Marco Montali and Angelo Montanari},
  title   = {{Cascades of Reset Automata}},
  note    = {Submitted to ACM Transactions on Computational Logic},
  year    = {2026}
}

@inproceedings{DBLP:conf/overlay/Geatti25,
  author       = {Luca Geatti},
  editor       = {Angelo Montanari and
                  Andrea Orlandini and
                  Nicola Saccomanno and
                  Stefano Tonetta},
  title        = {{Automata Cascades for Model Checking}},
  booktitle    = {Short Paper Proceedings of the 7th International Workshop on Artificial
                  Intelligence and Formal Verification, Logic, Automata, and Synthesis,
                  {OVERLAY} 2025, Bologna, Italy, October 26, 2025},
  series       = {{CEUR} Workshop Proceedings},
  pages        = {49--57},
  publisher    = {CEUR-WS.org},
  year         = {2025},
  url          = {https://ceur-ws.org/Vol-4142/paper6.pdf},
  bibsource    = {dblp computer science bibliography, https://dblp.org}
}

@inproceedings{DBLP:conf/kr/ArtaleGGMM23,
  author       = {Alessandro Artale and
                  Luca Geatti and
                  Nicola Gigante and
                  Andrea Mazzullo and
                  Angelo Montanari},
  editor       = {Pierre Marquis and
                  Tran Cao Son and
                  Gabriele Kern{-}Isberner},
  title        = {{A Singly Exponential Transformation of LTL[X, {F]} into Pure Past
                  {LTL}}},
  booktitle    = {Proceedings of the 20th International Conference on Principles of
                  Knowledge Representation and Reasoning, {KR} 2023, Rhodes, Greece,
                  September 2-8, 2023},
  pages        = {65--74},
  year         = {2023},
  doi          = {10.24963/KR.2023/7},
  bibsource    = {dblp computer science bibliography, https://dblp.org}
}

@inproceedings{simon2005piecewise,
  title={Piecewise testable events},
  author={Simon, Imre},
  booktitle={Automata Theory and Formal Languages: 2nd GI Conference Kaiserslautern, May 20--23, 1975},
  pages={214--222},
  year={2005},
  organization={Springer}
}

@article{brzozowski1980languages,
  title={Languages of {R}-trivial monoids},
  author={Brzozowski, Janusz A. and Fich, Faith E},
  journal={Journal of Computer and System Sciences},
  volume={20},
  number={1},
  pages={32--49},
  year={1980},
  publisher={Elsevier}
}

@article{thomas1982classifying,
  title={Classifying regular events in symbolic logic},
  author={Thomas, Wolfgang},
  journal={Journal of Computer and System Sciences},
  volume={25},
  number={3},
  pages={360--376},
  year={1982},
  publisher={Elsevier}
}

\ifextended
\clearpage
\appendix

\section{Proofs of \cref{sec:chains}}

\chaintocascades*
\begin{proof}
For $n=2$ the claim is trivially satisfied since a $\Sigma$-chain $\chain{\autom_1,  \autom_2}$ is also a cascade. We now consider the case $n\geq 3$. We build the cascade $\cascade'$ as follows:
\begin{itemize}
    \item $\autom'_1 \coloneqq \autom_1$;
    \item $\autom'_2 \coloneqq \autom_2$;
    \item for every $i \geq 3$, $\autom'_i \coloneqq \seq{\Sigma \times Q_1 \times \dots \times Q_{i-1}, Q_i, \delta'_i, q_i^0, F_i }$ where: \begin{itemize}
        \item for each $j = 1,\dots, n$, $Q_j$ is the set of states of $\autom_j$;
        \item $q_i^0$ is the initial state of $\autom_i$;
        \item $F_i$ is the set of final states of $\autom_i$;
        \item $\delta_i'$ satisfies $\tau_{(a, q_1, \dots, q_{i-1})}^{\autom'_i} \coloneqq \tau^{\autom_i}_{(a, q_{i-1})}$ for every $(a, q_1, \dots, q_{i-1}) \in \Sigma  \times Q_1 \times \dots \times Q_{i-1}$. 
    \end{itemize}
\end{itemize}

  From the definition of $\delta'_i$ it follows that if $\autom_i$ is
  a pure-reset (resp. reset, permutation, permutation-reset) automaton,
  then also $\autom'_i$ is a pure-reset  (resp. reset, permutation,
  permutation-reset) automaton.

We now have to prove that $\lang(\cascade) = \lang(\cascade')$.
Let $\sigma = \sigma_1\dots \sigma_k \in \Sigma^*$ be a word. The following holds:
\begin{itemize}
    \item[] $\sigma \in \lang(\cascade)$ 
    \item[$\Leftrightarrow$] $\sigma$ induces the run \(
    \tau = (q^0_1, \dots, q^0_n), \dots, (q^k_1, \dots, q^k_n)
    \)  on the $\Sigma$-chain product $\autom_1 \odot \dots \odot \autom_n$, 
    where $(q^k_1, \dots, q^k_n)$ is an accepting state
    \item[$\Leftrightarrow$] $\sigma$ induces 
    the run \(
    \tau = (q^0_1, \dots, q^0_n), \dots, (q^k_1, \dots, q^k_n)
    \) on the cascade product $\autom'_1 \circ \dots \circ \autom'_n$,
    where $(q^k_1, \dots, q^k_n)$ is an accepting state
    \item[$\Leftrightarrow$] $\sigma \in \lang(\cascade')$.
\end{itemize}
This concludes the proof of the Proposition.
\end{proof}

\chainsize*
\begin{proof}
  We start with the first point.  From the definition of $\Sigma$-chain
  product (\cref{def:chain:dfa}), for all $1\le i\le n$, the number of
  states of $\autom_1 \odot \dots \odot \autom_i$ is $\prod_{j=1}^i |Q_i|$
  and, since for each state there must be an outgoing transition labeled
  with every symbol in $\Sigma$, the number of transitions in it (\ie
  $|\autom_1 \odot \dots \odot \autom_i|$) is $|\Sigma| \cdot
  \prod_{j=1}^{i} |Q_j|$. It follows that $|\autom_1 \odot \dots \odot
  \autom_n| = |\Sigma| \cdot \prod_{i=1}^{n} |Q_i|$.  Since we can assume,
  without loss of generality\footnote{
    In fact, if $|Q_i|=1$, the component $\autom_i$ can be removed and the
    alphabet and the transitions of $\autom_{i+1}$ can be trivially
    modified in a way that the language of the resulting $\Sigma$-chain
    does not change.
  },
  that $|Q_i| \ge 2$ for all $1\le i\le n$, we have that $|\autom_1 \odot
  \dots \odot \autom_n| \ge |\Sigma| \cdot 2^n$.

  As for the second statement, by definition of
  $\Sigma$-chain product, we have that the alphabet of $\autom_i$ is
  $\Sigma \times Q_{i-1}$ for every $2\le i\le n$.
  Since, for each symbol in this alphabet, there
  must be an outgoing transition from every state, the
  total number of transitions in $\autom_i$ (\ie $|\autom_i|$) is $|\Sigma|
  \times |Q_{i-1}| \times |Q_i|$. The size of $\autom_1$ is
  $|\Sigma|\times|Q_1|$, since the first component of a $\Sigma$-chain reads only the input alphabet $\Sigma$.  	 By definition, $|\cascade| = \sum_{i=1}^n
  |\autom_i|$, and thus we have that $|\cascade| = (|\Sigma|\cdot|Q_1|)
  + \sum_{i=2}^n (|\Sigma| \cdot |Q_i| \cdot|Q_{i-1}| )$.  Since for every
  $1 \leq i \leq n$ we have that $|Q_i| \leq m$, we obtain $|\cascade|
  \leq \sum_{i=1}^n (|\Sigma| \cdot m \cdot m ) = n \cdot
  |\Sigma| \cdot m^2$.
\end{proof}

\chainsuccintness*
\begin{proof}
    Theorem 5.8 of~\cite{journalversionofstacs} proves that, for any $h>1$,
    the language $\lang_h$ is not definable by any cascade of reset
    automata of height less than $h$, but it can be captured with a cascade
    $\cascade_h$ of two-state reset automata of height $h$.  This implies
    that $\cascade_h$ is the smallest cascade of reset automata recognizing
    $\lang_h$.  By \cref{prop:size:cascades} it follows that $|\cascade_h|
    \ge 2\cdot |\Sigma|\cdot (2^h-1) \in \Omega(2^h\cdot|\Sigma|)$. This
    proves the first item of the claim.

    Consider now the $\Sigma$-chain $\cascade'_h
    = \chain{\autom_1,\dots,\autom_h}$, defined as follows:
    \begin{itemize}
        \item the first automaton is defined as $\autom_1=\seq{\Sigma,
          \set{q_{0}^1, q_1^1}, \delta^1, q_0^1, \set{q_1^1}}$ where for
          every $a \in \Sigma$, the function $\tau_{a}$ is a reset on state
          $q_1^1$;
        \item for every $1\le i\le h$, the automaton $\autom_i$ is defined
          as $\autom_i = \seq{\Sigma\times\set{q_{0}^{i-1}, q_1^{i-1}},
          \set{q_{0}^i, q_1^i}, \delta^i, q_0^1, \set{q_1^i}}$ where for
          every $a \in \Sigma$, the function $\tau_{(a,q_{1}^{i-1})}$ is
          a reset on state $q_1^i$ and the function
          $\tau_{(a,q_{0}^{i-1})}$ is a reset on state $q_0^i$.
    \end{itemize}
    It holds that:
    \begin{enumerate}
        \item $\cascade'_h$ is a two-state $\Sigma$-chain of reset automata of height $h$;
        \item $|\cascade'_h| \leq |m|^2 \cdot h \cdot |\Sigma|$, where
          $m=\max\set{|Q_i| , 1\le i\le h}$; since $m=2$, it follows that
          $|\cascade'_h| \leq 4 \cdot h \cdot |\Sigma| \in O(h \cdot
          |\Sigma|)$;
        \item $\lang(\cascade'_h) = \lang_h$.
    \end{enumerate}
    Point 1 follows directly by definition of $\cascade'_h$. Point
    2 follows from \cref{prop:chain:size}. We now turn to point 3.  We
    prove by induction on $i = 1, \dots,h$ that the $\Sigma$-chain
    $\chain{\autom_1,\dots,\autom_i}$ recognizes the language $\lang_i$
    and that its last component $\autom_i$ does not reach its final state
    $q_1^i$ for every word $\sigma \notin \lang_i$.
    
    \proofcase{{Base step ($i = 1$).}} The reset automaton $\autom_1$ recognizes the language $\Sigma\Sigma^* = \lang_1$. Moreover, $\autom_1$ does not reach its final state $q_1^1$ for every word in $\Sigma^* \setminus \lang_i = \set{\varepsilon}$.

    \proofcase{{Inductive step.}} Assume that (a) the $\Sigma$-chain of reset automata $\chain{\autom_1,\dots,\autom_i}$ recognizes the language $\lang_i$ and (b) $\autom_i$ does not reach its final state $q_1^i$ for every word in $\Sigma^* \setminus \lang_i$.
    By point (b), the automaton $\autom_i$ reaches the state $q_1^i$ only after reading words having a length of at least $i$. In particular the automaton $\autom_i$ transitions from $q_0^i$ to $q_1^i$ when the $i$-th symbol is consumed.
    This means that for every word with a length less or equal to $i$ (\ie for every word in $\Sigma^* \setminus \lang_{i+1})$ the automaton $\autom_{i+1}$ reads symbol of type $(\cdot, q_0^i)$. When the $i+1$-th symbol is consumed the automaton $\autom_{i+1}$ transitions from state $q_0^{i+1}$ to state $q_1^{i+1}$ and stays on $q_1^{i+1}$ for every subsequent symbol.
    This implies that the $\Sigma$-chain of reset automata $\chain{\autom_1,\dots,\autom_i, \autom_{i+1}}$ recognizes the language $\lang_{i+1}$ and $\autom_{i+1}$ does not reach its final state $q_1^{i+1}$ for every word in $\Sigma^* \setminus \lang_{i+1}$.

    This concludes the proof of the Theorem.
\end{proof}

\lemmadirectprodct*
\begin{proof}
  We first prove that $\DP[A] \subseteq \SCH[A]$, for each $A \in
  \set{\PureResets, \Resets, \Permutations, \PermutationResets}$.
  Since the direct product implements the closure under intersection, it
  suffices to prove that $\SCH[A]$ is closed under intersection, for each
  class $A$. 
  Let $\cascade = \chain{\autom_1, \dots, \autom_m}$ and $\cascade'
  = \chain{\automb_1, \dots, \automb_{n}}$ be two $\Sigma$-chains of
  automata.
  Let $Q_m$ be the set of states of $\autom_m$. We define the automaton
  $\automb_1'$ as the automaton obtained from $\automb_1$ by changing its
  alphabet to $\Sigma \times Q_m$ and by defining $\tau^{\automb_1'}_{(a,q)}
  = \tau^{\automb_1}_{a}$, for each $a \in \Sigma$ and each $q \in Q_m$. We define
  $\cascade'' = \chain{\autom_1, \dots,
  \autom_m,\automb'_1,\dots,\automb_{n}}$.
  It is easy to see (for example by Proposition 4.5 of
  \cite{DBLP:conf/stacs/BorelliGMM25}) that $\autom_1 \odot \dots \odot
  \autom_m \odot \automb'_1 \odot \dots \odot \automb_{n}$ is isomorphic to
  $(\autom_1 \odot \dots \odot \autom_m) \times (\automb_1 \odot \dots
  \odot \automb_n)$. It follows that $\lang(\cascade'') = \lang(\cascade)
  \cap \lang(\cascade')$. Moreover, it is easy to see that if $\cascade$
  and $\cascade'$ are $\Sigma$-chains of automata in the class $A \in
  \set{\PureResets, \Resets, \Permutations, \PermutationResets}$, then
  $\cascade''$ is a $\Sigma$-chain of automata in the class $A$.

  We now prove that $\DP[\Resets] \subsetneq \SCH[\Resets]$ and
  $\DP[\PureResets] \subsetneq \SCH[\PureResets]$. The fact that $\lang_1
  \in \SCH[\Resets]$ follows directly by \cref{lemma:k:piecewise}.

   We prove that $\lang_1 \notin \DP[\Resets]$. Let $\DP[\Resets,
   \Sigma]^i$ be the set of languages recognizable by a direct product of
   $i$ reset automata over the alphabet $\Sigma$.
   First we prove that 
     \[\DP[\Resets, \Sigma]^1 \subsetneq \DP[\Resets, \Sigma]^2
     = \DP[\Resets, \Sigma]^3 = \DP[\Resets, \Sigma]^4 = \dots\]
   Second, we show that $\lang \notin \DP[\Resets, \set{a,b}]^2$. This
   implies that $\lang \notin \DP[\Resets]$.
    
    Let $\Sigma = \set{a,b}$. The class $\DP[\Resets, \Sigma]^1$ is the class of languages recognizable by reset automata over $\Sigma$. Recall that, from Theorem 4.7 \cite{DBLP:conf/stacs/BorelliGMM25}, a language $\lang$ is recognized by a reset automaton if and only if $\lang = J \cup (\Sigma^* R I^*)$ for some $I, R \subseteq \Sigma$ such that $I \cap R = \emptyset$ and either $J = I^*$ or $J = \emptyset$.
    It follows that, for $\Sigma=\set{a,b}$, the class $\DP[\Resets,
    \Sigma]^1$ consists of the following 12 languages:
    \begin{itemize}
    \item $\lang^1 = \emptyset$
    \item $\lang^2 = \set{\varepsilon}$
    \item $\lang^3 = \Sigma^*$
    \item $\lang^4 = \Sigma^+$
    \item $\lang^5 = a^*$
    \item $\lang^6 = b^*$
    \item $\lang^7 = \Sigma^* a$
    \item $\lang^8 = \Sigma^* b$
    \item $\lang^9 = \Sigma^* a \cup \set{\varepsilon}$
    \item $\lang^{10} = \Sigma^* b \cup \set{\varepsilon}$
    \item $\lang^{11} = \Sigma^* b a^*$
    \item $\lang^{12} = \Sigma^* a b^*$
\end{itemize}
The class of languages $\DP[\Resets, \Sigma]^2$ is the class $\DP[\Resets, \Sigma]^2 = \set{\lang_1 \cap \lang_2 \suchthat \lang_1, \lang_2 \in \DP[\Resets, \Sigma]^1}$.
It can be algorithmically verified that $\DP[\Resets, \Sigma]^2 = \DP[\Resets, \Sigma]^1 \cup \set{\lang^{13}, \lang^{14}, \lang^{15}, \lang^{16}, \lang^{17}}$ where:
\begin{itemize}
    \item $\lang^{13} = a^+$
    \item $\lang^{14} = b^+$
    \item $\lang^{15} = \Sigma^*a \cap \Sigma^*ba^*$
    \item $\lang^{16} =  \Sigma^*b \cap \Sigma^*ab^*$
    \item $\lang^{17} = \Sigma^*ab^* \cap \Sigma^*ba^*$
\end{itemize}
  Moreover, it can be algorithmically verified that $\DP[\Resets, \Sigma]^2
  = \DP[\Resets, \Sigma]^3$ and so $\DP[\Resets, \Sigma]^2 = \DP[\Resets,
  \Sigma]^i$ for every $i > 2$. Now, since it is easy to see that $\lang$
  is not of any the forms of $\lang_i$ (with $i=1,\dots,17$), it follows
  that $\lang \notin \DP[\Resets, \set{a,b}]^2$ and thus $\lang \not\in
  \DP[\Resets]$.

  We now prove that $\lang_2 \in \SCH[\PureResets]$ but $\lang_2 \notin
  \DP[\PureResets]$. It can be algorithmically verified that
  $\DP[\PureResets] = \DP[\PureResets, \Sigma]^i$ for every $i >1$ and in
  particular $\DP[\PureResets, \Sigma]^1 = \set{\lang^1, \lang^2, \lang^3,
  \lang^4,\lang^7, \lang^8, \lang^9, \lang^{10} }$. Since $\lang_2 \notin
  \DP[\PureResets, \Sigma]^1$, this proves that  $\lang_2 \notin
  \DP[\PureResets]$. Consider now the $\Sigma$-chain
  $\chain{\autom,\automb}$ such that:
\begin{itemize}
    \item $\autom = \set{\set{a,b}, \set{q_0, q_1}, \delta_1, q_0, \set{q_1}}$ where $\tau_a$ is a reset on state $q_1$ and $\tau_b$ is a reset on state $q_0$;
     \item $\automb = \set{\set{a,b}\times \set{q_0,q_1}, \set{s_0, s_1}, \delta_2, s_0, \set{s_1}}$ where $\tau_{(a,q_1)}$ is a reset on state $s_1$ every other symbol induces a reset on state $s_0$.
\end{itemize}
Since it holds that $\lang(\chain{\autom, \automb}) = \Sigma^*aa$, we conclude that $\lang_2 \in \SCH[\PureResets]$.
\end{proof}

\cascadethreetochainthree*
\begin{proof}
Let $\autom_1,\autom_2,\autom_3$ be of the form:
\begin{itemize}
\item $\autom_1 = \seq{\Sigma, Q, \delta_1, q_0, F_1}$;
\item $\autom_2 = \seq{\Sigma \times Q, S, \delta_2, s_0, F_2}$;
\item $\autom_3 = \seq{\Sigma \times Q \times S, T, \delta_3, t_0, F_3}$
\end{itemize}
We define $\autom'_1,\autom'_2,\autom'_3$ as follows:
\begin{itemize}
\item $\autom'_1 \coloneqq  \autom_1$;
\item $\autom'_2 \coloneqq \seq{\Sigma \times Q, Q \times S, \delta'_2, (q_0, s_0), F_1 \times F_2}$ where $\delta'_2$ is  such that $(\star)$ for every $a \in \Sigma, q \in Q, s \in S$ it holds $\delta'_2((q,s),(a,q)) = (q',s')$ if and only if $\delta_1(q,a) = q'$ and $\delta_2(s,(a,q)) = s'$;
\item $\autom'_3 \coloneqq  \seq{\Sigma \times (Q \times S), T, \delta'_3, t_0, F_3}$ where for every $a \in \Sigma, q \in Q, s \in S, t \in T$ we have that $\delta'_3(t,(a,(q,s))) = t'$ if and only if $\delta_3(t,(a,q,s)) = t'$.
\end{itemize}
First we prove that $\cascade'$ is a $\Sigma$-chain of permutation-reset automata. Clearly $\autom'_1$ and $\autom'_3$ are permutation-reset automata since $\autom_1$ and $\autom_3$ are. The next claim proves that the same holds for $\autom'_2$.
\begin{claim}
\label{claim:pr:a2}
The automaton $\autom'_2$ can be chosen to be permutation-reset.
\end{claim}
\begin{claimproof}

For every $(a,q) \in \Sigma \times Q$ we prove that the function $\tau_{(a,q)}^{\autom'_2}$ can be chosen either as a reset or a permutation, while still satisfying the property $(\star)$ required by the definition of $\autom'_2$. We divide in cases depending on the nature of $\tau_{(a,q)}^{\autom_2}$.

\proofcase{{Case $\tau_{(a,q)}^{\autom_2}$ is a reset on $s' \in S$.}} In
  this case, it suffices to set the function  $\tau_{(a,q)}^{\autom'_2}$ as
  a reset on state $(q',s')$ where $q' = \delta_1(q,a)$.  For every $s \in
  S$ we have that $\delta_2(s,(a,q)) = s'$ and $\delta'_2((q,s),(a,q))
  = (q',s')$ and so the property $(\star)$ is satisfied.

\proofcase{{Case $\tau_{(a,q)}^{\autom_2}$ is a permutation.}} In this
  case, we show that there exists a function $\tau_{(a,q)}^{\autom'_2}$
  which satisfies property $(\star)$ and is a permutation function. 
  Let $P \subseteq Q \times S$ be the set of states of $\autom'_2$ of the
  form $(q,s)$ (where $q$ is precisely the state in the symbol $(a,q)$),
  for some $s \in S$. We define $\tau_{(a,q)}^{\autom'_2}$ in this way:
  \begin{itemize}
    \item For each $(q,s) \in P$, we define $\tau_{(a,q)}^{\autom'_2}$ so
      as to satisfy $(\star)$, \ie $\tau_{(a,q)}^{\autom'_2}(q,s)=(q',s')$
      if and only if $\delta_1(q,a) = q'$ and $\delta_2(s,(a,q)) = s'$.
      Since $\tau_{(a,q)}^{\autom_2}$ is a permutation in $\autom_2$, it
      cannot be the case that two states $(q,s_1),(q,s_2) \in P$ are mapped
      by $\tau_{(a,q)}^{\autom'_2}$ into the same state $(q',s_3) \in
      Q \times S$, because that would mean that $\delta_2(s_1,(a,q))
      = \delta_2(s_2,(a,q)) = s_3$, contradicting the fact that
      $\tau_{(a,q)}^{\autom_2}$ is a permutation.  Therefore,
      $\tau_{(a,q)}^{\autom'_2}$ restricted to $P$ is a permutation.
    \item For the remaining states, we define $\tau_{(a,q)}^{\autom'_2}$ to
      be any permutation from $(Q \times S)\setminus P$ to $(Q \times S)
      \setminus \tau_{(a,q)}^{\autom'_2}(P)$, which is guaranteed to exist
      since $|(Q \times S)\setminus P| = |(Q \times S) \setminus
      \tau_{(a,q)}^{\autom'_2}(P)|$.
  \end{itemize}
  By construction, $\tau_{(a,q)}^{\autom'_2}$ satisfies property $(\star)$
  and is a permutation.
\end{claimproof}

  Moreover, it is easy to notice that if $\cascade$ is a cascade of
  pure-reset (resp. permutation) automata, then $\cascade'$ is
  a $\Sigma$-chain of pure-reset (resp. permutation) automata. We now prove
  that $\lang(\cascade) = \lang(\cascade')$.

\begin{claim}
$\lang(\cascade) \subseteq \lang(\cascade')$
\end{claim}
\begin{claimproof}
Let $s = \sigma_1\dots\sigma_n \in \lang(\cascade)$. The word $s$ induces on the cascade product $\autom_1 \circ \autom_2 \circ \autom_3$ the accepting run $\seq{(q_0, s_0, t_0),\dots,(q_n,s_n,t_n)}$. By definition of the run it holds that $\delta((q_i,s_i,t_i),a_{i+1}) = (q_{i+1},s_{i+1},t_{i+1})$ for every $i=0,\dots, n-1$ (where $\delta$ is the transition function of the automaton $\autom_1 \circ \autom_2 \circ \autom_3$).
This means (by definition of cascade product) that:
\begin{enumerate}
    \item $\delta_1(q_i,a_{i+1}) = q_{i+1}$ for every $i=0,\dots, n-1$
    \item $\delta_2(s_i,(a_{i+1},q_i)) = s_{i+1}$ for every $i=0,\dots, n-1$
    \item $\delta_3(t_i,(a_{i+1},q_i,s_i)) = t_{i+1}$ for every $i=0,\dots, n-1$
\end{enumerate}
By point (1) and from the definition of $\autom'_1$ it holds that (4) $\delta'_1(q_i,a_{i+1}) = q_{i+1}$ for every $i=0,\dots, n-1$.
By point (3) and by definition of $\autom'_3$ it holds that (5) $\delta'_3(t_i,(a_{i+1},(q_i,s_i))) = t_{i+1}$ for every $i=0,\dots, n-1$.
From points (1), (2) and from property ($\star$) it holds that (6) $\delta'_2((q_i,s_i),(a_{i+1},q_i)) = (q_{i+1},s_{i+1})$ for every $i=0,\dots, n-1$.
By points (4), (5) and (6) it holds that $s$ induces the run $\seq{(q_0, (q_0,s_0), t_0),\dots,(q_n,(q_n,s_n),t_n)}$ on the $\Sigma$-chain product $\autom'_1 \odot \autom'_2 \odot \autom'_3$. Since $(q_n,(q_n,s_n),t_n) \in F_1 \times (F_1 \times F_2)\times F_3$, the run is accepting and so $s \in \lang(\cascade')$.
\end{claimproof}

\begin{claim}
$\lang(\cascade') \subseteq \lang(\cascade)$
\end{claim}
\begin{claimproof}
Let $s = \sigma_1\dots\sigma_n \in \lang(\cascade')$.
The word $s$ induces on the $\Sigma$-chain product $\autom'_1 \odot \autom'_2 \odot \autom'_3$ the accepting run $\seq{(q_0, (q_0',s_0), t_0),\dots,(q_n,(q'_n,s_n),t_n)}$.
By definition of the run it holds that $\delta'((q_i,(q'_i,s_i),t_i),a_{i+1}) = (q_{i+1},( q'_{i+1},s_{i+1}),t_{i+1})$ for every $i=0,\dots, n-1$ (where $\delta'$ is the transition function of the automaton $\autom'_1 \odot \autom'_2 \odot \autom'_3$).
This means (by definition of $\Sigma$-chain product) that:
\begin{enumerate}
    \item $\delta'_1(q_i,a_{i+1}) = q_{i+1}$ for every $i=0,\dots, n-1$
    \item $\delta'_2((q'_i,s_i),(a_{i+1},q_i)) = (q'_{i+1},s_{i+1})$ for every $i=0,\dots, n-1$
    \item $\delta'_3(t_i,(a_{i+1},(q'_i,s_i))) = t_{i+1}$ for every $i=0,\dots, n-1$
\end{enumerate}
We now prove by induction that for every $i=0,\dots,n$ it holds $q_i = q'_i$.

\proofcase{{Base step ($i = 0$).}}
In this case, it holds $q_0 = q'_0$ by definition of $\autom'_2$.

\proofcase{{Inductive step.}}
Suppose that $q_m = q'_m$ for every $m \leq i$. We show that $q_{i+1} = q'_{i+1}$. By point (2) we have $\delta'_2((q'_i,s_i),(a_{i+1},q_i)) = (q'_{i+1},s_{i+1})$. Since $q'_i = q_i$, we obtain (4) $\delta'_2((q_i,s_i),(a_{i+1},q_i)) = (q'_{i+1},s_{i+1})$. By (4) and by property $(\star)$, it holds (5) $\delta_1(q_i,a_{i+1}) = q'_{i+1}$. However, by point (1) it follows (6) $\delta_1(q_i,a_{i+1}) = q_{i+1}$. Since the automaton $\autom_1$ is deterministic, by points (5) and (6) it must hold $q_{i+1} = q'_{i+1}$. This concludes the proof that (7) for every $i=0,\dots,n$ it holds $q_i = q'_i$.

By points (2), (7) and by the property ($\star$) it follows (9) $\delta_2(s_i,(a_{i+1},q_i)) = s_{i+1}$ for every $i=0,\dots, n-1$.
Moreover, by definition of $\autom'_3$ and by point (3), it holds that (10) $\delta_3(t_i,(a_{i+1},q'_i,s_i)) = \delta_3(t_i,(a_{i+1},q_i,s_i)) = t_{i+1}$ for every $i=0,\dots, n-1$.

By points (6),(9) and (10), the word $s$ induces on the cascade product $\autom_1 \circ \autom_2 \circ \autom_3$ the accepting run $\seq{(q_0, s_0, t_0),\dots,(q_n,s_n,t_n)}$ and so $s \in \lang(\cascade)$.
\end{claimproof}
    
\end{proof}

\chainification*
\begin{proof}
    If $n \le 2$ the claim is trivially verified since any cascade of
    height $2$ is also a $\Sigma$-chain.
    The case $n = 3$ is verified by \Cref{lemma:cascade3:to:chain3}.
    From now on, let $n > 3$.
    For every $i = 1, \dots, n$, let $\autom_i$ be of the form $\autom_i = \seq{\Sigma \times Q_1 \times \dots \times Q_{i-1}, Q_i, \delta_i, q^0_i, F_i}$.
    We define $\autom_1', \dots, \autom_n'$ as follows:
    \begin{itemize}
       \item $\autom'_1 \coloneqq  \autom_1$;
       \item For every $i = 2, \dots, n-1$, $\autom'_i \coloneqq \seq{\Sigma \times Q'_{i-1}, Q'_i , \delta'_i, {q^0}'_i,F_i'}$ where:
       \begin{itemize}
           \item $Q_i' = Q_1 \times \dots \times Q_i$;
           \item ${q^0}'_i = (q^0_1, \dots, q^0_i)$;
           \item $F'_i = F_1 \times \dots \times F_i$;
           \item[$(\star)$] $\delta'_i$ is such that for every $(a, q_1, \dots, q_i) \in \Sigma \times Q_1 \times \dots \times Q_i$ it holds $\delta'_i((q_1, \dots, q_i),(a,(q_1, \dots, q_{i-1}))) = (\tilde{q_1},\dots,\tilde{q_i})$ if and only if: 
           \begin{enumerate}
                \item $\delta_{i-1}' ((q_1, \dots, q_{i-1}), (a, (q_1, \dots, q_{i-2}))) = (\tilde{q_1},\dots,\tilde{q_{i-1}})$ and
               \item $\delta_i (q_i, (a, q_1, \dots, q_{i-1})) = \tilde{q_i}$
           \end{enumerate}
       \end{itemize}
        \item $\autom'_n \coloneqq  \seq{\Sigma \times Q'_{n-1}, Q_n, \delta'_n, q^0_n, F_n}$ where for every $(a, q_1, \dots, q_n) \in \Sigma \times Q_1 \times \dots \times Q_n$ it holds $\delta'_n(q_n,(a,(q_1,\dots , q_{n-1}))) = \tilde{q_n}$ if and only if $\delta_n(q_n,(a,q_1,\dots , q_{n-1})) = \tilde{q_n}$
    \end{itemize}
    \roberto{First notice that, by definition, it holds that $|\autom'_i| \in \mathcal{O}(|\autom_i|)$. This implies that $\cascade' \in \mathcal{O}(|\cascade|)$.}
    
    We begin by  proving that $\cascade'$ is a $\Sigma$-chain of permutation-reset automata. Clearly $\autom'_1$ and $\autom'_n$ are permutation reset automata since $\autom_1$ and $\autom_n$ are. The next claim proves that the same holds for the other automata of the chain.
    \begin{claim}
    For every $i = 2, \dots, n-1$, the automaton $\autom'_i$ can be chosen to be permutation-reset.
    \end{claim}
    \begin{claimproof}
    The proof is similar to the one of \cref{claim:pr:a2}.

    For every $(a,q_1,\dots,q_{i-1}) \in \Sigma \times Q_1 \times \dots \times Q_{i-1}$ we prove that the function $\tau_{(a,(q_1,\dots, q_{i-1}))}^{\autom'_i}$ can be chosen either as a reset or a permutation, while still satisfying the property $(\star)$ required by the definition of $\autom'_i$. We divide in cases depending on the nature of $\tau_{(a,q_1,\dots, q_{i-1})}^{\autom_i}$.

    \proofcase{{Case  $\tau_{(a,q_1,\dots, q_{i-1})}^{\autom_i}$ is a reset on $\tilde{q_i} \in Q_i$.}} In this case, it suffices to set the function  $\tau_{(a,(q_1,\dots, q_{i-1}))}^{\autom'_i}$ as a reset on state $(\tilde{q_1}, \dots, \tilde{q_i})$ where $(\tilde{q_1}, \dots, \tilde{q_{i-1}}) = \delta'_{i-1}((q_1, \dots, q_{i-1}),(a, (q_1, \dots, q_{i-2})))$. 
    For every $q \in Q_i$ we have that $\delta_i(q,(a,q_1,\dots, q_{i-1})) = \tilde{q_i}$ and $\delta'_i((q_1,\dots, q_{i-1},q),(a,(q_1,\dots, q_{i-1}))) = (\tilde{q_1}, \dots, \tilde{q_i})$ and so the property $(\star)$ is satisfied.
    
    \proofcase{{Case  $\tau_{(a,q_1,\dots, q_{i-1})}^{\autom_i}$ is a permutation.}} 
     In this
  case, we show that there exists a function $\tau_{(a,(q_1,\dots, q_{i-1}))}^{\autom'_i}$ which satisfies property $(\star)$ and is a permutation function. 
  Let $P \subseteq Q_1 \times \dots \times Q_i$ be the set of states of $\autom'_i$ of the form $(q_1,\dots,q_i)$ (where $q_1,\dots,q_{i-1}$ are precisely the states in the symbol $(a,(q_1,\dots, q_{i-1}))$),
  for some $q_i \in Q_i$. We define $\tau_{(a,(q_1,\dots, q_{i-1}))}^{\autom'_i}$ in this way:
  \begin{itemize}
    \item For each $(q_1,\dots,q_i) \in P$, we define $\tau_{(a,(q_1,\dots, q_{i-1}))}^{\autom'_i}$ so
      as to satisfy $(\star)$, \ie $\tau_{(a,(q_1,\dots, q_{i-1}))}^{\autom'_i}(q_1,\dots, q_i)=(q_1',\dots,q_i')$
      if and only if $\delta'_{i-1}((q_1,\dots,q_{i-1}),(a,(q_1,\dots,q_{i-2}))) = (q'_1,\dots,q'_{i-1})$ and\\ $\delta_i(q_i,(a,q_1,\dots,q_{i-1})) = q_i'$.
      Since $\tau_{(a,q_1,\dots, q_{i-1})}^{\autom_i}$ is a permutation in $\autom_i$, it
      cannot be the case that two states $(q_1,\dots, q_{i-1},q_{i}^1),(q_1,\dots, q_{i-1},q_{i}^2) \in P$ are mapped
      by $\tau_{(a,(q_1,\dots, q_{i-1}))}^{\autom'_i}$ into the same state $(q'_1,\dots, q'_{i-1},q'_{i}) \in
      Q_1 \times\dots \times Q_i$, because that would mean that $\delta_i(q_i^1,(a,q_1,\dots,q_{i-1}))
      = \delta_i(q_i^2,(a,q_1,\dots,q_{i-1})) = q'_i$, contradicting the fact that
      $\tau_{(a,q_1,\dots, q_{i-1})}^{\autom_i}$ is a permutation.  Therefore,
      $\tau_{(a,(q_1,\dots, q_{i-1}))}^{\autom'_i}$ restricted to $P$ is a permutation.
    \item For the remaining states, we define $\tau_{(a,(q_1,\dots, q_{i-1}))}^{\autom'_i}$ to
      be any permutation from $(Q_1\times \dots \times Q_1)\setminus P$ to $(Q_1 \times \dots \times Q_i)
      \setminus \tau_{(a,(q_1,\dots, q_{i-1}))}^{\autom'_i}(P)$, which is guaranteed to exist
      since $|(Q_1 \times\dots\times Q_i)\setminus P| = |(Q_1 \times \dots \times Q_i) \setminus
      \tau_{(a,(q_1,\dots, q_{i-1}))}^{\autom'_i}(P)|$.
  \end{itemize}
  By construction, $\tau_{(a,(q_1,\dots, q_{i-1}))}^{\autom'_i}$ satisfies property $(\star)$
  and is a permutation.
    
    \end{claimproof}

    Moreover, it is easy to notice that if $\cascade$ is a cascade of pure-reset (resp. permutation) automata, then $\cascade'$ is a $\Sigma$-chain of pure-reset (resp. permutation) automata. We now prove that it holds $\lang(\cascade) = \lang(\cascade')$.

    \begin{claim}
    $\lang(\cascade) \subseteq \lang(\cascade')$
    \end{claim}
    \begin{claimproof}

    Let $s = \sigma_1\dots\sigma_m \in \lang(\cascade)$. The word $s$ induces on the cascade product $\autom_1 \circ \dots \circ \autom_n$ the accepting run $\seq{(q^0_1,\dots, q^0_n),\dots,(q^m_1,\dots, q^m_n)}$. By definition of the run it holds that $\delta((q^j_1,\dots, q^j_n),a_{j+1}) = (q^{j+1}_1,\dots, q^{j+1}_n)$ for every $j=0,\dots, m-1$ (where $\delta$ is the transition function of the automaton $\autom_1 \circ \dots \circ \autom_n$).
    This means (by definition of cascade product) that:
    \begin{center}
        (1) \phantom{aaaa}$ \delta_i(q_i^j,(a_{j+1},q_1^j, \dots, q_{i-1}^j)) = q_i^{j+1} $\phantom{aaaa} for $i=1\dots, n$, and $j = 0,\dots,m-1$
    \end{center}
    By point (1) and by definition of $\autom'_1$ it holds that (2) $\delta'_1(q_1^{j},a_{j+1}) = q_1^{j+1}$ for every $j=0,\dots, m-1$.
    By point (1) and by definition of $\autom'_n$ it holds that (3) $\delta'_n(q_n^{j},(a_{j+1},(q_1^j,\dots,q_{n-1}^j))) = q_n^{j+1}$ for every $j=0,\dots, m-1$.
    We now prove (by induction on $i$) that it holds (4) \[\delta'_i((q^j_1, \dots, q^j_i),(a_{j+1},(q^j_1, \dots, q^j_{i-1}))) = (q^{j+1}_1, \dots, q^{j+1}_i)\] for every $i=2,\dots, n-1$, for every $j = 0, \dots, m-1$.

    \proofcase{{Base step ($i = 2$).}} In this case we have that $\delta_2(q_2^j,(a_{j+1},q_1^j)) = q_2^{j+1}$ for every $j = 0,\dots,m-1$. Using point (2) and the property ($\star$) we obtain that $\delta_2'((q_1^j,q_2^j),(a_{j+1},q_1^j)) = (q_1^{j+1}, q_2^{j+1})$ for every $j = 0, \dots, m-1$.

    \proofcase{{Inductive step.}} Assume that it holds  $\delta'_{i-1}((q^j_1, \dots, q^j_{i-1}),(a_{j+1},(q^j_1, \dots, q^j_{i-2}))) = (q^{j+1}_1, \dots, q^{j+1}_{i-1})$ for every $j = 0, \dots, m-1$. Using point (1) and the property ($\star$) we obtain that\\ $\delta'_i((q^j_1, \dots, q^j_i),(a_{j+1},(q^j_1, \dots, q^j_{i-1}))) = (q^{j+1}_1, \dots, q^{j+1}_i)$ for every $j = 0, \dots, m-1$.
 
    By points (2), (3) and (4) it holds that the word $s$ induces the run: \[\seq{\left(q^0_1, (q^0_1,q^0_2), \dots, (q^0_1,\dots,q^0_{n-1}), q_n^0\right),\dots,\left(q^m_1, (q^m_1,q^m_2), \dots, (q^m_1,\dots,q^m_{n-1}), q^m_n\right)}\] on the $\Sigma$-chain product $\autom'_1 \odot \dots \odot \autom'_n$. Since $\left(q^m_1, (q^m_1,q^m_2), \dots, (q^m_1,\dots,q^m_{n-1}), q^m_n\right) \in F_1 \times (F_1 \times F_2)\times \dots \times (F_1 \times \dots \times F_{n-1}) \times F_n$, the run is accepting and so $s \in \lang(\cascade')$.
    \end{claimproof}

    \begin{claim}
    $\lang(\cascade') \subseteq \lang(\cascade)$
    \end{claim}
    \begin{claimproof}
    Let $s = \sigma_1\dots\sigma_m \in \lang(\cascade')$.
    The word $s$ induces on the $\Sigma$-chain product $\autom'_1 \odot \dots \odot \autom'_n$ the accepting run $\tau$ \[\seq{\left(q^0_1, (q^0_{1,2},q^0_2), \dots, (q^0_{1,n-1},\dots,q^0_{n-2,n-1}, q^0_{n-1}), q_n^0\right),\dots,\left(q^m_1, (q^m_{1,2},q^m_2), \dots, (q^m_{1,n-1},\dots,q^m_{n-2,n-1}, q^m_{n-1}), q_n^m\right)}\].
    By definition of the run it holds that:
    {\small
    \begin{equation}
        \delta'\left(\left(q^j_1, (q^j_{1,2},q^j_2),  \dots, (q^j_{1,n-1},\dots,q^j_{n-2,n-1}, q^j_{n-1}), q_n^j\right), a^{j+1}\right) = \left(q^{j+1}_1, (q^{j+1}_{1,2},q^{j+1}_2), \dots, (q^{j+1}_{1,n-1},\dots,q^{j+1}_{n-2,n-1}, q^{j+1}_{n-1}), q_n^{j+1}\right)
    \end{equation}
    }
    
     for every $j=0,\dots, m-1$ (where $\delta'$ is the transition function of the automaton $\autom'_1 \odot \dots \odot \autom'_n$).
    This means (by definition of $\Sigma$-chain product) that:
    \begin{enumerate}
        \item $\delta'_1(q_1^j,a^{j+1}) = q_1^{j+1}$ for every $j=0,\dots, m-1$
        \item $\delta'_i((q_{1,i}^j,\dots,q_{i-1,i}^j,q_i^j),(a^{j+1},(q_{1,i-1}^j,\dots,q_{i-2,i-1}^j,q_{i-1}^j))) = (q_{1,i}^{j+1},\dots,q_{i-1,i}^{j+1},q_i^{j+1})$\\ for every $i=2,\dots, n-1$, for every $j=0,\dots, m-1$
        \item $\delta'_n(q_n^j,(a^{j+1},(q_{1,n-1}^j,\dots,q_{n-2,n-1}^j,q_{n-1}^j))) = q_n^{j+1}$ for every $j=0,\dots, m-1$
  \end{enumerate}

    For every $i = 1,\dots, n-2$, we have to prove that it holds $q_i^j = q_{i,i+1}^j=\dots = q_{i,n-1}^j$ for every $j = 0, \dots, m$. We proceed by induction on $j$.

    \proofcase{{Base step ($j=0$).}} By definition of the automata $\autom_{i+1}, \dots,\autom_{n-1}$ we have that $q_i^0 = q_{i,i+1}^0 = \dots = q_{i,n-1}^0$ for every $i = 1,\dots, n-2$.

    \proofcase{{Inductive step.}} Suppose that for every $i = 1, \dots, n-2$ it holds that $q_i^{k} = q_{i,i+1}^{k} =\dots = q_{i,n-1}^{k}$ for every $k \leq j$. By point (2) we have that \[\delta'_i((q_{1,i}^j,\dots,q_{i-1,i}^j,q_i^j),(a^{j+1},(q_{1,i-1}^j,\dots,q_{i-2,i-1}^j,q_{i-1}^j))) = (q_{1,i}^{j+1},\dots,q_{i-1,i}^{j+1},q_i^{j+1})\] for every $i=2,\dots, n-1$. Using the inductive hypothesis we obtain (4) \[\delta'_i((q_{1}^j,\dots,q_{i-1}^j,q_i^j),(a^{j+1},(q_{1}^j,\dots,q_{i-2}^j,q_{i-1}^j))) = (q_{1,i}^{j+1},\dots,q_{i-1,i}^{j+1},q_i^{j+1})\] for every $i=2,\dots, n-1$.
    By (4) and by the property ($\star$), it holds (5): 
    \begin{itemize}
        \item[(a)] $ \delta_{i-1}' ((q_1^j, \dots, q_{i-1}^j), (a, (q_1^j, \dots, q_{i-2}^j))) = ({q_{1,i}^{j+1}},\dots,{q_{i-1,i}^{j+1}})$
        \item [(b)] $\delta_{i} (q_{i}^j, (a,q_1^j, \dots, q_{i-1}^j)) = q_i^{j+1}$
    \end{itemize}

    for every $i=2,\dots, n-1$. We now prove (by induction on $i$) that it holds (6)\[
        \delta_{i-1}' ((q_1^j, \dots, q_{i-1}^j), (a^{j+1}, (q_1^j, \dots, q_{i-2}^j))) = ({q_{1}^{j+1}},\dots,{q_{i-1}^{j+1}})
    \]
    for every $i = 2, \dots, n-1$.  The base step ($i = 2$) is verified by point (1). For the inductive step, assume that it holds $\delta_{i-1}' ((q_1^j, \dots, q_{i-1}^j), (a^{j+1}, (q_1^j, \dots, q_{i-2}^j))) = ({q_{1}^{j+1}},\dots,{q_{i-1}^{j+1}})$. Using point (5b) and the property ($\star$) we conclude that $\delta'_i((q_{1}^j,\dots,q_{i-1}^j,q_i^j),(a^{j+1},(q_{1}^j,\dots,q_{i-2}^j,q_{i-1}^j))) = (q_{1}^{j+1},\dots,q_{i-1}^{j+1},q_i^{j+1})$.

     By points (5a) and (6) it must hold that \(({q_{1}^{j+1}},\dots,{q_{i-1}^{j+1}}) = ({q_{1,i}^{j+1}},\dots,{q_{i-1,i}^{j+1}})\) for every $i = 2,\dots, n-1$.
     This implies that it holds $q_i^{j+1} = q_{i,i+1}^{j+1} =\dots = q_{i,n-1}^{j+1}$ for every $i = 1,\dots, n-2$.

     This concludes the proof that (7) for every $i = 1,\dots, n-2$ for every $j = 0, \dots, m$ it holds $q_i^j = q_{i,i+1}^j=\dots = q_{i,n-1}^j$.

    Notice that the run $\tau$ can be rewritten as: 
    \[\tau = \seq{\left(q^0_1, (q^0_{1},q^0_2), \dots, (q^0_{1},\dots, q^0_{n-1}), q_n^0\right),\dots,\left(q^m_1, (q^m_{1},q^m_2), \dots, (q^m_{1},\dots, q^m_{n-1}), q_n^m\right)}\]
    Moreover, using point (7), points 1-3 can be rewritten as follows.
    \begin{enumerate}
        \item[(1')] $\delta'_1(q_1^j,a^{j+1}) = q_1^{j+1}$ for every $j=0,\dots, m-1$
        \item[(2')] $\delta'_i((q_{1}^j,\dots,q_i^j),(a^{j+1},(q_{1}^j,\dots,q_{i-1}^j))) = (q_{1}^{j+1},\dots,q_i^{j+1})$ \\for every $i=2,\dots, n-1$, for every $j=0,\dots, m-1$
        \item[(3')] $\delta'_n(q_n^j,(a^{j+1},(q_{1}^j,\dots,q_{n-1}^j))) = q_n^{j+1}$ for every $j=0,\dots, m-1$
  \end{enumerate}

  By point (1') and by definition of $\autom_1'$ it holds (8) $\delta_1(q_1^j,a^{j+1}) = q_1^{j+1}$ for every $j=0,\dots, m-1$.

  By point (2') and the property ($\star$) we obtain (9) $\delta_i(q_i^j, (a^{j+1}, q_1^j, \dots, q_{i-1}^j)) = q^{j+1}_i$ for every $i=2,\dots, n-1$, for every $j=0,\dots, m-1$.

  Finally, by definition of $\autom'_n$ and by point (3'), it holds that (10) $\delta_n(q_n^j,(a^{j+1},q_{1}^j,\dots,q_{n-1}^j)) = q_n^{j+1}$ for every $j=0,\dots, m-1$.

By points (8),(9) and (10), the word $s$ induces on the cascade product $\autom_1 \circ \dots \circ \autom_n$ the accepting run $\seq{(q_1^0, \dots, q_n^0),\dots,(q_1^m,\dots,q_n^m)}$ and so $s \in \lang(\cascade)$.
    \end{claimproof}
\end{proof}

\corocascadeschains*
\begin{proof}
    Let $\lang \in \SCH[\PermutationResets]$. There exists a $\Sigma$-chain
    $\cascade$ of permutation-reset automata such that $\lang
    = \lang(\cascade)$. By \Cref{prop:chain:to:cascades} there exists
    a cascade $\cascade'$ of permutation-reset automata such that
    $\lang(\cascade) = \lang(\cascade') = \lang$. This implies that $\lang
    \in \C[\PermutationResets]$ and so it holds that
    $\SCH[\PermutationResets]\subseteq \C[\PermutationResets]$. The same
    holds also the classes $\Resets$, $\PureResets$, and $\Permutations$.

    Viceversa, let $\lang \in \C[\PermutationResets]$, then there exists
    a cascade $\cascade$ of permutation-reset automata such that $\lang
    = \lang(\cascade)$. By \Cref{teo:chainification} there exists
    a $\Sigma$-chain $\cascade'$ of permutation-reset automata such that
    $\lang(\cascade) = \lang(\cascade') = \lang$. This implies that $\lang
    \in \SCH[\PermutationResets]$ and so it holds that
    $\C[\PermutationResets] = \SCH[\PermutationResets]$. The same holds for
    the classes $\PureResets$ and $\Permutations$.

    Moreover, since we established that $\SCH[\PermutationResets]
    = \C[\PermutationResets]$, and by \cref{lemma:star:free}, a language
    $\lang$ is regular if and only if $\lang \in \C[\PermutationResets]$,
    we obtain that $\lang$ is regular if and only if $\lang \in
    \SCH[\PermutationResets]$.  Finally, since $\C[\Resets]$ is exactly the
    set of star-free languages and since $\SCH[\Resets] \subseteq
    \C[\Resets]$, we have that if $\lang \in \SCH[\Resets]$ then $\lang$ is
    star-free.
\end{proof}

\section{Proofs of \cref{sec:chains:reset}}

\heighthierarchylemma*
\begin{proof}
    The first item follows from the proof of \Cref{teo:chain:succinctness}.
    Consider now the second item. Assume by contradiction that there exists a $\Sigma$-chain $\cascade_h$ of reset automata of height $m <h$ recognizing the language $\lang_h$.
    By \Cref{prop:chain:to:cascades}, there exists a cascade $\cascade'_h$ of reset automata of height $m$ recognizing $\lang_h$. This is in contradiction with Theorem 5.8 of \cite{journalversionofstacs}. We conclude that such a $\Sigma$-chain $\cascade_h$ cannot exists.
\end{proof}

\lemmanoheight*
\begin{proof}
    In the proof of \cref{lemma:directprodct}, we established that  $\DP[\Resets, \set{a,b}]^2 = \DP[\Resets, \set{a,b}]^i$ for every $i > 2$. This means that every language definable by a direct product of $i$ reset automata over $\Sigma = \set{a,b}$ is also definable by a direct product of two reset automata over $\Sigma = \set{a,b}$.
\end{proof}

\lemmaGenWidthCollpase*
\begin{proof}
	Before proceeding with the proof, we rewrite the function $f$.
	\begin{claim}
    \label{functionf}
	It holds that 
	\begin{enumerate*}[label=(\roman*)]
	\item $f(1) = k+1$;
	\item for $i > 1$, $f(i) = (f(i-1)\cdot k)+1$.
	\end{enumerate*} 
	\end{claim}
	\begin{claimproof}
    It holds that $f(1) = \frac{k^2-1}{k-1} = k+1$ and so point (i) is
    verified. Point (ii) is verified as follows:
	$$f(i) = \frac{k^{i+1}-1}{k-1} = \frac{k^i-1}{k-1}k+1
    = \frac{k^{(i-1)+1}-1}{k-1}k+1 = f(i-1) \cdot k + 1 $$
	\end{claimproof}

    We now proceed with the proof of the Lemma.  Let $\lang$ be accepted by
    some $\Sigma$-chain of reset automata $\cascade
    = \chain{\autom_1,\autom_2,\dots,\autom_h}$. We construct a new $\Sigma$-chain of
    reset automata $\cascade'
    = \chain{\autom'_1,\autom'_2,\dots,\autom'_h}$ where $\autom'_i$ has
    $f(i)=\frac{k^{i+1}-1}{k-1}$ states and such that $\lang(\cascade)
    = \lang(\cascade')$. The proof follows the same construction of Lemma
    5.9 of \cite{journalversionofstacs}. 

    We first show that, since each $\autom_i$ is a reset automaton, it can
    reach at most $|\Sigma_i|+1$ states, where $\Sigma_i$ is the input
    alphabet of $\autom_i$.

    \begin{claim}
    \label{claim:f}
      For every $i = 1, \dots, h$, it holds that $\autom_i$ can reach at
      most $|\Sigma_i|+1$ states.
    \end{claim}
    \begin{proof}
      Since $\autom_i$ is a reset automaton, for each $a \in \Sigma_i$,
      $\tau_a$ is either a reset function on a state $q_a$ or the identity
      function. Since symbols that induce the identity function do not
      contribute to the reachability of new states, the worst case is the
      one in which each $a \in \Sigma_i$ induces a reset function on
      a different state $q_a$. In this case, in addition to the initial
      state of $\autom_i$, there are $|\Sigma|$ other states that are
      reachable, and thus $\autom_i$ can reach at most $|\Sigma_i|+1$
      states.
    \end{proof}

    In the automaton obtained by applying the $\Sigma$-chain product to
    $\chain{\autom_1,\autom_2,\dots,\autom_h}$, the states of
    $\autom_{i-1}$ (for $1< i\le h$) that are not reachable in
    $\autom_{i-1}$ cannot contribute for the reachability of a state in
    $\autom_{i}$, that is, any transition $(a,q)$, such that $a \in \Sigma$
    and $q \in Q_{i-1}$ is not reachable in $\autom_{i-1}$, is never taken
    $\autom_{i}$.  This implies that we can modify each component
    $\autom_i$ of the starting $\Sigma$-chain $\cascade$ into a new one
    $\autom'_i$ by restricting it only to symbols $(a,q)$ such that $q$ is
    reachable in $\autom_{i-1}$. By induction and by~\cref{claim:f}, it
    holds that the number of states reachable in $\autom'_i$ is at most
    $(|\Sigma| \cdot f(i-1))+1$ which, by~\cref{functionf}, is exactly
    $f(i)$.

    The new $\Sigma$-chain $\cascade' = \chain{\autom'_1, \dots,
    \autom'_{h}}$ obtained in this way is such that $\lang(\cascade)
    = \lang(\cascade')$ and $\autom'_i$ has $f(i)$ states, for $i
    = 1,\dots, h$.
\end{proof}

\lemmagenpiecewiselanguages*
\begin{proof}
  The proof of the lemma is a proof by induction on $m$. For the base case,
  we use the construction of \cref{claim:lj}. For the inductive step, we
  use the construction of \cref{claim:inductive}.
  \begin{claim}
  \label{claim:lj}
    Let $s \in \Sigma^+$. Let $I$ a subset of $\Sigma$. There exists
    a $\Sigma$-chain $\cascade$ of two-state reset automata such that:
    \begin{itemize}
      \item $\lang(\cascade) = \Sigma^* s I^*$;
      \item $\cascade$ has height $|s|$.
    \end{itemize}
  \end{claim}
  \begin{proof}
        Let $s$ be the word $a_n \dots a_1$ (where $n$ is a positive integer).
        We show that there exists a $\Sigma$-chain of two-states reset automata $\cascade_{a_n,\dots,a_1} = \chain{\autom_n, \dots, \autom_1}$ such that $\lang(\cascade_{a_n,\dots,a_1}) = \Sigma^* a_n \dots a_1 I^*$.
        The automata of the $\Sigma$-chain are defined as follows:
        \begin{itemize}
            \item if $n=1$, then  $\autom_1 = \seq{\Sigma, \set{q_0^1, q_1^1}, \delta^1, q_0^1, \set{q_1^1}}$ where:
                \begin{itemize}
                    \item the function $\tau_{a_1}^{\autom_{1}}$ is a reset on state $q_1^1$;
                    \item for every $a \in I \setminus \set{a_1}$, the function  $\tau_a^{\autom_{1}}$ is the identity function;
                    \item for every $a \in \Sigma \setminus (I \cup \set{a_1})$, the function $\tau_{a_1}^{\autom_{1}}$ is a reset on state $q_0^1$;
                \end{itemize}

            \item if $n>1$, then  $\autom_1 = \seq{\Sigma \times \set{q_0^2, q_1^2}, \set{q_0^1, q_1^1}, \delta^1, q_0^1, \set{q_1^1}}$ where:
                \begin{itemize}
                    \item the function $\tau_{(a_1,q_1^2)}^{\autom_{1}}$ is a reset on state $q_1^1$;
                    \item for every $a \in I \setminus \set{a_1}$, the function  $\tau_{(a,q_1^2)}^{\autom_{1}}$ is the identity function;
                    \item for every $a \in I$, the function  $\tau_{(a,q_0^2)}^{\autom_{1}}$ is the identity function;

                    \item for every $a \in \Sigma \setminus (I \cup \set{a_1})$, the function $\tau_{(a,q_1^2)}^{\autom_{1}}$ is a reset on state $q_0^1$;
                    \item for every $a \not\in I$, the function  $\tau_{(a,q_0^2)}^{\autom_{1}}$ is a reset on state $q_0^1$;
                \end{itemize}

            \item for $i = 2,\dots,n-1$ we define $\autom_i$  as $\autom_i = \seq{\Sigma\times \set{q_0^{i+1}, q_1^{i+1}}, \set{q_0^i, q_1^i}, \delta^i, q_0^i, \set{q_0^i, q_1^i}}$ where:
            \begin{itemize}
                    \item the function $\tau_{(a_i, q_1^{i+1})}^{\autom_{i}}$ is a reset on state $q_1^i$;
                    \item every other symbol induces a reset function on state $q_0^i$;
            \end{itemize}

            \item if $n>1$, then we define $\autom_n$  as $\autom_n = \seq{\Sigma, \set{q_0^n, q_1^n}, \delta^n, q_0^n, \set{q_0^n, q_1^n}}$ where:
            \begin{itemize}
                    \item the function $\tau_{a_n}^{\autom_{n}}$ is a reset on state $q_1^n$;
                    \item every other symbol induces a reset function on state $q_0^n$.
            \end{itemize}
        \end{itemize}

        First, notice that $\cascade_{a_n\dots a_1}$  is a $\Sigma$-chain
        of two-state reset automata of height $n$.  Now, we have to show
        that $\lang(\cascade_{a_n\dots a_1}) = \Sigma^* a_n \dots a_1 I^*$.
        We divide in cases, depending on the value of $n$.

        \proofcase{{Case $n = 1$.}} Using Theorem 4.7 of \cite{DBLP:conf/stacs/BorelliGMM25}, we obtain that $\lang(\cascade_{a_{n=1}}) =\lang(\autom_1) = \Sigma^* a_1 (I \setminus \set{a_1})^* $. In the case $a_1 \not\in I$, we trivially obtain that $\lang(\autom_1) = \Sigma^* a_1 I^*$. In the case $a_1 \in I$, let $w \in \Sigma^* a_1 (I \setminus \set{a_1})^*$, then $w \in \Sigma^* a_1 I^*$; vice-versa let $w \in \Sigma^* a_1 I^*$, there exists a partition of $w = w_1 \cdot w_2$ such $w_1$ contains the last occurrence of $a_1$ and so $w_1 \in \Sigma^* a_1$ and $w_2 \in (I \setminus \set{a_1})^*$ which implies $w \in \Sigma^* a_1 (I \setminus \set{a_1})^*$. This proves that even in the case $a_1 \in I$ it holds $\lang(\autom_1) = \Sigma^* a_1 I^*$.

        \proofcase{{Case $n > 1$.}} First, we prove that for $i = n,\dots,2$ the automaton $\autom_i$ is in state $q_1^i$ if and only if the $\Sigma$-chain has read a word of type $w \cdot a_n \dots a_i$ where $w \in \Sigma^*$. We proceed by induction on $i$.

        \proofcase{{Base step $(i = n)$.}} Using Theorem 4.7 of \cite{DBLP:conf/stacs/BorelliGMM25} and the definition of $\autom_n$ it holds that $\lang_{q_1^n}(\autom_n) = \Sigma^* a_n$.

        \proofcase{{Inductive step.}} For an index $2 < i \leq n$, assume that it holds that the automaton $\autom_i$ is in state $q_1^i$ if and only if the $\Sigma$-chain has read a word of type $w \cdot a_n \dots a_i$ where $w \in \Sigma^*$. By definition of the automaton $\autom_{i-1}$, (i) the function $\tau_{(a_{i-1}, q_1^{i})}^{\autom_{i-1}}$ is a reset on state $q_1^{i-1}$ and (ii) every other symbol induces a reset function on state $q_0^{i-1}$.
        This means that $\autom_{i-1}$ arrives in state $q_1^{i-1}$ if and only if $\autom_{i}$ is in state $q_1^{i}$ and  the input symbol is $a_{i-1}$.
        Since $\autom_{i}$ is in state $q_1^{i}$ if and only if the input is of type $w \cdot a_n \dots a_i$ where $w \in \Sigma^*$, we obtain that  $\autom_{i-1}$ arrives in state $q_1^{i-1}$ if and only if the input is of type $w \cdot a_n \dots a_i \cdot {a_{i-1}}$ where $w \in \Sigma^*$. This proves the inductive step.

        By the above reasoning we obtain that the automaton $\autom_2$ is in state $q_1^2$ if and only if the $\Sigma$-chain has read a word of type $w \cdot a_n \dots a_2$ where $w \in \Sigma^*$. We now prove that the automaton $\autom_1$ is in state $q_1^1$ if and only if  the $\Sigma$-chain has read a word of type $w_1 \cdot a_n \dots a_1 \cdot w_2$ where $w_1 \in \Sigma^*$ and $w_2 \in I^*$.
        The only way of entering the state $q_1^1$ is through the reset symbol $(a_1, q_{1}^2)$. This means that the automaton $\autom_2$ must be in state $q_1^2$ (and this happens only if the the $\Sigma$-chain has read a word of type $w_1 \cdot a_n \dots a_2$ where $w_1 \in \Sigma^*$) and the input symbol must be $a_1$.
        We obtain that $\autom_1$ enters the state $q_1^1$ only when the $\Sigma$-chain reads a word of  $w_1 \cdot a_n \dots a_1$ where $w_1 \in \Sigma^*$.
        After arriving in state $q_1^1$, $\autom_1$ can remain in that state only if it reads symbols inducing the identity function. 
        Again, we divide in cases depending on whether or not it holds $a_1 \in I$.

        \proofcase{Case $a_1 \notin I$.} In this case, by definition of $\autom_1$, the set of symbols inducing the identity function is the set $\set{(a,q) : a \in I, q, \in \set{q_0^2,q_1^2}}$. This means that $\autom_1$, no matter the current state of $\autom_2$, can remain in the state $q_1^1$ reading (for any number of times) any of the symbols in $I$. This implies that $\autom_1$ is in state $q_1^1$ if and only if  the $\Sigma$-chain has read a word of type $w_1 \cdot a_n \dots a_1 \cdot w_2$ where $w_1 \in \Sigma^*$ and $w_2 \in I^*$.

        \proofcase{Case $a_1 \in I$.} In this case, by definition of $\autom_1$, the set of symbols inducing the identity function is the set $\set{(a,q) : a \in (I \setminus \set{a_1}), q, \in \set{q_0^2,q_1^2}} \cup \set{(a_1,q_0^2)}$.
        Consider words of type $w_1 \cdot a_n \dots a_1$ (where $w_1 \in \Sigma^*$), we study which types of words $w_2$ are such that the automaton $\autom_1$ is in state $q_1^1$ if and only if the $\Sigma$-chain has read a word of type $w_1 \cdot a_n \dots a_1 \cdot w_2$ where $w_1 \in \Sigma^*$. We divide in three sub-cases depending on the structure of $w_2$.
        
        \indent \proofcase{Sub-case 1. $a_n\dots a_2$ does not appear in $w_2$.} 
        In this case, after reading the prefix $w_1 \cdot a_n \dots a_1 $, the automaton $\autom_1$ is in state $q_1^1$ and the automaton $\autom_2$ will never transition to state $q_1^2$ (otherwise it means that the word $a_n\dots a_2$ appears in $w_2$ which is in contradiction with the hypothesis of this sub-case). The automaton $\autom_1$, no matter the current state of $\autom_2$, remains in its state $q_1^1$ if it reads  (for any number of times) any of the symbols in $I  \setminus \set{a_1}$. Moreover, $\autom_1$ also stays in state $q_1^1$, if it reads (for any number of times) the symbol $a_1$, provided that $\autom_2$ is in state $q_0^2$. Since $\autom_2$ will always be in state $q_0^2$, we obtain that the automaton $\autom_1$ is in state $q_1^1$ if and only if  the $\Sigma$-chain has read a word of type $w_1 \cdot a_n \dots a_1 \cdot w_2$ where $w_1 \in \Sigma^*$ and $w_2 \in I^*$.

        \indent \proofcase{Sub-case 2. $a_n\dots a_2$ appears in $w_2$ but $a_n\dots a_1$ does not appear in $w_2$.} In this case, after reading the prefix $w_1 \cdot a_n \dots a_1 $, the automaton $\autom_1$ is in state $q_1^1$. At some point, the automaton $\autom_2$ will transition to state $q_1^2$, however, the automaton $\autom_1$ will never read a symbol of type $(a_1,q_1^2)$ (otherwise it means that the word $a_n\dots a_1$ appears in $w_2$ which is in contradiction with the hypothesis of this sub-case). 
        In other terms, every time that $\autom_2$ is in state  $q_1^2$, it is guaranteed that the input symbol will not be $a_1$. Proceeding as in Sub-case 1, we obtain that $\autom_1$ stays in its state $q_1^1$ only if it reads (for any number of times) any symbol in $I \setminus \set{a_1}$, no matter the state of $\autom_2$, or it reads  the symbol $a_1$ and the state of $\autom_2$ is $q_0^2$. We obtain that  $\autom_1$ is in state $q_1^1$ if and only if  the $\Sigma$-chain has read a word of type $w_1 \cdot a_n \dots a_1 \cdot w_2$ where $w_1 \in \Sigma^*$ and $w_2 \in I^*$.

        \indent \proofcase{Sub-case 3. $a_n\dots a_1$ appears in $w_2$.} Let $\sigma$ be a word of type $w_1 \cdot a_n\dots a_1 \cdot w_2$ where $w_1 \in \Sigma^*$. In this case, we can partition $w_2$ as $w_2 = w_2^1 \cdot a_n\dots a_1 \cdot w^2_2$ where $a_n\dots a_1$ does not appear in $w_2^2$. This means that we can rewrite $\sigma = w_1 \cdot a_n\dots a_1 \cdot w_2$ as $\sigma = w'_1 \cdot  a_n\dots a_1 \cdot w'_2$ where $w_1' = w_1\cdot a_n\dots a_1\cdot w_2^1$ and $w'_2 = w_2^2$. It holds that (i) $w'_1 \in \Sigma^*$,  and (ii) $a_n\dots a_1$ does not appear in $w_2'$. This means that we can now apply sub-case 1 or sub-case 2. 
        
        In every case, we have proven that the automaton $\autom_1$ is in state $q_1^1$ if and only if  the $\Sigma$-chain has read a word of type $w_1 \cdot a_n \dots a_1 \cdot w_2$ where $w_1 \in \Sigma^*$ and $w_2 \in I^*$.
        Since all the states of the components $\autom_2, \dots, \autom_n$ are final states and the only final state of $\autom_1$ is the state $q_1^1$, it follows that the language accepted by the $\Sigma$-chain is the set of words that lead the automaton $\autom_1$ to the state $q_1^1$.
        We conclude that it holds $\lang(\cascade_{a_n\dots a_1}) = \Sigma^* a_n \dots a_1 I^*$.
    \end{proof}

    \begin{claim}
    \label{claim:inductive}
        For every $j = 1,\dots, m-1$, let $\cascade$ be a $\Sigma$-chain of two-states reset automata such that $\lang(\cascade) = \lang_{\sigma_1, \dots,\sigma_j}$ and $\cascade$ has height $|\sigma_1|+\dots +|\sigma_{j}|$. Let $\cascade'$ be a $\Sigma$-chain of two-states reset automata of height $|\sigma_{j+1}|$ such that $\lang(\cascade') = \Sigma^* \sigma_{j+1}\Sigma^*$.
        There exists a $\Sigma$-chain of two-states reset automata $\cascade''$ such that:
        \begin{itemize}
            \item $\lang(\cascade'') = \lang_{\sigma_1, \dots,\sigma_{j+1}}$;
            \item $\cascade''$ has height $|\sigma_1|+\dots +|\sigma_{j+1}|$.
        \end{itemize}
    \end{claim}
    \begin{claimproof}
    Let $n$ = $|\sigma_1|+ \dots+|\sigma_j|$ and let $n' = |\sigma_{j+1}|$. Let $\cascade = \chain{\automb_1,\dots, \automb_{n}}$ and let $\cascade'$ be the $\Sigma$-chain of \cref{claim:lj} where $s = \sigma_{j+1}$ and $I = \Sigma$. $\cascade'$ is of the form $\cascade' = \chain{\autom_{n'},\dots, \autom_1}$.
    Assume that:
    \begin{enumerate}
        \item the set of states of $\automb_{n}$ is of the form $\set{s_0^n,s_1^n}$ ;
        \item the only final state of $\automb_{n}$ is the state ${s_1^n}$ ;
        \item for every $i = 0,\dots,n-1$, all the states of $\automb_{i}$ are set as final ones.
    \end{enumerate}
    Points (2) and (3) imply that (4) the language of $\cascade$ is the set of words that lead the automaton $\autom_n$ to the state ${s_1^n}$.

    Let $\autom_{n'}'$ be a new automaton such that:
    \begin{itemize}
        \item the set of states of $\autom_{n'}'$ is the set of states of $\autom_{n'}$ (\ie the set $\set{q_0^{n'}, q_1^{n'}}$) ;
        \item the input alphabet of $\autom_{n'}'$  is $\Sigma \times \set{s_0^n, s_1^n}$ ;
        \item the initial state of $\autom_{n'}'$ is the initial state of $\autom_{n'}$ ;
        \item the set of final states of $\autom_{n'}'$ is the set of final states of $\autom_{n'}$ ;
        \item for every $a \in \Sigma$, the function $\tau^{\autom_{n'}'}_{(a,s_0^n)}$ is a reset function on state $q_0^{n'}$ ;
        \item for every $a \in \Sigma$, the function $\tau^{\autom_{n'}'}_{(a,s_1^n)}$ corresponds to the function  $\tau^{\autom_{n'}}_{a}$.
    \end{itemize}
    Moreover, let $\automb'_n$ be a copy of the automaton $\automb_n$ where  all states are set as final ones.
    We prove that the $\Sigma$-chain $\cascade''=\chain{\automb_1,\dots,\automb_{n-1},\automb'_{n}, \autom'_{n'}, \autom_{n'-1}, \dots, \autom_1}$ recognizes the language $\lang_{\sigma_1,\dots, \sigma_{j+1}} = \Sigma^* \sigma_1 \Sigma^* \dots \Sigma^* \sigma_{j+1} \Sigma^*$.

    By point (4) and by definition of $\lang_{\sigma_1,\dots,\sigma_j}$ it follows that (5) the automaton $\automb'_{n}$ transitions from state  $s_0^{n}$ to state $s_1^{n}$ if and only if the $\Sigma$-chain $\cascade$ reads a word of type $w = w_1 \cdot \sigma_1\cdot w_2 \dots w_{j} \cdot \sigma_{j}$ where $w_i \in \Sigma^*$ for $i=1,\dots,j$.
    Moreover, by point (4) and by definition of $\lang_{\sigma_1,\dots,\sigma_j}$ it follows that (6) after that the $\Sigma$-chain has read the word $w$, the automaton $\automb'_{n}$ remains in the state $s_1^{n}$ for every possible continuation $w_{j+1} \in \Sigma^*$ of the word $w$.

    By the construction of the $\Sigma$-chain $\cascade'$ (\cf \cref{claim:lj}), it holds that (7) for each $i = n'-1,\dots,1$, the automaton $\autom_i$ remains in its initial state $q_0^i$ as long as the automaton $\autom_{i+1}$ is in its initial state $q_0^{i+1}$.  
    By construction, (8) the automaton $\autom'_{n'}$ remains in its initial state $q_0^{n'}$ as long as $\automb'_n$ is in its initial state $s_0^n$.      
    By (7) and (8), it follows by induction on the levels of the $\Sigma$-chain that  
    (9) for every prefix of the input during which $\automb'_n$ is in state $s_0^n$, the sub-$\Sigma$-chain $\chain{\autom'_{n'}, \autom_{n'-1}, \dots, \autom_1}$ remains in its initial configuration, \ie, each component $\autom_i$ is in its initial state $q_0^i$.
    Moreover, (10) once $\automb'_n$ reaches state $s_1^n$, the sub-$\Sigma$-chain $\chain{\autom'_{n'}, \autom_{n'-1}, \dots, \autom_1}$ induces the same sequence of states as the $\Sigma$-chain $\cascade'$. Indeed, by point (6), after this transition the automaton $\autom'_{n'}$ only reads symbols of the form $(a, s_1^n)$; by construction, for every $a \in \Sigma$, the transition function $\tau^{\autom'_{n'}}_{(a,s_1^n)}$ coincides with $\tau^{\autom_{n'}}_{a}$, and all the remaining components $\autom_{n'-1}, \dots, \autom_1$ are identical to those of $\cascade'$. Therefore, the two chains induce the same computation on every continuation of the input.
    Since $\cascade'$ (and so the sub-$\Sigma$-chain $\chain{\autom'_{n'}, \autom_{n'-1}, \dots, \autom_1}$) is in its final state only if it reads a word of type $w_{j+1} \cdot \sigma_{j+1} \cdot w_{j+2}$ where $w_{j+1},w_{j+2} \in \Sigma^*$, by points (5),(9) and (10) it follows that the full $\Sigma$-chain $\cascade''$ is in its final state only if it reads a word of type $w_1 \cdot \sigma_1\cdot w_2 \dots w_{j} \cdot \sigma_{j} \cdot w_{j+1} \cdot \sigma_{j+1} \cdot w_{j+2}$ where $w_i \in \Sigma^*$ for $i=1,\dots,j+2$.
    We conclude that $\lang(\cascade'') = \lang_{\sigma_1,\dots, \sigma_{j+1}} = \Sigma^* \sigma_1 \Sigma^* \dots \Sigma^* \sigma_{j+1} \Sigma^*$ and $\cascade''$ is a $\Sigma$-chain of two-states reset automata $\cascade''$ of height $|\sigma_1|+\dots +|\sigma_{j+1}|$.

    Finally, notice that:
    \begin{enumerate}
        \item the set of states of the last level is of the form $\set{t_0^{n+n'},t_1^{n+n'}}$ ;
        \item the only final state of the last level is the state ${t_1^{n+n'}}$ ;
        \item all the states of all the other levels are final ones.
    \end{enumerate}

    These three points ensure that we can iterate the construction of the claim for every $j = 1,\dots, m$.
    
    \end{claimproof}
    
\end{proof}

\lemmakpiecewise*
\begin{proof}
    Let $\sigma = a_1\dots a_n \in \Sigma^*$ be a word. We define
    $\lang_\sigma \subseteq\Sigma^*$ as the language $\lang_\sigma
    = \Sigma^* a_1 \Sigma^* \dots \Sigma^* a_n \Sigma^*$.  Let $\lang$ be
    a $k$-piecewise-testable language.  By \Cref{def:piecewise}, $\lang$
    is a boolean combination of languages of type $\lang_{\sigma}$ with
    $|\sigma|\leq k$.  By \Cref{lemma:gen:piecewise:languages}, each
    language $\lang_{\sigma}$ is recognized by a $\Sigma$-chain of
    two-states reset automata of height $|\sigma| \leq k$.  It follows that
    $\lang$ is recognized by a boolean combination of $\Sigma$-chains of
    two-states reset automata of height at most $k$.
\end{proof}

\lemmagentriviallanguages*
\begin{proof}
    We prove the claim of the lemma by induction on $m$.
    
    \proofcase{Base step ($m=0$).}
    Let  $\autom$ be the automaton $\autom = \seq{\Sigma, \set{q_0, q_1}, \delta_1, q_0, \set{q_0}}$ such that:
        \begin{itemize}
            \item $\tau_a^{\autom} = id$ for every $a \in I_0$;
            \item the function $\tau_a^{\autom}$ is a reset on state $q_1$ for every $a \in \Sigma\setminus I_0$.
        \end{itemize}  It holds that $\lang(\autom) = I_0^* = \lang^{I_0}$.

    \proofcase{Base step ($m=1$).} Let $\automb$ be the automaton $\automb = \seq{\Sigma \times \set{q_0,q_1}, \set{s_0, s_1}, \delta_2, s_0, \set{s_1}}$ such that:
    \begin{itemize}
        \item $\tau_{(a,q_0)}^{\automb}$ is a reset on state $s_1$ for every $a \in B_1$;
        \item $\tau_{(a,q)}^{\automb} = id$ for every $(a,q) \in I_B = (I_1 \times   \set{q_0,q_1}) \setminus \set{(a',q_0):a' \in B_1}$;
        \item every other symbol induces a reset function on state $s_0$. 
    \end{itemize}
    Let $\autom'$ be the automaton $\autom$ (defined in the case $m=0$) with all states set as final. Let $\cascade$ be the $\Sigma$-chain $\cascade=\chain{\autom', \automb}$.

    Let $\sigma = a_1 \dots a_n$ be a word. Let $\tau = \seq{q^0, \dots, q^{n-1}}$ be the run on the automaton $\autom'$ induced by the word $a_1\dots a_{n-1}$.   The automaton $\automb$ reads the word $\sigma \vert\vert \tau'  = (a_1, q^0), \dots, (a_n, q^{n-1})$.
    Since all the states of $\autom'$ are final, the word $\sigma$ is accepted by the $\Sigma$-chain $\cascade$ if and only if the word $\sigma \vert\vert \tau'$ is accepted by the automaton $\automb$.
    Since $\automb$ is a reset automaton, $\sigma\vert \vert \tau$ is accepted by $\automb$ if and only if $\sigma\vert \vert \tau \in {\Sigma'}^*\cdot R_1\cdot  I_B^*$ where $\Sigma' = \Sigma \times \set{q_0, q_1}$ and $R_1$ is the set of symbols inducing a reset on the only final state $s_1$ (notice that $R_1 = B_1 \times \set{q_0}$). We prove that $\sigma\vert \vert \tau \in {\Sigma'}^*\cdot R_1\cdot  I_B^*$ if and only if $\sigma \in I_0^* \cdot B_1 \cdot I_1^*$.
    We first prove the left-to-right direction. Let $\sigma\vert \vert \tau \in {\Sigma'}^*\cdot R_1\cdot  I_B^*$. This implies that:
    \begin{itemize}
        \item there exists an index $1 \leq h \leq n$ such that $(a_h, q^{h-1}) \in R_1$;
        \item for all $k > h$, it holds $(a_k,q^{k-1}) \in I_B$.
    \end{itemize}
    The automaton $\automb$ takes the transition labeled by $(a_{h}, q^{h-1})$, only if the input symbol is $a_{h}$ and the automaton $\autom'$ is in state $q^{h-1}$. By definition of the set $R_1$, we have that $a_{h} \in B_1$ and that $q^{{h}-1} = q_0$.
    This implies that the word $a_1\dots a_h$ is in the language $\lang_{q_0}(\autom') \cdot B_1 = I_0^* \cdot B_1$. 
    Since $I_B \subseteq I_1 \times \set{q_0,q_1}$, for all $k > h$, it holds $a_k \in I_1$. This implies that $\sigma = a_1\dots a_n \in I_0^* \cdot B_1 \cdot I_1^*$.

    Vice-versa, let  $\sigma = a_1\dots a_n \in I_0^* \cdot B_1 \cdot I_1^*$.
    This means that:
    \begin{itemize}
        \item there exists an index $1 \leq h \leq n$ such that $a_h \in B_1$;
        \item for all $1 \leq k < h$, it holds $a_k \in I_0$;
        \item for all $h < k \leq n$, it holds $a_k \in I_1$.
    \end{itemize}
    Let $h'$ be the largest index $h$ satisfying these three conditions. 
    Since for all $1 \leq k < h'$, it holds $a_k \in I_0$, we obtain (by definition of $\autom'$) that $\delta(q^{k-1},a^k) = q^k = q_0$ for all $1 \leq k < h'$. Since $q^0 = q_0$, we have that $q^k = q_0$ for every $0 \leq k \leq h'-1$. Since it holds that $a_{h'} \in B_1$ and $q^{h'-1} = q_0$ we have that $(a^{h'}, q^{h'-1}) \in R_1$.
    This means that $(q^0, a_1), \dots, (q^{h'-1}, a_{h'}) \in {\Sigma'}^* \cdot R_1$.
    For every $h' < k \leq n$, if $a_k \in I_1 \setminus B_1$ then we obtain that for every $q \in \set{q_0, q_1}$ the pair $(a_k,q)$ is in the set $(I_1 \setminus B_1) \times \set{q_0,q_1} \subseteq I_B$.
    If it holds $a_k \in I_1 \cap B_1$, then $(a_k,q_1) \in I_B$ and $(a_k, q_0) \notin I_B$. However, we prove that $(a_k, q_0)$ cannot appear in $\sigma \vert \vert \tau$ for every $k > h'$.
    Suppose, by contradiction that $(a_k, q^{k-1}=q_0)$ appears in $\sigma \vert \vert \tau$ for some $k > h'$. This means that $\delta(q^{k-2}, a^{k-1}) = q_0 = q^{k-1}$ for some $q^{k-2} \in \set{q_0, q_1}$. By definition of $\autom'$, it must be $q^{k-2}=q_0$ which implies that $a_{k-1} \in I_0$, moreover, it holds that $q^x= q_0$ for every $x = 0,\dots,k-1$   and so $a_1\dots a_{k-1} \in I_0^*$ for some $k > h'$. However, this would imply that $h'' = k$ satisfies the three conditions and it holds $h'' > h'$ which contradicts the maximality of  $h'$.
    We conclude that $\sigma \vert \vert \tau = (q^0, a_1), \dots, (q^{h'-1}, a_{h'}),\dots (q^{n-1},a_n) \in {\Sigma'}^* \cdot R_1 \cdot I_B^*$.

    This proves that $\lang(\chain{\autom',\automb}) = I_0^* \cdot B_1 \cdot I_1^* = \lang^{I_0,I_1}_{B_1}$.
    
    \proofcase{Inductive step.} Suppose that there exists a $\Sigma$-chain $\cascade = \chain{\autom_0,\dots,\autom_{j}}$ of two-states reset automata such that $\lang(\cascade) = I_0^*\cdot B_1 \cdot I_1^*  \dots   B_j \cdot I_j^* = \lang^{I_0,\dots, I_j}_{B_1,\dots,B_j}$. Let $Q_i = \set{q_0^i, q_1^i}$ be the set of states of $\autom_i$ for every $i = 0,\dots, j$. Assume that the only final state of $\autom_j$ is the state $q_1^j$ and that all the states of $\autom_i$ are final ones for every $i = 0,\dots, j-1$.

    We prove that there exists a two-state reset automaton $\autom_{j+1} = \seq{\Sigma \times Q_j,\set{q_0^{j+1}, q_1^{j+1}}, \delta^{j+1},q_0^{j+1},\set{q_1^{j+1}}}$ such that the $\Sigma$-chain of two-states reset automata $\cascade' = \chain{\autom_0,\dots, \autom_{j-1},\autom'_{j}, \autom_{j+1}}$ recognizes the language $\lang(\cascade') = \lang(\cascade) \cdot B_{j+1} \cdot I_{j+1}^* = \lang^{I_0,\dots, I_{j+1}}_{B_1,\dots,B_{j+1}}$ where $\autom'_j$ is the automaton $\autom_j$ with all the states set as final ones. The transition function of $\autom_{j+1}$ is built as follows:
    \begin{itemize}
        \item $\tau_{(a,q_1^j)}^{\autom_{j+1}}$ is a reset on state $q_1^{j+1}$ for every $a \in B_{j+1}$;
        \item $\tau_{(a,q)}^{\autom_{j+1}} = id$ for every $(a,q) \in I_B = (I_{j+1} \times   \set{q_0^{j},q_1^{j}}) \setminus \set{(a',q_1^j):a' \in B_{j+1}}$;
        \item every other symbol induces a reset function on state $q_{0}^{j+1 }$. 
    \end{itemize}

    Notice that language of $\autom'_j$ is a subset of $(\Sigma \times \set{q_0^{j-1}, q_1^{j-1}})^*$. For $i = 0,1$, we denote with $\lang'_{q^j_i}(\autom'_j)$ the subset of $\Sigma^*$ of words that leads the $\Sigma$-chain product $\autom_1\odot\dots \odot \autom_{j-1} \odot \autom'_j$ from the initial state $(q_0^0,\dots, q_{0}^j)$ to a state of the form $(q^0,\dots, q^{j-1}, q^j_1)$ where $(q^0,\dots, q^{j-1}) \in Q_0\times \dots \times Q_{j-1}$. Moreover we say that a word $\sigma \in \Sigma^*$ induces a run $\tau = \seq{s_0,\dots, s_i}$ on the automaton $\autom'_j$ if and only if the the word $\sigma$ induces on the $\Sigma$-chain product $\autom_1\odot\dots \odot \autom_{j-1} \odot \autom'_j$ the run $(t_0^0,\dots, t^{j-1}_0, s_0), \dots , (t^0_i,\dots, t^{j-1}_i, s_i)$ where $(q^0_k,\dots, q^{j-1}_k, s_k) \in Q_0\times \dots \times Q_{j}$ for every $k = 0,\dots, i$. Notice that, since the only final state of $\autom_j$ is the state $q_1^j$ and that all the states of $\autom_i$ are final ones for every $i = 0,\dots, j-1$, it holds that $\lang(\cascade) = \lang'_{q_1^j}(\autom'_j)$.

    The proof proceeds exactly like the case $m=1$.
    
    Let $\sigma = a_1 \dots a_n$ be a word. Let $\tau = \seq{q^0, \dots, q^{n-1}}$ be the run on the automaton $\autom_{j}'$ induced by the word $a_1\dots a_{n-1}$. The automaton $\autom_{j+1}$ reads the word $\sigma \vert\vert \tau'  = (a_1, q^0), \dots, (a_n, q^{n-1})$.
    Since all the states of $\autom_i$ are final for every $i = 0,\dots, j-1$ and also all the states of $\autom'_j$ are final, the word $\sigma$ is accepted by the $\Sigma$-chain $\cascade'$ if and only if the word $\sigma \vert\vert \tau'$ is accepted by the automaton $\autom_{j+1}$.
    Since $\autom_{j+1}$ is a reset automaton, $\sigma\vert \vert \tau$ is accepted by $\autom_{j+1}$ if and only if $\sigma\vert \vert \tau \in {\Sigma'}^*\cdot R_1\cdot  I_B^*$ where $\Sigma' = \Sigma \times \set{q_0^j, q_1^j}$ and $R_1$ is the set of symbols inducing a reset on the only final state $q_1^{j+1}$ (notice that $R_1 = B_{j+1} \times \set{q_1^j}$). We prove that $\sigma\vert \vert \tau \in {\Sigma'}^*\cdot R_1\cdot  I_B^*$ if and only if $\sigma \in \lang(\cascade) \cdot B_{j+1} \cdot I_{j+1}^*$.
    We first prove the left-to-right direction. Let $\sigma\vert \vert \tau \in {\Sigma'}^*\cdot R_1\cdot  I_B^*$. This implies that:
    \begin{itemize}
        \item there exists an index $1 \leq h \leq n$ such that $(a_h, q^{h-1}) \in R_1$;
        \item for all $k > h$, it holds $(a_k,q^{k-1}) \in I_B$.
    \end{itemize}
    The automaton $\autom_{j+1}$ takes the transition labeled by $(a_{h}, q^{h-1})$, only if the input symbol is $a_{h}$ and the automaton $\autom'_j$ is in state $q^{h-1}$. By definition of the set $R_1$, we have that $a_{h} \in B_{j+1}$ and that $q^{{h}-1} = q_1^j$.
    This implies that the word $a_1\dots a_h$ is in the language $\lang'_{q_1^j}(\autom'_j) \cdot B_{j+1} = \lang(\cascade) \cdot B_{j+1}$. 
    Since $I_B \subseteq I_{j+1} \times \set{q_0^j,q_1^j}$, for all $k > h$, it holds $a_k \in I_1$. This implies that $\sigma = a_1\dots a_n \in \lang(\cascade) \cdot B_{j+1} \cdot I_{j+1}^*$.

    Vice-versa, let  $\sigma = a_1\dots a_n \in \lang(\cascade) \cdot B_{j+1} \cdot I_{j+1}^*$.
    This means that:
    \begin{itemize}
        \item there exists an index $1 \leq h \leq n$ such that $a_h \in B_{j+1}$;
        \item $a_1 \dots a_{h-1} \in \lang(\cascade) = \lang'_{q_j^1}(\autom_j')$;
        \item for all $h < k \leq n$, it holds $a_k \in I_{j+1}$.
    \end{itemize}
    Let $h'$ be the largest index $h$ satisfying these three conditions. 
    Since $a_1 \dots a_{h-1} \in \lang'_{q_j^1}(\autom_j')$ we have that $q^{h'-1} = q_1^{j}$. Since it holds that $a_{h'} \in B_{j+1}$ and $q^{h'-1} = q_j^1$ we have that $(a^{h'}, q^{h'-1}) \in R_1$.
    This means that $(q^0, a_1), \dots, (q^{h'-1}, a_{h'}) \in {\Sigma'}^* \cdot R_1$.
    For every $h' < k \leq n$, if $a_k \in I_{j+1} \setminus B_{j+1}$ then we obtain that for every $q \in Q_j$ the pair $(a_k,q)$ is in the set $(I_{j+1} \setminus B_{j+1}) \times Q_j \subseteq I_B$.
    If it holds $a_k \in I_{j+1} \cap B_{j+1}$, then $(a_k,q_1^j) \notin I_B$ and $(a_k, q_0^j) \in I_B$. 
    However, we prove that $(a_k, q_1^j)$ cannot appear in $\sigma \vert \vert \tau$ for every $k > h'$. Suppose, by contradiction that $(a_k, q^{k-1}=q_1^{j})$ appears in $\sigma \vert \vert \tau$ for some $k > h'$. This means that $a_1\dots a_{k-1} \in \lang'_{q_j^1}(\autom_j') = \lang(\cascade)$ for some $k > h'$. However, this would imply that $h'' = k$ satisfies the three conditions and it holds $h'' > h'$ which contradicts the maximality of  $h'$. We conclude that $\sigma \vert \vert \tau = (q^0, a_1), \dots, (q^{h'-1}, a_{h'}),\dots (q^{n-1},a_n) \in {\Sigma'}^* \cdot R_1 \cdot I_B^*$.

    This proves that $\lang(\cascade') = \lang(\cascade) \cdot B_{j+1}\cdot I_{j+1}^* = \lang^{I_0,\dots, I_{j+1}}_{B1, \dots, B_{j+1}}$.
\end{proof}

\lemmartrivial*
\begin{proof}
     Let $\lang$ be a \rtrivial language. By \Cref{def:rtrivial}, $\lang$ is a finite union of languages of the form $\lang_{a_1,\dots,a_n}^{I_0,\dots, I_n} = I_0^* a_1 I_1^*\dots a_n I_n^*$ where $a_i \in \Sigma$ for $1 = 0,\dots, n$, $I_i \subseteq \Sigma \setminus_{a_{i+1}}$ for $i = 0,\dots, n-1$ and $I_n \subseteq \Sigma$.
     By \Cref{lemma:gen:trivial:languages}, each language $\lang_{a_1,\dots,a_n}^{I_0,\dots, I_n}$ is recognized by a $\Sigma$-chain of two-states reset automata of height $n+1$.
     It follows that $\lang$ is recognized by a finite union  of $\Sigma$-chains of two-states reset automata.
\end{proof}

\subsection*{$\Sigma$-Chains of reset automata and \ppLTL}

In this section we present the pastification procedure for $\Sigma$-chains of reset automata.

The section is structured as follows:
\begin{itemize}
    \item We recall the syntax and semantics of Pure Past Linear Temporal Logic.
    \item We describe the pastification procedure.
    \item We prove \cref{pastification}.
\end{itemize}

\paragraph*{Pure Past Linear Temporal Logic}
\label{sub:ppLTL}

In the following, we introduce syntax and semantics of \emph{Pure Past
Linear Temporal Logic} (\ppLTL). Let $\AP$ be a finite set of atomic
propositions. The syntax of \ppLTL formulae over $\AP$ is generated by the
following grammar:
\begin{align}
  \phi \coloneqq p &  \choice \lnot p
                      \choice \phi\lor\phi 
                      \choice \phi\land\phi  & \text{Boolean connectives} \\
                   &  \choice \ltl{Y\phi} 
                      \choice \ltl{Z\phi} 
                      \choice \ltl{\phi S \phi}
                      \choice \ltl{\phi T \phi} & \text{past operators}
\end{align}
where $p\in\AP$. 
The temporal operators are respectively called:
$\ltl{Y}$, \emph{yesterday};
$\ltl{Z}$, \emph{weak yesterday};
$\ltl{S}$, \emph{since};
$\ltl{T}$, \emph{trigger}.

For every formula $\phi$, we define the \emph{tree complexity} (\ie the
size of $\phi$), denoted by $\treesize{\phi}$, inductively as follows:
\begin{enumerate*}[label=(\roman*)]
  \item $\treesize{p} \coloneqq 1$ and $\treesize{\lnot p} \coloneqq 1$,
  \item $\treesize{\otimes \phi} \coloneqq \treesize{\phi} + 1$, for $\otimes \in
    \set{\ltl{Y},\ltl{Z}}$, and 
  \item $\treesize{\phi_1 \oplus \phi_2} \coloneqq \treesize{\phi_1} + \treesize{\phi_2} + 1$ for
    $\oplus \in \set{\ltl{&},\ltl{||},\ltl{S},\ltl{T}}$.
\end{enumerate*}
Moreover, for every formula $\phi$, we define the \emph{dag complexity},
denoted by $\dagsize{\phi}$, as the number of syntactically different
subformulas of $\phi$. Moreover, the dag complexity of a set of formulas
$\set{\phi_1,\dots,\phi_n}$ is given by the sum of the number of
syntactically different subformulas in $\set{\phi_1,\dots,\phi_n}$ and $n$.
While the tree complexity measures the size of the formula, the dag
complexity measures the space needed to represent the formula with the
\emph{sub-term sharing} technique.  For every formula $\phi$, it holds
$\dagsize{\phi} \leq \treesize{\phi}$, while $\dagsize{\phi}
= \treesize{\phi}$ holds if and only if there is no subformula of $\phi$
appearing multiple times. As an example consider the formulas $\phi'
= \ltl{X p \land X p}$ and $\phi'' = {X p \land X q}$. It holds
$\dagsize{\phi'} = 3$, $\treesize{\phi'} = 5$, $\dagsize{\phi''}
= \treesize{\phi''} = 5$ and $\dagsize{\set{\phi', \phi''}}
= 8$.

Formulae of \ppLTL over the alphabet $\AP$ are interpreted over
\emph{finite nonempty words} over the alphabet $\Sigma \coloneqq 2^{\AP}$.
Given a word $\sigma\in(2^\AP)^+$, the \emph{satisfaction} of a formula
$\phi$ by $\sigma$ at position $i \in \set{0,\dots,|\sigma|-1}$, denoted by
$\sigma,i\models\phi$, is defined as follows:
\begin{itemize}
  \item $\sigma,i \models p$                 $\Leftrightarrow$ $p\in\sigma_i$;
  \item $\sigma,i \models \ltl{\neg p}$      $\Leftrightarrow$ $p\not\in\sigma_i$;
  \item $\sigma,i \models \ltl{\phi_1 || \phi_2}$  $\Leftrightarrow$
          $\sigma,i \models \phi_1$ or $\sigma,i \models \phi_2$;
  \item $\sigma,i \models \ltl{\phi_1 && \phi_2}$ $\Leftrightarrow$
          $\sigma,i \models \phi_1$ and $\sigma,i \models \phi_2$;
  \item $\sigma,i \models \ltl{Y\phi}$    $\Leftrightarrow$
          $i > 0$ and $\sigma,i-1\models \phi$;
  \item $\sigma,i \models \ltl{Z\phi}$    $\Leftrightarrow$
          either $i = 0$ or $\sigma,i-1\models \phi$;
  \item $\sigma,i \models \ltl{\phi_1 S \phi_2}$    $\Leftrightarrow$
        there exists $j\le i$ such that $\sigma,j\models\phi_2$,
        and $\sigma,k\models\phi_1$ for all $k$, with $j < k \le i$;
\item $\sigma,i \models \ltl{\phi_1 T \phi_2}$  $\Leftrightarrow$
        either $\sigma,j\models\phi_2$ for all $0\le j \leq i$, or there
        exists $k \le i$ such that $\sigma,k\models\phi_1$ and 
        $\sigma,j\models\phi_2$ for all $j\le k \le i$
\end{itemize}
For every \ppLTL formula $\phi$, we say that a word $\sigma$ satisfies
$\phi$, written $\sigma\models\phi$, if $\phi$ is true at the last position
of $\sigma$, that is, $\sigma,n\models\phi$ with $n=|\sigma|-1$. The
\emph{language} of $\phi$, denoted by $\lang(\phi)$, is the set of words
$\sigma\in\Sigma^+$ such that $\sigma\models\phi$.

\paragraph*{The pastification procedure for $\Sigma$-chains of reset automata}

Given a $\Sigma$-chain $\cascade$ of $n$ reset automata, using \Cref{prop:chain:to:cascades}, we can transform it into an equivalent cascade $\cascade'$ of $n$ reset automata and then, applying Theorem 8.4 of \cite{journalversionofstacs}, we can compile it into an equivalent $\ppLTL$ formula of dag-complexity exponential in $n$.

In this section we show that, given a $\Sigma$-chain of $n$ reset automata, there exists a \ppLTL formula $\pastify{(\cascade)}$ such that the dag-complexity of $\pastify{(\cascade)}$ is quadratic in the number of components $n$. Crucially we avoid the use of \cref{prop:chain:to:cascades} which would cause an unnecessary exponential blow-up.

Let $\AP$ be a set of propositions and let $\cascade = \chain{\autom_1, \dots, \autom_n}$ be a $\Sigma$-chain of reset automata over the alphabet $\Sigma = 2^{\AP}$. 
For every $1 \leq i \leq n$, let $\Sigma_i$, $Q_i$, $q_i^0$ and $F_i$ denote the input alphabet, the set of states, the initial state and the set of final states of the automaton $\autom_i$ respectively.

For every $1 \leq i \leq n$ we define two kinds of formulas:
\begin{itemize}
    \item for $a \in \Sigma_i$, we define the formula $\Phi_a$ such that a word $\sigma=a_1 \dots a_k \in \Sigma^+$ satisfies $\Phi_a$ at time point $j \in \set{1,\dots,k}$ (in symbols, $\sigma,j \models \Phi_a$) if and only if the $j$-th symbol read by the automaton $\autom_i$ is the symbol $a$;
    \item for $q_i \in Q_i$, we define the formula $\Phi_{q_i}(\autom_i)$ which is satisfied by all and only the words  $\sigma \in \Sigma^+$ such that, the run of $\sigma$ on the $\Sigma$-chain $\cascade$ terminates in a state $(s_0,\dots,s_i,\dots,s_n)$ where for each $1 \leq j \leq n$ it holds $s_j \in Q_j$ and the state $s_i$ is the state $q_i$.
\end{itemize}

First we define the formula $\Phi_a$, for every $a \in \Sigma_i$, as follows:

\begin{align}
  \Phi_{a} &\coloneqq \bigwedge_{p \in a} p \land \bigwedge_{p \notin a} \lnot p & \text{if }i =1, a \in \Sigma_1 = 2^{\AP}\\
  \ltl{\Phi_{(a,q_{i-1})}} &\coloneqq 
  \begin{cases}
     \ltl{\Phi_{a} & Y \Phi_{q_{i-1}}(\autom_{i-1})}  & \emph{ if } q_{i-1} \neq q_{i-1}^0 \\
      \ltl{\Phi_{a} & Z \Phi_{q_{i-1}}(\autom_{i-1})} & \emph{ otherwise}
\end{cases} & \text{if } i > 1, (a,q_{i-1}) \in \Sigma_i
\end{align}
Observe that, for every symbol $(a,q_{i-1}) \in \Sigma_i$ processed by the automaton $\autom_i$, the formula $\Phi_{(a,q_{i-1})}(\autom_{i-1})$ depends directly only on the formula $\Phi_{q_{i-1}}$ associated with the immediately preceding level $\autom_{i-1}$.
This constitutes a fundamental difference with respect to the pastification procedure for cascades of reset automata \cite{DBLP:conf/focs/MalerP90,journalversionofstacs}, where a formula defined at the $i$-th level may depend directly on formulas originating from all preceding levels.

Let $R_i \subseteq \Sigma_i$ be the set of letters $a \in \Sigma_i$ such that
$\tau_a$ is a reset function in the automaton $\autom_i$.
For every symbol $a \in R_i$ and for every state $q_i \in Q_i$ of $\autom_i$, we say that \emph{$a$ enters in $q_i$} if and only if $\tau_a$ is a reset on the state $q_i$.
Otherwise, if $\tau_a$ is a reset on a state $q'$, with $q' \neq q_i$, we say
that \emph{$a$ leaves $q_i$}. As in~\cite{DBLP:conf/focs/MalerP90,journalversionofstacs}, we define:
\begin{align}
  in_{q_i} &\coloneqq 
        \bigvee_{\set{a \in R_i \suchthat a \text{ enters in } q_i}} \Phi_a \\
  out_{q_i} &\coloneqq
        \bigvee_{\set{a \in R_i \suchthat a \text{ leaves } q_i}} \Phi_a
\end{align}
For every $q_i \in Q_i$, we define $\Phi_{q_i}(\autom_i)$ as follows \cite{journalversionofstacs}:
\begin{align}
  \Phi_{q_i}(\autom_i) &\coloneqq 
    \ltl{((! out_{q_i}) S (in_{q_i})) \lor \Phi'_{q_i}(\autom_i)} \\
  \ltl{\Phi'_{q_i}(\autom_i)} &\coloneqq
    \begin{cases}
      \bot & \mbox{if } q_i \neq q_i^0 \\
      \ltl{H (! out_{q_i} \land !in_{q_i})} & \mbox{otherwise}
    \end{cases}
\end{align}

We define the \emph{pastification} of the  $\Sigma$-chain $C = \chain{\autom_1, \dots,
\autom_n}$, that
is, the \ppLTL formula equivalent to $\lang(\autom_1 \odot \dots \odot
\autom_n)$, denoted with $\pastify(\cascade)$, as follows:
\[
\pastify(\cascade) \coloneqq \bigwedge_{i = 1\dots n} \bigvee_{q_i \in F_i} \Phi_{q_i}(\autom_i)
\]

\Cref{pastification} states that the pastification of the $\Sigma$-chain $\cascade$ has a dag-complexity quadratic in the size of the $\Sigma$-chain.
We shall prove \Cref{pastification} in the next paragraph.

\paragraph*{Proof of \cref{pastification}}
\pastification*
\begin{proof}
    Let $\AP$ be a set of propositions and let $\cascade = \chain{\autom_1, \dots, \autom_n}$ be a $\Sigma$-chain of reset automata over the alphabet $\Sigma = 2^{\AP}$. For every $1 \leq i \leq n$, let $\Sigma_i$, $Q_i$, $q_i^0$ and $F_i$ denote the input alphabet, the set of states, the initial state and the set of final states of the automaton $\autom_i$ respectively.

    \begin{claim}
    \label{claim:pastification:formula}
        Let $\sigma = \sigma_1\dots\sigma_k  \in \Sigma^+$ and let $1 \leq j \leq k$. For every $1 \leq i \leq n$ it hods that:
        \begin{itemize}
            \item for every $a \in \Sigma_i$,  the $j$-th symbol read by the automaton $\autom_i$ is the symbol $a$ if and only if $\sigma,j \models \Phi_a$;
            \item for every $q_i \in Q_i$, the automaton $\autom_i$ is in state $q_i$ after the $\Sigma$-chain $\cascade$ has processed the first $j$ symbols of $\sigma$ if and only if $\sigma,j \models \Phi_{q_i}(\autom_i)$.
        \end{itemize}
    \end{claim}
    \begin{proof}
        The claim is proven by induction on $i$.
        
        \proofcase{Base step ($i=1$).} The first item is trivially satisfied using the definition of $\Phi_a$ for every $a \in \Sigma_1 = \Sigma = 2^{\AP}$.
        The second item is satisfied since  by Lemma 8.1 of \cite{journalversionofstacs} we have that $\lang(\Phi_{q_1}(\autom_1))  = \lang_{q_1}(\autom_1)\setminus \set{\varepsilon}$ for every $q_1 \in Q_1$.

        \proofcase{Inductive step. } Suppose that the claim holds for $i = m$, we prove that it continues to hold for $i = m+1$.
        Suppose that the $j$-th symbol read by the automaton $\autom_{m+1}$ is the symbol $(a,q_m) \in \Sigma_{m+1} = \Sigma \times Q_m$. 
        We divide in cases depending on the value of $q_m$.

        \proofcase{Case $q_m \neq q_m^0$.}
        In the case the $j$-th symbol read by the automaton $\autom_{m+1}$ is the symbol $(a,q_m)$ if and only if:
        \begin{itemize}
            \item the $j$-th symbol of $\sigma$ is the symbol $a$; \textit{and}
            \item the automaton $\autom_{m}$  is in state $q_m$ after the $\Sigma$-chain $\cascade$ has processed the first $j-1$ symbols of $\sigma$.
        \end{itemize}
        Using the inductive hypothesis, this occurs if and only if:
        \begin{itemize}
            \item $\sigma, j \models \Phi_a$; \textit{and}
            \item $\sigma, j-1 \models \Phi_{q_m}(\autom_m)$.
        \end{itemize}
        This is equivalent to $\sigma, j \models \Phi_a \land \ltl{Y}\Phi_{q_m}(\autom_m)$. 

        \proofcase{Case $q_m = q_m^0$.}
        We obtain that the $j$-th symbol read by the automaton $\autom_{m+1}$ is the symbol $(a,q_m) \in \Sigma_{m+1} =  \Sigma \times Q_m$ if and only if :
        \begin{itemize}
            \item \textit{either} the $j$-th symbol of $\sigma$ is the symbol $a$ \textit{and} the automaton $\autom_{m}$  is in state $q_m$ after the $\Sigma$-chain $\cascade$ has processed the first $j-1$ symbols of $\sigma$;
            \item \textit{or} the $j$-the symbol of $\sigma$ is the symbol $a$ and and $j = 1$.
        \end{itemize}
        This occurs precisely when $\sigma, j \models \Phi_a \land \ltl{Z}\Phi_{q_m}(\autom_m)$.
        
        In every case we obtain that for every symbol $(a,q_m) \in \Sigma_{m+1}  =  \Sigma \times Q_m$,  the $j$-th symbol read by the automaton $\autom_{m+1}$ is the symbol $(a,q_m)$ if and only if $\sigma,j \models \Phi_{(a,q_m)}$.
        This proves the first item; we now proceed to prove the second one.
        
        Let $q_{m+1}\in Q_{m+1}$.
        Again, we divide in cases depending on the value of $q_{m+1}$.

        \proofcase{Case $q_{m+1} \neq q_{m+1}^0$.}
        Since $\autom_{m+1}$ is a reset automaton, its current state is uniquely determined by the most recent reset transition that has been taken.
        We have that the automaton $\autom_{m+1}$ is in state $q_{m+1}$ after the $\Sigma$-chain $\cascade$ has processed the first $j$ symbols of $\sigma$ if and only if:
        \begin{itemize}
            \item there exists an index $1 \leq r \leq j$ such that the $r$-th symbol read by the automaton $\autom_{m+1}$ is the symbol $(a, q_{m}) \in \Sigma_{m+1} = \Sigma \times Q_{m}$ inducing a reset function entering state $q_{m+1}$; \textit{and}
            \item for every $r < h \leq j$ the automaton $\autom_{m+1}$ never reads a symbol inducing a reset function leaving state $q_{m+1}$.
        \end{itemize}
        Since the first item of the claim is proved for $i = m+1$, this is equivalent to:
        \begin{itemize}
            \item there exists an index $1 \leq r \leq j$ such that it holds $\sigma,r \models \Phi_{(a,q_{m})}$ where $(a, q_{m}) \in \Sigma_{m+1}$ is a symbol inducing a reset function entering state $q_{m+1}$; \textit{and}
            \item for every $r < h \leq j$ it holds $\sigma, h \not\models \Phi_{(a',q'_{m})}$ for every symbol $(a', q'_{m}) \in \Sigma_{m+1}$ inducing a reset function leaving state $q_{m+1}$.
        \end{itemize}
        Using the definitions of the formula $in_{q_{m+1}}$ and $out_{q_{m+1}}$, this corresponds to $\sigma,j \models \ltl{(! out_{q_{m+1}}) S (in_{q_{m+1}})}$.

        \proofcase{Case $q_{m+1} = q_{m+1}^0$.}
        In this case, the automaton $\autom_{m+1}$ is in state $q_{m+1}$ after the $\Sigma$-chain $\cascade$ has processed the first $j$ symbols of $\sigma$ if and only if:
        \begin{itemize}
            \item \textit{either} there exists an index  $1 \leq r \leq j$ such that the $r$-th symbol read by the automaton $\autom_{m+1}$ is the symbol $(a, q_{m}) \in \Sigma_{m+1} = \Sigma \times Q_{m}$ inducing a reset function entering state $q_{m+1}$ \textit{and} for every $r < h \leq j$ the automaton $\autom_{m+1}$ never reads a symbol inducing a reset function leaving state $q_{m+1}$;
            \item \textit{or} for every $1 \leq h \leq j$ the automaton $\autom_{m+1}$ never reads a symbol inducing a reset function leaving state $q_{m+1}$.
        \end{itemize}
        This holds if and only if $\sigma,j \models \ltl{(! out_{q_{m+1}}) S (in_{q_{m+1}})} \lor \ltl{H(!out_{q_{m+1}} \land !in_{q_{m+1}})}$.

        In every case, we obtain that  the automaton $\autom_{m+1}$ is in state $q_{m+1}$ after the $\Sigma$-chain $\cascade$ has processed the first $j$ symbols of $\sigma$ if and only if $\sigma,j \models \Phi_{q_{m+1}}(\autom_{m+1}) = \ltl{((! out_{q_{m+1}}) S (in_{q_{m+1}})) \lor \Phi'_{q_{m+1}}(\autom_{m+1})}$.
        This concludes the proof of the claim.
    \end{proof}

    \begin{claim}
    \label{claim:pastification}
        For every $1 \leq i \leq n$, for every $(q_1,\dots,q_i) \in Q_1 \times \dots \times Q_i$ it holds $\lang(\Phi_{q_1}(\autom_1) \land  \dots \land \Phi_{q_i}(\autom_i))  = \lang_{(q_1,\dots,q_i)}(\autom_1 \odot \dots \odot \autom_i)\setminus \set{\varepsilon}$.
    \end{claim}
    \begin{claimproof}
        Let $\sigma = \sigma_1\dots\sigma_k \in \Sigma^+$ be a non empty word.
        Using \cref{claim:pastification:formula}, we have that:
        \begin{itemize}
            \item[] $\sigma \in \lang_{(q_1,\dots,q_i)}(\autom_1 \odot \dots \odot \autom_i)\setminus \set{\varepsilon}$
            \item[$\Leftrightarrow$] the word $\sigma$ induces on the $\Sigma$-chain product $\autom_1 \odot \dots \odot \autom_{i}$ a run of type $(q_1^0,\dots,q_{n}^0),\dots,(q_1^k,\dots,q_{i}^k)$ where $(q_1^k,\dots,q_{i}^k) = (q_1,\dots,q_{i})$
            \item[$\Leftrightarrow$] the automaton $\autom_j$ is in state $q_j$ after the $\Sigma$-chain $\cascade$ has processed the first $k$ symbols of $\sigma$ for every $1 \leq j \leq i$
            \item[$\Leftrightarrow$] $(\sigma,k) \models \Phi_{q_j}(\autom_j)$ for every $1 \leq j \leq i$
            \item[$\Leftrightarrow$] $\sigma \models \Phi_{q_j}(\autom_j)$ for every $1 \leq j \leq i$
            \item[$\Leftrightarrow$] $\sigma \models \Phi_{q_1}(\autom_1) \land \dots \land \Phi_{q_i}(\autom_i)$
            \item[$\Leftrightarrow$] $\sigma \in \lang(\Phi_{q_1}(\autom_1) \land  \dots \land \Phi_{q_i}(\autom_i))$
        \end{itemize}
        
    \end{claimproof}
    
    \begin{claim}
        $\lang(\pastify(\cascade)) = \lang(\cascade) \setminus
        \set{\varepsilon}$
    \end{claim}
    \begin{claimproof}
    By definition of the formula $\pastify{(\cascade)}$ we have:
    \begin{align}
    \pastify(\cascade) & \coloneqq \bigwedge_{i = 1\dots n} \bigvee_{q_i \in F_i} \Phi_{q_i}(\autom_i)\\
    &\equiv \bigvee_{(q_1,\dots,q_n) \in F_1 \times \dots \times F_n} \left(\Phi_{q_1}(\autom_1) \land \dots \land \Phi_{q_n}(\autom_n) \right)
    \end{align}
    Using \cref{claim:pastification} we obtain:
    \begin{align}
    \lang(\pastify(\cascade)) & = \lang\left( \bigvee_{(q_1,\dots,q_n) \in F_1 \times \dots \times F_n} \left(\Phi_{q_1}(\autom_1) \land \dots \land \Phi_{q_n}(\autom_n) \right) \right)\\ 
    &=  \bigcup_{(q_1,\dots,q_n) \in F_1 \times \dots \times F_n} \lang\left((\Phi_{q_1}(\autom_1) \land \dots \land \Phi_{q_n}(\autom_n)\right) \\
    & = \bigcup_{(q_1,\dots,q_n) \in F_1 \times \dots \times F_n} \lang_{(q_1,\dots, q_n)}\left(\autom_1 \odot\dots  \odot \autom_n \right) \setminus \set{\varepsilon}\\
    & = \lang\left(\autom_1 \odot\dots  \odot \autom_n \right) \setminus \set{\varepsilon}
    \end{align}
    This proves the first item of the theorem.
    \end{claimproof}

    \begin{claim}
        $\dagsize{\pastify(\cascade)} \in \mathcal{O}(|\cascade|\cdot ( |\AP| + n)) \in  \mathcal{O}(|\cascade|\cdot( |\AP| + |\cascade|))$
    \end{claim}
    \begin{claimproof} 
    For every $1\leq i \leq n $, let $Q^{\Phi}_i$ be the set $Q_i^{\Phi}=\set{\Phi_{q_i}(\autom_i) \suchthat q_i \in Q_i}$.

    By definition of the formula $\pastify{(\cascade)}$, we have that: \[ \dagsize{\pastify{ (\cascade)}} \in \mathcal{O}\left(\dagsize{\bigcup_{1 \leq i \leq n} Q^{\Phi}_i}\right)\]

    We now compute a bound for the quantity $\dagsize{\bigcup_{1 \leq i \leq n} Q^{\Phi}_i}$.
    
    First notice that in \cite{journalversionofstacs} it is proven that $\dagsize{Q^{\Phi}_1} \in \mathcal{O}(|\autom_1|\times |\AP|)$.

    For a fixed integer $1 < i \leq n$, given the number $\dagsize{\bigcup_{k\leq i}Q_k^{\Phi}}$ we try to bound the quantity $\dagsize{\bigcup_{k\leq i+1}Q_k^{\Phi}}$.
   For every $a \in \Sigma$, the formula $\Phi_a$ is a sub-formula of the set $Q_1^{\Phi}$. Moreover, for every $q_i \in Q_i$, the formula $\Phi_{q_i}(\autom_i)$ is in the set $Q_i^{\Phi}$ by definition.  For every $(a,q) \in \Sigma \times Q_i$, the formula $\Phi_{(a,q_i)}$ is either $\Phi_a \land \ltl{Z}\Phi_q(\autom)$ or $\Phi_a \land \ltl{Y}\Phi_q(\autom)$.  This implies that:
    \begin{align}
      \dagsize{ \bigcup_{k\leq i}Q_k^{\Phi} \cup \set{\Phi_{(a,q_i)} \suchthat (a,q_i) \in \Sigma
        \times Q_i}} & =  \dagsize{\bigcup_{k\leq i}Q^{\Phi}_k} + |\bigcup_{k\leq i}Q_k^{\Phi}| + 3|(\Sigma \times Q_i)|
        \\
      &= \dagsize{\bigcup_{k\leq i}Q^{\Phi}_k} + \sum_{k\leq i}|Q_k^{\Phi}| + 3|(\Sigma \times Q_i)|
    \end{align}
    The first equality, follows from the following observations: 
    \begin{itemize}
      \item the set $\bigcup_{k \leq i}Q_k^{\Phi} \cup \set{\Phi_{(a,q_i)} \suchthat (a,q_i) \in
        \Sigma \times Q_i}$ has cardinality $|\bigcup_{k \leq i}Q_k^{\Phi}| + |\Sigma \times Q_i|$;
      \item the set $\bigcup_{k \leq i}Q_k^{\Phi}$ has $\dagsize{\bigcup_{k \leq i}Q_k^{\Phi}}$ distinct
        sub-formulas;
      \item for every $(a,q_i) \in \Sigma \times Q_i$ (suppose $q_i \neq q_i^0$,
        the case $q_i=q_i^0$ is analogous) there are only two new sub-formulas
        not present in $\bigcup_{k \leq i}Q_k^{\Phi}$, namely $\ltl{Y \Phi_{q_i}(\autom_i)}$ and
        $\Phi_a \land \ltl{Y}\Phi_{q_i}(\autom_i)$.
    \end{itemize}

    Let $q_{i+1} \in Q_{i+1}$. In the case the automaton $\autom_{i+1}$ is a pure-reset automaton, every formula $\Phi_{(a,q_{i})}$ is a sub-formula of $\Phi_{q_{i+1}}(\autom_{i+1})$. In every case, we obtain:
    \begin{align}
      \dagsize{\bigcup_{k \leq i}Q_k^{\Phi} \cup \set{\Phi_{q_{i+1}}(\autom_{i+1})}} &\leq
        \dagsize{\bigcup_{k \leq i}Q_k^{\Phi} \cup \set{\Phi_{q_{i+1}}(\autom_{i+1})} \cup
        \set{\Phi_{(a,q_i)} \suchthat (a,q_i) \in \Sigma \times Q_i}}\\
      & \in \mathcal{O}\left(\dagsize{\bigcup_{k\leq i}Q^{\Phi}_k} + \sum_{k\leq i}|Q_k^{\Phi}| + 3|(\Sigma \times Q_i)|\right)  + \mathcal{O}(|R_{i+1}|)\\
      & \in \mathcal{O}\left(\dagsize{\bigcup_{k\leq i}Q^{\Phi}_k} + \sum_{k\leq i}|Q_k^{\Phi}| + 3|(\Sigma \times Q_i)|\right)  + \mathcal{O}(|\Sigma_{i+1}|)\\
      & \in \mathcal{O}\left(\dagsize{\bigcup_{k\leq i}Q^{\Phi}_k} + \sum_{k\leq i}|Q_k^{\Phi}| + 3|\Sigma \times Q_i)|\right)  + \mathcal{O}(|\Sigma \times Q_i|)\\
      & \in \mathcal{O}\left(\dagsize{\bigcup_{k\leq i}Q^{\Phi}_k} + \sum_{k\leq i}|Q_k^{\Phi}| + |\Sigma \times Q_i|\right)
    \end{align}
Now, we give the bound on the quantity $\dagsize{\bigcup_{k\leq i+1}Q_k^{\Phi}}$:
    \begin{align}
      \dagsize{\bigcup_{k \leq i+1} Q_k^{\Phi} } = \dagsize{\bigcup_{k \leq i} Q_k^{\Phi} \cup Q_{i+1}^{\Phi}} &\leq \dagsize{\bigcup_{k \leq i} Q_k^{\Phi} \cup Q_{i+1}^{\Phi} \cup \set{\Phi_{(a,q_i)} \suchthat (a,q_i) \in \Sigma \times Q_i}} \\
      & \in \mathcal{O}\left(\dagsize{\bigcup_{k\leq i}Q^{\Phi}_k} + \sum_{k\leq i}|Q_k^{\Phi}| + 3|(\Sigma \times Q_i)|\right)  + \mathcal{O}(|\Sigma \times Q_i|) \times |Q_{i+1}|\\
      & \in \mathcal{O}\left(\dagsize{\bigcup_{k\leq i}Q^{\Phi}_k} + \sum_{k\leq i}|Q_k^{\Phi}| + |\autom_{i+1}|\right)\\
    \end{align}
    Since it holds that:
    \begin{itemize}
        \item $\dagsize{Q^{\Phi}_1} \in \mathcal{O}(|\autom_1|\times |\AP|)$
        \item $\dagsize{\bigcup_{k \leq i+1} Q_k^{\Phi} } \in \mathcal{O}\left(\dagsize{\bigcup_{k\leq i}Q^{\Phi}_k} + \sum_{k\leq i}|Q_k^{\Phi}| + |\autom_{i+1}|\right)$
    \end{itemize}
    we obtain that:
  \begin{align}
    \dagsize{\bigcup_{k \leq n} Q_k^{\Phi} } & \in \mathcal{O}(|\autom_1|\cdot|\AP|+\sum_{k \leq n}|\autom_k|+n\sum_{k\leq n}|Q_k|)\\
   & \in \mathcal{O}(|\autom_1|\cdot|\AP|+|\cascade|+n\sum_{k\leq n}|Q_k|)\\
   & \in \mathcal{O}(|\cascade|\cdot|\AP|+|\cascade|+n |\cascade|)\\
& \in \mathcal{O}(|\cascade|\cdot(n + |\AP|)) \\
& \in \mathcal{O}(|\cascade|\cdot(|\cascade| + |\AP|)) \\
  \end{align}

  This implies that: \[ \dagsize{\pastify{ (\cascade)}} \in \mathcal{O}(|\cascade|\cdot(|\AP| + n))  \in \mathcal{O}(|\cascade|\cdot(|\AP| + |\cascade|)) \]
  
    This proves the second item of the theorem.
    \end{claimproof}
    
\end{proof}

\fi

\end{document}